\documentclass[11pt]{article}
\usepackage[utf8]{inputenc}
\usepackage{titlesec}
\titlespacing*{\paragraph}{0pt}{4pt}{3pt}

\usepackage[margin=1in]{geometry} %

\newif\ifQIP
\QIPfalse

\newif\ifFUTURE

\FUTUREfalse

\newif\ifSUBMISSION
\SUBMISSIONfalse

\ifSUBMISSION

\newcommand{\agnote}[1]{}

\newcommand{\abnote}[1]{}

\else

\newcommand{\agnote}[1]{\textcolor{olive}{\small (Alex: #1)}}

\newcommand{\abnote}[1]{\textcolor{blue}{\small (Anne: #1)}}

\fi

\usepackage{amsthm,amssymb} %
\usepackage{fullpage}
\usepackage{booktabs} %
\usepackage{dsfont} %
\usepackage[inline]{enumitem}%
\usepackage{float} %
\usepackage{subcaption}

\usepackage{xcolor}
\definecolor{darkgray}{rgb}{0.25, 0.25, 0.25}
\newcommand{\reg}[1]{{\color{darkgray}#1}}

\ifQIP
\usepackage[margin=.5in]{geometry} %
\else
\usepackage[margin=1in]{geometry} %
\fi
\usepackage{graphicx} %
\usepackage{hyperref} %
\usepackage[utf8]{inputenc} %
\usepackage{mdframed} %
\usepackage{microtype} %
\usepackage{charter} %
\usepackage{soul} %
\usepackage{suffix} %
\usepackage{tikz} %
\usepackage{xcolor} %
\usepackage{xspace} %
\usepackage{amsmath}
\usepackage{cleveref}
\usepackage{multicol}
\usepackage{stackengine}
\usepackage{scalerel}

\newcommand{\myparagraph}[1]{\paragraph{#1}.}

\renewcommand{\hat}[1]{\widehat{#1}} %
\renewcommand{\tilde}[1]{\widetilde{#1}} %

\newcommand{\comp}[1]{\overline{#1}}

\newcommand{\COL}{3-\textsf{COL}}

\newcommand{\calC}{\mathcal{C}}

\newcommand{\Id}{\mathsf{I}}
\newcommand{\T}{\mathsf{T}}
\newcommand{\X}{\mathsf{X}}
\newcommand{\Y}{\mathsf{Y}}
\newcommand{\Z}{\mathsf{Z}}
\newcommand{\Pg}{\mathsf{P}}
\newcommand{\CNOT}{\mathsf{CNOT}}
\newcommand{\CZ}{\mathsf{C}\text{-}\mathsf{Z}}
\newcommand{\Had}{\mathsf{H}}

\newcommand{\real}{\mathbb{R}} %
\let\epsilon=\varepsilon %
\newcommand{\eps}{\epsilon} %

\newcommand{\microspace}{\mspace{.5mu}} %
\newcommand{\ket}[1]{\ensuremath{\lvert\microspace #1
    \microspace\rangle}} %
\newcommand{\bra}[1]{\ensuremath{\langle\microspace #1
    \microspace\rvert}} %
\newcommand{\ketbra}[2]{\lvert #1 \rangle \! \langle #2 \rvert} %
\newcommand{\kb}[1]{\ketbra{#1}{#1}} %

\newcommand{\Paren}[1]{\left(#1\right)}

\newcommand{\abs}[1]{\lvert#1\rvert}

\newcommand{\set}[1]{\{#1\}}

\newcommand{\norm}[1]{\lVert#1\rVert}
\newcommand{\trnorm}[1]{\lVert#1\rVert_{\mathrm{tr}}}

\newcommand{\trNorm}[1]{\left\lVert#1\right\rVert_{\mathrm{tr}}}

\newcommand{\comment}[1]{}

\newcommand{\class}[1]{\textup{#1}\xspace} %

\newcommand{\CLDM}{\class{CLDM}} %

\newcommand{\NP}{\class{NP}} %
\newcommand{\QMA}{\class{QMA}} %
\newcommand{\SimQMA}{\class{SimQMA}} %
\newcommand{\QAM}{\class{QAM}} %
\newcommand{\PP}{\class{PP}} %
\newcommand{\MIP}{\class{MIP}} %
\newcommand{\QNIZK}{\class{QNIZK}} %
\newcommand{\QZK}{\class{QZK}}

\newcommand{\XiQZK}{\Xi\text{-}\class{QZK}}

\newtheorem{theorem}{Theorem}[section] %
\newtheorem*{mainresult}{Main Result}%
\newtheorem{application}{Application}%
\newtheorem{example}{Example} %
\newtheorem{lemma}[theorem]{Lemma} %
\newtheorem{corollary}[theorem]{Corollary} %
\newtheorem{definition}[theorem]{Definition} %
\newtheorem{remark}[theorem]{Remark} %

\DeclareMathOperator{\Tr}{Tr}
\newcommand{\tr}[1]{\Tr\left(#1\right)}

\DeclareMathOperator{\poly}{poly}
\DeclareMathOperator{\negl}{negl}

\hypersetup{pagebackref, colorlinks=true, urlcolor=blue,
  linkcolor=blue, citecolor=magenta}

\usetikzlibrary{calc,positioning}
\usetikzlibrary{decorations.pathreplacing}

\mdfdefinestyle{figstyle}{ %
  linecolor=black!7, %
  backgroundcolor=black!7, %
  innertopmargin=10pt, %
  innerleftmargin=25pt, %
  innerrightmargin=25pt, %
  innerbottommargin=10pt %
}

\definecolor{White}{rgb}{1,1,1} %
\definecolor{Black}{rgb}{0,0,0} %
\definecolor{LightGray}{rgb}{.8,.8,.8} %
\colorlet{ChannelColor}{LightGray} %
\colorlet{ChannelTextColor}{Black} %
\colorlet{ReadoutColor}{White} %

\newcommand{\cR}{\mathcal R}

\newcommand{\ayes}{A_{yes}}
\newcommand{\ano}{A_{no}}

\newcommand{\Enc}{{\mathsf{Enc}}}
\newcommand{\Dec}{{\mathsf{Dec}}}
\newcommand{\MS}{{\mathsf{MS}}}

\def\01{\{0,1\}}
\newcommand{\unary}{\mathsf{unary}}
\newcommand{\hist}{\Phi}

\newcommand{\commit}[2]{\mathsf{comm}^{#2}_{#1}}
\newcommand{\inr}{\in_\$}

\newcommand{\VerOtpX}{V^{\mathsf{otp}}_x}
\newcommand{\WitOtp}{\psi^{\mathsf{otp}}}
\newcommand{\VerSim}{V^{(s)}}
\newcommand{\VerSimX}{V^{(s)}_x}
\newcommand{\WitSim}{\psi^{(s)}}
\newcommand{\Sim}{\mathsf{Sim}_{\VerSim}}
\newcommand{\SimC}{\mathsf{Sim}_{\calC}}
\newcommand{\SimZK}{\mathsf{Sim}}
\newcommand{\SimInt}{\Sim^{\mathsf{Int}}}
\newcommand{\SimSnap}{\Sim^{\mathsf{snap}}}

\usepackage{soul}

 \title{
$\QMA$-hardness of Consistency of Local Density Matrices\\ with Applications to Quantum Zero-Knowledge
}

\date{}

	\author{
    Anne Broadbent~\thanks{University of Ottawa, Department of Mathematics and Statistics \texttt{abroadbe@uottawa.ca}}
\and
Alex B.~Grilo~\thanks{Sorbonne Universit\'e, CNRS, LIP6
\texttt{Alex.Bredariol-Grilo@lip6.fr}}
}

\begin{document}

\maketitle
\ifQIP
\else

\begin{abstract}
  We provide several advances to the understanding of the class of
  Quantum Merlin-Arthur proof systems~($\QMA$), the quantum analogue of \NP{}. Our central contribution is proving a longstanding conjecture that
  the \emph{Consistency of Local Density Matrices} (CLDM) problem is $\QMA$-hard under Karp reductions. The input of CLDM consists of local reduced density matrices on sets of at most~$k$ qubits, and the problem asks if there is an $n$-qubit global quantum state that is locally consistent with all of the $k$-qubit local density matrices. The containment of this problem
  in~$\QMA$ and the $\QMA$-hardness under Turing reductions were proved by Liu [APPROX-RANDOM 2006]. Liu also  conjectured  that CLDM is $\QMA$-hard under Karp reductions, which is desirable for applications, and we finally prove this conjecture.
  We establish this result  using the techniques of {\em simulatable codes} of Grilo, Slofstra, and Yuen [FOCS~2019], simplifying their proofs and tailoring them to the context of $\QMA$.

  In order to develop  applications of CLDM, we propose a framework that we call \emph{locally simulatable proofs} for $\QMA$: this provides %
   $\QMA$ proofs that can be efficiently verified by probing only $k$ qubits and, furthermore, the reduced density matrix of any $k$-qubit subsystem of an accepting
   witness can be computed in polynomial time, independently of the witness. Within this framework, we show several advances in zero-knowledge in the quantum setting.
  We show for the first time  a commit-and-open computational zero-knowledge proof system for all of $\QMA$,
  as a quantum analogue of a ``sigma'' protocol.
  We then define a  \emph{Proof of Quantum Knowledge}, which guarantees that a prover is effectively in possession of a quantum witness in an interactive proof, and show that our zero-knowledge proof system satisfies this definition. Finally, we show that our proof system can
be used to establish that $\QMA$ has a quantum \emph{non-interactive} zero-knowledge proof system  in
  the secret parameter setting.
\end{abstract}
\pagenumbering{gobble}
\clearpage

\setcounter{tocdepth}{2}
\tableofcontents
\clearpage
\pagenumbering{arabic}
\setcounter{page}{1}

\ifQIP

   \label{sec:Intro}

The complexity class $\QMA$ is the quantum analogue of
$\NP$, the class of problems whose solutions can be verified in deterministic
polynomial time. More precisely,  in $\QMA$, an all-powerful prover produces a
quantum proof that is verified by a quantum polynomially-bounded verifier. Given
the probabilistic nature of quantum computation, we require that for true
statements, there exists a quantum proof that makes the verifier accept with high probability (this
is called {\em completeness}),
whereas all ``proofs'' for  false statements are rejected with high
probability (which is called {\em soundness}).

The class $\QMA$ was first defined by Kitaev \cite{KSV02}, who also showed that deciding if a $k$-local
Hamiltonian problem has low-energy states  is $\QMA$-complete. The importance of this result is two-fold: first, from a theoretical computer science perspective, it is  the
 quantum analogue of the Cook-Levin theorem, since it establishes the first non-trivial $\QMA$-complete problem. Secondly, it shows deep links between physics and complexity theory, since the $k$-local Hamiltonian problem is an important  problem in many-body
physics. Thus, a better understanding of \QMA{} would lead
to a better understanding of the power of quantum resources in proof
verification, as we well as the role of {\em quantum entanglement}  in
low-energy states.

Follow-up work strengthened our
understanding of this important complexity class, \emph{e.g.},  by showing that $\QMA$ is contained in the complexity
class $\PP$~\cite{KW00}\footnote{\PP{} is the complexity class of
decision problems that can be solved by probabilistic polynomial-time algorithms
with error strictly smaller than $\frac{1}{2}$.};
that it is possible to reduce completeness and soundness errors without increasing the length of the
witness~\cite{MW05}; understanding the difference between quantum and classical
proofs~\cite{AK07,GKS16,FK18}; the possibility of perfect
completeness~\cite{Aar09b}; and, more recently, the relation of \QMA{} with non-local
games~\cite{NV17,NV18,CGJV19}.

Also, much follow-up work focused on understanding the complete problems for
$\QMA$, mostly by improving the parameters of the  \QMA{}-hard Local Hamiltonian
problem, or making it closer to models more physically
relevant~\cite{KR03,Liu06,KKR06,OT08,CM14,HNN13,BC18}.  In 2014, a
survey of $\QMA$-complete languages~\cite{Boo14} contained a list of 21 general
problems that are known to be $\QMA$-complete\footnote{We remark that these
problems can be clustered as variations of a handful of base problems.}, and since then,
the situation has not drastically changed. This contrasts with the development
of $\NP$, where only a few years after the developments surrounding
\textsf{3--SAT},  Karp published a theory of reducibility, including a list of
21 $\NP$-complete problems~\cite{Kar72}; while 7 years later, a celebrated book
by Garey and Johnson surveyed over 300 $\NP$-complete problems~\cite{GJ90}.\footnote{The first edition of Garey and Johnson~\cite{GJ90} was published in 1979.}

Recently, the role of \QMA{} in quantum cryptography has also been explored. For instance, several results used ideas of the
\QMA{}-completeness of the Local Hamiltonian problem in order to perform verifiable
delegation of quantum computation~\cite{FHM18,Mah18,Gri19}.
Furthermore, another line of work studies {\em zero-knowledge protocols} for
\QMA{}~\cite{BJSW16,BJSW20,VZ19arxiv}; which is extremely relevant, given the fundamental importance  of
 zero-knowledge protocols for \NP{} in cryptography .

Despite the multiple advances in our understanding of $\QMA$ and related techniques,
 a number of fundamental open questions remain. %
In this work, we solve some of these open problems by showing:
\begin{enumerate*}[label=(\roman*)]
\item \QMA{}-hardness of the
Consistency of Local Density Matrix (CLDM) problem under Karp reductions;
\item  ``commit-and-open''
  Zero-Knowledge (ZK) proof of quantum knowledge (PoQ) protocols for \QMA{}; and
\item  a non-interactive zero-knowledge (NIZK)
protocol in the secret parameter scenario.
\end{enumerate*}
Our main technical contribution consists in showing that every problem in \QMA{} admits a verification algorithm whose history state\footnote{See \Cref{eq:history-state-ex}.} is {\em  locally simulatable}, meaning that the reduced density matrices on any small set of qubits is efficiently computable (without knowledge of the quantum witness).
In order to be able to explain our results in more details and appreciate their contribution to a better
understanding of \QMA{}, we first give an overview of these areas and how they relate to these particular problems.

\subsection{Background}
\label{sec:background}
In this section, we discuss the background on the topics that are relevant to this work, summarizing their current state-of-the-art.

\myparagraph{Consistency of Local Density Matrices (CLDM)}
The Consistency of Local Density
matrices problem ($\CLDM$) is as follows: given the classical description of local density matrices
$\rho_1,\ldots,\rho_m$, each on a set of at most $k$ qubits and for a global system of $n$ qubits, is
there a state $\tau$ that is consistent with such reduced states? Liu~\cite{Liu06}
showed that this problem is in \QMA{} and that it is \QMA{}-hard under Turing
reductions, \emph{i.e.}, a deterministic polynomial time algorithm with access to an
oracle that solves $\CLDM$ in unit time can solve any problem in \QMA{}.

We remark that this type of reduction is rather troublesome for \QMA{}, since
the class is not known (nor expected) to be closed under complement, \emph{i.e.}, it is
widely believed that $\QMA{} \ne \class{co}\QMA{}$. If this is indeed the case, then Turing
reductions do not allow a black-box generalization of results regarding the $\CLDM$ problem  to all
problems in~\QMA{}.
This highlights the open problem of establishing the $\QMA{}$-hardness of the $\CLDM$ problem under Karp reductions, \emph{i.e.}, to show  an efficient mapping between yes- and no-instances of any \QMA{} problem  to yes- and
no-instances of $\CLDM$, respectively.

\myparagraph{Zero-Knowledge (ZK) Proofs for $\QMA$}
In an \emph{interactive} proof, a limited party, the \emph{verifier}, receives the
 help of some untrusted powerful party, the \emph{prover}, in order to decide if some
statement is true.
This is  a generalization of a \emph{proof}, where we allow multiple rounds of interaction.
As usual, we require that the  \emph{completeness} and \emph{soundness} properties hold.
For cryptographic applications, the \emph{zero-knowledge (ZK)} property is often
desirable: here, we require that the verifier learn nothing
from the interaction with the prover. This property is formalized by
showing the existence of an efficient \emph{simulator},
which is able to reproduce (\emph{i.e., simulate}) the output of any given
verifier on a \emph{yes} instance (without having direct access to the actual
prover or witness)\footnote{Different definitions of ``reproduce'' result in
different definitions of zero-knowledge protocols. A protocol is \emph{perfect
zero-knowledge} if the distribution of the output of the simulator is exactly the
same as the distribution of output of transcripts of the protocol.
A protocol is \emph{statistical
zero-knowledge} if such distributions are statistically close. Finally, a
protocol is \emph{computational zero-knowledge} if no efficient algorithm can
distinguish both distributions. The convention is that in the absence of such specification,
we are considering the case of computational zero-knowledge.}.

As paradoxical as it sounds, statistical zero-knowledge
interactive proofs  are known to be possible for a host of languages, including the Quadratic
Non-Residuosity, Graph Isomorphism, and Graph Non-Isomorphism
problems~\cite{GMW91,GMR89}; furthermore, all languages that can be proven by
multiple provers ($\MIP$) admits perfect zero-knowledge MIPs~\cite{BGKW88}.
What is more, by introducing computational assumptions, it was shown that all languages that admit an
 interactive proof system also admit a zero-knowledge interactive proof
 system~\cite{BOGG88}. Zero-knowledge interactive proof systems have had a
 profound impact in multiple areas, including cryptography \cite{GMW87} and complexity
 theory \cite{Vad07}.

We now briefly review the zero-knowledge interactive proof system for the
$\NP$-complete problem of Graph 3-colouring (\COL{}). This is a
3-message proof system, and has the additional property that, given a witness,
the prover is efficient. As a first message, the prover \emph{commits} to a
\emph{permutation} of the given 3-colouring (meaning that the prover randomly
permutes the colours to obtain colouring $c$, and produces a list $(v_i,
\mathsf{commit}(c(v_i)))$, using a cryptographic primitive $\mathsf{commit}$
which is a \emph{commitment scheme}). In the second message, the verifier
chooses uniformly at random an edge $\{v_i, v_j\}$ of the graph. The prover
responds with the information that allows the verifier to open the commitments to
the colouring of the vertices of this edge (and nothing more). The verifier \emph{accepts} if and
only if the revealed colours are different. It is easy to see that the protocol is complete and sound. For the zero-knowledge property, the simulator consists in a process that \emph{guesses} which edge will be requested by the verifier and commits to a colouring that satisfies the prover in case this guess is correct. If the guess is incorrect, the technique of \emph{rewinding} allows the simulator to re-initialize the interaction until it is eventually successful. Protocols that follow the \emph{commit-challenge-response} structure of this proof system are called \emph{$\Sigma$-protocols}\footnote{The Greek letter $\Sigma$  visualizes the flow of the protocol.} and, due to their simplicity, they play a very important role, for instance in the celebrated Fiat-Shamir transformation~\cite{FS87}.

The foundations of zero-knowledge in the quantum world were established by
Watrous, who showed a technique called \emph{quantum rewinding}~\cite{Wat09b}
which is used to show the security of some classical zero-knowledge proofs
(including the protocol for \COL{} described above), even against
quantum adversaries. The importance of this technique is that quantum
measurements typically \emph{disturb} the measured state. When we consider
quantum adversaries, such difficulties concern even \emph{classical} proof
systems, due to the rewinding technique that is ubiquitous (see example in the
case of \COL{} above).
Indeed, in the quantum setting, intermediate measurements (such as checking if
the guess is correct) may compromise the success of future executions, since it
is not possible {\em a priori} to ``rewind'' to a previous point in the execution in a
black-box way.

Another dimension where quantum information poses new challenges is in the study
of interactive proof systems for \emph{quantum} languages.
 We point out that Liu~\cite{Liu06}
observed very early on that the $\CLDM$ problem should admit a simple zero-knowledge proof
system following the ``commit-and-open'' approach, as in the \COL{} protocol.
Inspired by this observation,
recent progress has established the existence of zero-knowledge protocols
for all of $\QMA$~\cite{BJSW16,BJSW20}. We note that although the proof system used there is
reminiscent of a $\Sigma$-protocol, 
there are a number of reasons why it is not a ``natural'' quantum analogue of a $\Sigma$ protocol. These include:
\begin{enumerate*}[label=(\roman*)]
\item the use of a coin-flipping protocol, which makes the communication cost higher than 3 messages;
\item the fact that the verifier's message is not a random challenge; and
\item the final answer from the prover is not only the opening of some committed
values.
\end{enumerate*}

Recently, Vidick and Zhang~\cite{VZ19arxiv} showed how to make classical all of the
interaction between the verifier and the prover in~\cite{BJSW16,BJSW20}, by
considering {\em argument systems}\footnote{
Argument systems are a relaxation of proof systems where
the prover is also bounded to  polynomial-time computation, and, for positive instances, the
prover is provided a witness to the \NP{} instance. This model allows much more
efficient protocols which enables it to be used in practice
\cite{BSCG+14,PHGR16,BSCR+19}.
}
instead of
proof systems. In their protocol, they compose the result of
Mahadev~\cite{Mah18} for verifiable delegation of quantum computation by
classical clients with the zero-knowledge protocol of~\cite{BJSW16,BJSW20}.

\myparagraph{Zero-Knowledge Proofs of Knowledge (PoK)} In a zero-knowledge proof,
the verifier becomes convinced of the \emph{existence} of a witness, but this a
priori has no bearing on the prover actually having in her possession  such a
witness. In some circumstances, it is important to guarantee that the prover
actually has a witness. This is the realm of a   \emph{zero-knowledge proof of
knowledge (PoK)}~\cite{GMR89,BG93}.

We give an example to depict this subtlety. Let us consider the task of anonymous
credentials~\cite{Cha83}. In this setting, Alice wants to authenticate into
some online service using her private credentials. In order to protect her credentials, she could engage in a
zero-knowledge proof; this, however would be unsatisfactory, since the verifier in this scenario would be become convinced of the \emph{existence} of accepting credentials, which
does not necessarily translate to Alice actually being in the \emph{possession} of these credentials.
To remedy this situation,  the PoK property establishes an
``if-and-only-if'' situation: if the verifier accepts, then we can guarantee that
the  prover actually \emph{knows} a witness. This notion is formally defined by
requiring the existence of an \emph{extractor},
which is polynomial-time process $K$ that outputs a valid witness when given
oracle access to some prover $P^*$ that makes
the verifier accept with high enough probability.

In the quantum case, there has been some positive results in terms of the
security of classical proofs of knowledge for $\NP$ against quantum
adversaries~\cite{Unr12}. However, in the fully quantum case (that is,
proofs of quantum knowledge for $\QMA$), no scheme has been
proposed.  One of the possible reasons why no such proof of quantum knowledge
protocols was proposed is the lack of a {\em simple} zero-knowledge proof for~\QMA{}.

\myparagraph{Non-Interactive Zero-Knowledge Proofs (NIZK)}
The interactive nature of zero knowledge proof systems (for instance, in $\Sigma$-protocols) means that in some situations they are not applicable since they require the parties to be simultaneously online.
Therefore, another desired property  of such
proof systems is that they are \emph{non-interactive}, which means the whole
protocol consists in a single message from the prover to the
verifier. \emph{Non-interactive zero-knowledge
proofs} (NIZK) is a fundamental construction in modern cryptography and has far-reaching
applications, for instance to cryptocurrencies~\cite{BSCG+14}.

We note that NIZK is known to be impossible in the standard
model~\cite{GO94}, \emph{i.e.}, without extra assumptions, and therefore NIZK has been considered in different
models. In one of the models most relevant in cryptography,
we assume a common reference string (CRS)~\cite{BFM88}, which can be seen as a trusted
party sending a random string to both the prover and the verifier.
In another model,
the trusted party is allowed to send different (but correlated) messages to the
prover and the verifier; this is called the secret parameter
setup~\cite{PS05}. Classically, this model has been shown to be very powerful,
since even its {\em statistical} zero-knowledge version is equivalent to all of
the problems in the complexity class \class{AM} (this is the class that
contains problem that can be verified by public-coin polynomial-time
verifiers). As mentioned in~\cite{PS05}, this model
encompasses another model for NIZK where the prover and the verifier perform an
{\em offline} pre-processing phase (which is independent of the input) and then
the prover provides the ZK proof~\cite{KMO89}. This inclusion holds since the parties could
perform secure multi-party computation to compute the trusted party's operations.

In the quantum case, very little is known on non-interactive zero-knowledge.
Chailloux, Ciocan, Kerenidis and Vadhan studied this problem in a setup where the
message provided by the trusted party can depend on the instance of the
problem~\cite{CCKV08}. Recently, some results also showed that the
Fiat-Shamir transformation for classical protocols is still safe in the quantum
setting, in the quantum random oracle model~\cite{LZ19,DFMS19,Cha19iacr}.
One particular and intriguing open question is the possibility of NIZKs for \QMA{}.

\subsection{Results}
\label{sec:results}
As we have shown so far, the state-of-the-art in the study of $\QMA$ is that the
body of knowledge is still developing, and that there are some specific goals
that, if achieved, would help us better understand $\QMA$ and devise
new protocols for quantum cryptography. Given this context,  we present
now our results in more detail.

Our first result (\Cref{sec:consistency}) is to
show that the CLDM  problem is \QMA{}-hard under
Karp reductions, solving the  $14$-year-old problem proposed  by Liu~\cite{Liu06}.

\begin{mainresult} \label{result:cldm}
  The $\CLDM$ problem is $\QMA$-complete under Karp reductions.
\end{mainresult}
We capture the techniques used in establishing the above into a new
characterization of \QMA{} that provides the best-of-both worlds in terms of two
proof systems for \QMA{} in an abstract way:
we define \SimQMA{} as the complexity class with proof
systems that are \begin{enumerate*}[label=(\roman*)]
\item  locally verifiable (as in the Local Hamiltonian problem), and
\item  every reduced density matrix of the witness can be efficiently computed (as in the $\CLDM$ problem).
\end{enumerate*} This results is the basis for our applications to quantum cryptography:

\begin{application}\label{result:simqma}
  $\SimQMA = \QMA$.
\end{application}

Next, we define a quantum notion of a classical $\Sigma$-protocol, which we call
a $\Xi$-protocol\footnote{Besides being an excellent symbolic reminder of the interaction in a 3-message proof system, $\Xi$ is chosen  as a shorthand for what we might otherwise call a $q\Sigma$ protocol, due to the resemblance with the pronunciation as ``\textbf{csi}gma''.} (please note, both a $\Sigma$ and $\Xi$ protocol is also referred to
throughout as  ``commit-and-open'' protocols.)  Using our characterization
given in \Cref{result:simqma}, we show a $\QMA$-complete language that admits a  $\Xi$-protocol. Taking into account the importance of $\Sigma$ protocols for zero-knowledge proofs, we are able to show (\Cref{sec:xizk-protocol}) a quantum analogue of the celebrated $\cite{GMW91}$ paper:
\begin{application}\label{result:xi-proof}
  All problems in $\QMA$ admit a computational zero-knowledge $\Xi$-proof system.
\end{application}

Then we provide the definition of
Proof of Quantum Knowledge (PoQ).\footnote{This definition is  joint work with Coladangelo, Vidick and Zhang \cite{CVZ19}.} In short, we say that a proof system is a PoQ if there exists a quantum polynomial-time
\emph{extractor}~$K$ that has oracle access to a quantum prover  which makes the verifier accept with high enough probability, and  the extractor is able to output a sufficiently good
witness for a ``\QMA{}-relation''. We note that this definition for a PoQ is not a
straightforward adaptation of the classical definition; this is because \NP{} has many properties such as perfect completeness,
perfect soundness and even that proofs can be copied, that are not expected to
hold in the \QMA{} case. More details are given in~\Cref{sec:PoQ}.
We are then able to show that our $\Xi$
protocol for \QMA{} described in Result~\ref{result:xi-proof} is PoQ.
This is the first proof
of knowledge for~\QMA{}.\footnote{See also independent and concurrent work by Coladangelo, Vidick and Zhang~\cite{CVZ19}.}

\begin{application}\label{result:NIZK}
All problems in $\QMA$ admit a zero-knowledge proof of quantum knowledge  proof system
  and a statistical zero-knowledge proof of quantum knowledge argument system.
\end{application}
We remark that using techniques for post-hoc delegation of quantum
computation~\cite{FHM18}, our PoQ  for \QMA{} may be understood as a
\emph{proof-of-work} for quantum computations, since it could
be used to convince a verifier that the prover has indeed created the {\em history state} of some
pre-defined computation. This is very relevant in the scenario of
testing small-scale quantum computers in the most adversarial model possible: the zero-knowledge property ensures that the verifier learns nothing but the truth of the statement, while the PoQ property means that the prover has indeed prepared a ground state with the given properties.
Comparatively, all currently known protocols either make assumptions on the devices, or certify only the answer of the computation, but not the knowledge of the prover.

Finally, using the techniques of \Cref{result:xi-proof}, we show that every problem
in~\QMA{} has a non-interactive {\em statistical}
zero-knowledge proof in the secret parameter model.
We are even able to
strengthen
our result to the complexity class \class{QAM} (recall that in a~\class{QAM} proof system, the verifier first sends a random string to the
prover, who answers with a quantum proof). Note that $\class{QAM}$ trivially
contains~$\QMA$.

\begin{application}
  All problems in \class{QAM}  have a
  non-interactive statistical zero-knowledge protocol in the secret parameter
  model.
\end{application}

  Note that, as in the classical case~\cite{PS05}, our result also implies a QNIZK protocol
where the prover and the verifier run an offline (classical) pre-processing phase
(independent of the witness) and then the prover sends the quantum ZK proof to
the verifier.
We note also that even though these models are less relevant to the cryptographic applications of
NIZK, we think that our result moves us towards a QNIZK protocol for \QMA{}
in a more standard model.

\subsection{Techniques}
\label{sec:intro-simulatable}
\label{sec:techniques}

The starting point for our results are  {\em locally simulatable codes}, as
defined in \cite{GSY19}. We give now a rough intuition on the properties
of such codes and leave the details to \Cref{sec:simulatable-history}.

First, a quantum error correcting code is
\emph{$s$-simulatable} if there exists an efficient classical algorithm that
outputs the reduced density matrices of codewords on every subset of at most~$s$
qubits.
Importantly, this algorithm is oblivious of the logical state that
is encoded.
We note
that it was already known that  the reduced density matrices
of codewords hide the encoded information, since
quantum error correcting codes can be used in secret sharing
protocols~\cite{CGL99}, and in~\cite{GSY19} they show that there exist
codes such that the classical description of the reduced density matrices of the
codewords can be efficiently computed.
Next, \cite{GSY19} extends the notion of simulatability of {\em
logical operations} on encoded data as follows.
Recalling the theory of fault-tolerant quantum computation, according to which some quantum
error-correcting codes allow computations over \emph{encoded} data by using ``transversal'' gates and
encoded magic states. The definition of $s$-simulatability is extended to require that the simulator
also efficiently computes the reduced
density matrix on at most $s$ qubits of intermediate steps of the
{\em physical} operations that implement a logical gate on the encoded data
(again, by transversal gates and magic states).

\begin{example}\label{ex:simulatability}
 Let us suppose that the encoding map $\Enc$ admits transversal application of the one-qubit gate $G$,\emph{i.e.}, $G^{\otimes N}\Enc(\ket{\psi}) = \Enc(G\ket{\psi})$.
  The simulatability property requires that the density matrices on at most~$s$ qubits of $(G^{\otimes t} \otimes I^{\otimes (N -t)})\Enc(\ket{\psi})$ should be efficiently computed, for every $0 \leq t \leq N$.
\end{example}

In~\cite{GSY19}, the authors show that the concatenated Steane code is a locally simulatable code.
With this tool, in~\cite{GSY19}, it is shown that  every
$\MIP^*$ protocol\footnote{$\MIP^*$ is the set of languages that admit
a classical \emph{multi-prover} interactive proof, where, in addition, the
provers share entanglement} can be made zero-knowledge, thus quantizing the celebrated
result of~\cite{BGKW88}.
Here, we provide an alternative proof for the simulatability of concatenated Steane codes. Our new proof is much simpler than the proof provided in \cite{GSY19}, but it holds for a slightly weaker statement (but which is already sufficient to derive the results in \cite{GSY19}).
Then, for the first time,  we apply the
techniques of simulatable codes from~\cite{GSY19} to~$\QMA$, which enables us to solve many open problems as previously described.

In order to explain our approach to achieving our main result,  we first recall the  quantum
Cook-Levin theorem proved by~Kitaev \cite{KSV02}. In his proof, Kitaev uses the
circuit-to-Hamiltonian construction~\cite{Fey82}, mapping an arbitrary \QMA{}
verification circuit $V = U_T \ldots U_1$ to a local Hamiltonian $H_V$ that enforces
that low energy states are {\em history states} of the computation,
\emph{i.e.}, a \emph{superposition} of the snapshots of $V$ for every timestep $0 \leq
t\leq T$:
\begin{equation}\label{eq:history-state-ex}
    \ket{\hist} =
    \frac{1}{\sqrt{T+1}}\sum_{t=0 \ldots T+1} \ket{t} \otimes
    U_t\ldots U_1\ket{\psi_{init}}.
\end{equation}
In the above, the first register is called the {\em clock} register, and it
encodes the timestep of the computation, while the second register contains the
snapshot of the computation at time $t$, \emph{i.e.}, the quantum gates $U_1, \ldots ,U_t$
applied to the initial state $\ket{\psi_{init}} = \ket{\phi}\ket{0}^{\otimes A}$, that consists of the quantum
witness and auxiliary qubits. The Hamiltonian $H_V$ also
guarantees that $\ket{\psi_{init}}$ has the correct form at $t = 0$, and that
the final step {\em accepts}, \emph{i.e.}, the output qubit is close to~$\ket{1}$.

In \cite{GSY19}, they note that an important obstacle to making a state similar to $\ket{\hist}$\footnote{In \cite{GSY19}, they are simulating history states for $\mathsf{MIP}^*$ computation and therefore they need to deal also with arbitrary Provers' operations.}
locally simulatable is its dependence on the witness state $\ket{\phi}$.
The solution is to consider
a different verification algorithm $V'$ that implements $V$ on {\em
encoded data}, much like in the theory of fault-tolerant quantum computing. In more details, for a fixed locally simulatable code,
$V'$  expects the encoding of the original witness $\Enc(\ket{\phi})$ and then,
with her raw auxiliary states, she creates encodings of auxiliary states $\Enc(\ket{0})$ and
magic states $\Enc(\ket{\MS})$, and then performs the computation $V$ through
transversal gates and magic state
gadgets, and finally decodes the output qubit. This gives rise to a new history state:

\begin{equation}
    \ket{\hist'} =
    \frac{1}{\sqrt{T'+1}}\sum_{t=0 \ldots T'+1} \ket{t} \otimes
    U'_t\ldots U'_1\ket{\psi_{init}'},
\end{equation}
where $\ket{\psi_{init}'} = \Enc(\ket{\phi})\ket{0}^{\otimes A'}$ and
$U'_1,\ldots,U_{T'}$ are the gates of $V'$ described above.
Using the techniques from \cite{GSY19},\footnote{We remark that we also need to fix a small bug in their proof.  The bug fix deals with technicalities regarding $V'$ and~$\ket{\psi'}$ that are beyond the scope of this overview. See \Cref{S:new-verification} and \Cref{R:bug} for more details.}   we can show that from the properties of the locally
simulatable codes, the reduced density matrix on every set of $5$ qubits of $\ket{\hist'}$ can be
efficiently computed.  In this work, we prove that these reduced density
matrices are in fact \QMA{}-hard instances of \CLDM. More concretely, we show
that these reduced density matrices of a hypothetical history state of an
accepting \QMA{}-verification can  always be computed, and there exists a global
state (namely the history state) consistent with these reduced density matrices
if and only if the original \QMA{} verification accepts with overwhelming probability (and
therefore we are in the case of a yes-instance).

Our main result opens up a number of possible applications to cryptographic
settings. However, as we discussed in ~\Cref{sec:results} we face a tradeoff. In $\CLDM$,
we have
the description of the local density matrices, which yields a zero-knowledge $\Xi$ protocol.
On the other hand, the $\QMA$ verification for CLDM is non-local: we need multiple copies of the global state to
perform tomography on the reduced states,\footnote{See \Cref{lem:containment-qma}.} instead of a single
copy that is needed in the Local Hamiltonian problem.

In order to combine these two desired properties in a single object, we describe
a powerful technique that we call \emph{locally simulatable proofs}.
In a
locally simulatable proof system for some problem $A = (\ayes,\ano)$, we require that:
\begin{enumerate*}[label=(\roman*)]
\item the
verification test performed by the verifier acts on at most $k$ out of the $n$
qubits of the proof, and
\item for every $x \in \ayes$,  there exists a
locally simulatable witness $\ket{\psi}$, \emph{i.e.}, a state $\ket{\psi}$ that passes all the
local tests and such that for every $S\subseteq [n]$ with $|S| \leq k$,  it is possible to compute the reduced state of the
$\ket{\psi}$ on $S$ efficiently (without the help of the prover).
\end{enumerate*}
Notice that we have no extra restrictions on $x \in \ano$, since any quantum witness
should make this verifier  reject with high probability.

We then show that  all problems in $\QMA$ admit a locally simulatable proof
system. In order to achieve this, we use the local tests on the encoded version of the $\QMA$ verification algorithm that come from the Local Hamiltonian
problem, together with the fact that the
history state of such computation is a low-energy state and is simulatable (which is used to establish the \QMA{}-hardness of \CLDM{}).

We remark that a
direct classical version of locally simulatable proofs as we define them is impossible. This is because,
given the local values of a classical proof, it is always possible to reconstruct the full proof by gluing these pieces together.
The fact that this operation is hard to perform quantumly is intrinsically related to entanglement: given the local density
matrices, it is not a priori possible to know which parts are entangled in order to glue
them together. As discussed in the next section, this allows us to achieve a
type of simple zero-knowledge protocol that defies all classical intuition.

\subsubsection{Locally Simulatable Proofs in Action}
\label{subsection:intro-locally-simulatable-proofs}
We now sketch how each of \Cref{result:xi-proof}--\Cref{result:NIZK} is obtained via the lens of locally simulatable proofs.

\myparagraph{Zero Knowledge}
We use the characterization $\QMA=\SimQMA$ to  give a new zero-knowledge
proof system for~$\QMA$. Our protocol is much simpler than previous
results~\cite{BJSW16,BJSW20}, and it follows the ``commit-challenge-response'' structure of a
 $\Sigma$-protocol. Since our commitment is a quantum state (the challenge and
response are classical), we call this type of protocol a ``\emph{$\Xi$-protocol}'' (see \Cref{sec:results}).

The main idea is to use the quantum one-time pad to split the first message in
the protocol
into a quantum and a classical part. More concretely,
the prover sends $X^aZ^b\ket{\psi}$ and commitments to each bit of $a$ and $b$ to
the verifier, where $\ket{\psi}$ is  a locally simulatable quantum witness for some
instance~$x$ and $a$ and
$b$ are uniformly random strings.
The
verifier picks some $c \in [m]$, which corresponds to one of the tests  of the
simulatable proof system, and asks the prover to open the commitment of the
encryption keys to the corresponding qubits.  The honest prover opens the
commitment corresponding to the one-time pad keys of the qubits involved in
test $c$.   The verifier then checks if:
\begin{enumerate*}[label=(\roman*)]
\item the openings are correct and,
\item  the decrypted
reduced state passes test $c$.
\end{enumerate*}

Assuming the existence of unconditionally binding and computationally hiding
commitment schemes, we show that our protocol
is a computational zero-knowledge proof system for \QMA{}.  Completeness and
soundness follow trivially, whereas the zero-knowledge property is established
by constructing a simulator that exploits the properties of the locally simulatable proof
system and the rewinding technique of Watrous~\cite{Wat09b}.

To the best of our knowledge, this is the first time that quantum techniques are used in zero-knowledge to achieve a
commit-and-open protocol that \emph{requires no randomization of the witness}.
Indeed, for reasons already discussed, all classical zero-knowledge $\Sigma$
protocols require a \emph{mapping} or \emph{randomization} of the witness
(\emph{e.g.} in the $\COL$ protocol, this is the permutation that is applied to
the colouring before the commitment is made).
We thus conclude that quantum
information enables a new level of  encryption that is
not possible classically:
the ``juicy'' information is present in the global state, whose local
parts are {\em fully} known~\cite{GSY19}.

\myparagraph{Proof of Quantum Knowledge for~$\QMA$}

As discussed in \Cref{sec:results}, our first challenge here is to define a Proof of Quantum Knowledge (PoQ). We recall that in the
classical setting, we require an extractor that outputs some witness that passes
the $\NP$ verification with probability~$1$, whenever the verifier accepts with
probability greater than some parameter $\kappa$, known as the knowledge
error.

In the quantum case, given:
\begin{enumerate*}[label=(\roman*)]
\item that we are not able to clone quantum states and
\item \QMA{} is not known to be closed under perfect completeness, the best that
we can hope for is to extract some quantum state that would pass the $\QMA$
verification with some probability to be related to the acceptance probability
in the interactive protocol, whenever this latter value is above some
threshold~$\kappa$.
\end{enumerate*}

To define a PoQ, we first  fix the verification algorithm $V_x$ for some instance
of a problem in  \QMA{}.
We also assume   $P^*$ to be a prover that makes the verifier accept with probability at least
$\eps > \kappa$ in the $\Xi$ protocol.\footnote{Note that we reserve the word
``verifier'' here for the $\Xi$ protocol and refer to $V_x$ as the \QMA{}  verification
algorithm.} We assume that
$P^*$ only performs unitary operations on a private and message
registers.
We then define a quantum polynomial-time algorithm $K$ that has oracle access to
$P^*$, meaning that $K$ can execute the unitary operations of $P^*$, their
inverse operations and has access to the message register of
$P^*$.\footnote{This model is already considered by \cite{Unr12} in his work of
quantum proofs of knowledge for \NP{}.}
The
protocol is said to be a Proof of Quantum Knowledge if $K$ outputs, with
non-negligible probability, some quantum state $\rho$ that  would make $V_x$
accept with probability at least $q(\eps,n)$, where $q$ is known as the quality function,
or aborts otherwise.

The difficulty in showing that our $\Xi$ protocols are PoQs lies in the fact
that any measurement performed by the
extractor disturbs the state held by $P^*$, and therefore
when we rewind $P^*$ by applying the inverse of his operation, we do not come back to the original state. We
overcome this difficulty in the following way.  We set $\kappa$ to be some value
very close to $1$, namely $\kappa = 1 - \frac{1}{p(n)}$ for some  large enough
polynomial $p$. Our extractor starts by simulating $P^*$ on the first message of
the $\Xi$ protocol, and then holds the (supposed) one-time-padded state and the
commitments to the one-time-pad keys.
$K$ follows by iterating over all possible challenges of the
$\Xi$ protocol, runs $P^*$ on this challenge,  perform the verifier's check and
then rewinds~$P^*$. By the assumption that $P^*$ has a very high acceptance
probability, the measurements performed by $K$ do not disturb the state too
much, and in this case, $K$ can retrieve the correct one-time pads for every
qubit of the witness. If $K$ is successful (meaning that $k$ is able to open every
committed bit), then $K$ can decode the original one-time-padded state
and it is a good witness for $V_x$ with high probability.

We then analyse the sequential repetition of the protocol, that allows us to
have a PoQ with {\em exponentially small} knowledge error $\kappa$, and
extracts one good witness from  $P^*$ (out of the polynomially many copies that
$P^*$ should have in order to cause the verifier to accepted in the multiple runs of the protocol).

\myparagraph{Non-Interactive zero knowledge proof for QMA in the secret parameter model}

Finally, in \Cref{sec:NIZK}, we achieve our non-interactive statistical
zero-knowledge protocol for \QMA{} in the secret parameter setting using
 techniques similar to our $\Xi$ protocol: the trusted party
chooses the one-time pad key and a random (and small) subset of these values
that are reported to the verifier. Since the prover does not know which are the
values that were given to the verifier, he should act as in the
$\Xi$-protocol, but now the verifier does not actually need to ask for the
openings, since the trusted dealer has already sent them.
Although this is a less natural model, we hope that this result will shed some
light in developing $\QNIZK$ proofs for $\QMA$ in more commonly-used models.

\subsection{Open problems}

\myparagraph{Further $\QMA$-complete languages} We note that a number of
problems are currently known to be $\QMA$-complete under Turing reductions,
including the $N$-representability~\cite{LCV07}~\footnote{
  In \cite{LCV07}, the authors reduce the Local Hamiltonian problem on qubits
  into the Local Hamiltonian problem on fermions, and then they propose a Turing
  reduction from LH on fermions to the $N$-representability problem. The missing
  step is reducing CLDM directly to the $N$-representability problem, which
  might be straightforward, but needs a formal proof.  }
and bosonic
$N$-representability problems~\cite{WMN10} as well as the universal functional
of density function theory (DFT)~\cite{SV09}. It is an open question if these
problems can be shown to be $\QMA$-complete under Karp reductions using
the techniques presented in our work.

\myparagraph{Complexity of $k$ CLDM for $k < 5$} We prove in this work that
$5$-CLDM
is \QMA{}-hard under Karp reductions. We leave as an open problem proving if the
problem is still \QMA{}-complete for $k < 5$.

\myparagraph{Marginal reconstruction problem}
We remark that the classical version of CLDM is defined as follows: given the
description of $m$ marginal distributions on sets of bits $C_1,\ldots,C_m$, such
that $|C_i| \leq k$, decide if there is a probability distribution that is close
to those marginals, or such a distribution does not exist.  This problem was
proven $\NP$-complete by Pitowsky~\cite{Pit91}, and its containment
in $\NP$ is proved  by using the fact that such distribution can be seen as
a point $p$ in the {\em correlation polytope} in a polynomial-size Hilbert space. In
this case, by Caratheodory's theorem, $p$ is a convex combination of
polynomially many vertices of such polytope, and therefore these
vertices serve as the $\NP$-proof and a linear program verifies if there is a
convex combination of them that is consistent with the marginals of the
problem's instance.

The difference here is
that the proof and the marginals are different (but connected) objects.
We leave as an open
problem if we can extract a  notion of a locally simulatable classical proof from this (or any other)
problem, and its applications to cryptography and complexity theory. In
particular, we wonder if there is a natural zero-knowledge protocol for this problem.

\myparagraph{Applications of quantum ZK protocols} In classical cryptography, ZK
and PoK protocols are a fundamental primitive since they are crucial ingredients in
a plethora of applications. We discussed in  \Cref{sec:results}
that our quantum
ZK PoQ for \QMA{} could be used as a proof-of-work for quantum
computations. An interesting open problem
is finding other settings in which the benefits of our simple  ZK protocols for \QMA{}
can be applied. We list now some
possibilities that could be explored in future work:
authentication with uncloneable credentials~\cite{CDS94};
proof of quantum ownership~\cite{BJM19iacr}; or
ZK PoQ verification for  quantum money~\cite{AC12}.

\myparagraph{Practical ZK protocols for QMA} Even if we reach a conceptually much simpler ZK protocol
for QMA, the resources needed for it are still very far
from practical. We leave as an open problem if one could devise other protocols
that are more feasible from a physical implementation viewpoint, which could
include classical communication protocols based on the protocols proposed by
Vidick and Zhang~\cite{VZ19arxiv}, or
device-independent ones based on the ideas of Grilo~\cite{Gri19}.

\myparagraph{Non-interactive Zero-knowledge protocols for $\QMA$ in the CRS model} In this work, we propose a
QNIZK protocol where the information provided by the trusted dealer is
asymmetric. We leave as an open problem if one could devise a protocol where the
dealer distributes a common reference string (CRS)(or shared EPR pairs) to the prover
and the verifier.

A possible way of achieving such non-interactive protocol would be to explore
the properties of $\Xi$-protocols, as done classically with $\Sigma$-protocols.
For instance, the well-known Fiat-Shamir transformation~\cite{FS87}  allows us
to make $\Sigma$-protocols non-interactive (in the Random Oracle model). We
wonder if there is a version of this theorem when the first message can be
quantum.

\myparagraph{Witness indistinguishable/hiding protocols for $\QMA$}
Classically, there are two weaker notions that can substitute for ZK in
different applications. In Witness Indistinguishable (WI) proofs, we require that the
verifier cannot distinguish if she is interacting with a prover holding a witness
$w_1$ or~$w_2$, for any $w_1 \ne w_2$. In Witness Hiding (WH), we require that the
verifier is not able to cook-up a witness for the input herself. We note that
zero-knowledge
implies both  such definitions, and we leave as an open problem
finding WI/WH protocols for $\QMA$ with more desirable properties than the known ZK
protocols.

\myparagraph{Computational Zero-Knowledge proofs vs.~Statistical Zero-Knowledge
arguments}
Classically, it is known that the class of problems
with computational ZK proofs is closely related to the class of problems with statistical
ZK arguments~\cite{OV07}. We wonder if this relation is also true in the quantum setting.

\subsection{Concurrent and subsequent works}
Concurrently to this work, Bitansky and Shmueli~\cite{BS19}  proposed
the first quantum zero-knowledge argument system  for \QMA{} with constant
rounds and negligible soundness. Their main building block is a {\em non
black-box quantum extractor} for a post-quantum commitment scheme.

Also concurrently to this work,
Coladangelo, Vidick and Zhang~\cite{CVZ19} proposed a non-interactive argument of quantum knowledge for \QMA{}
with a quantum setup phase (that is independent of the witness) and a classical online phase.
Subsequently to our work, Alagic, Childs, Grilo and Hung~\cite{ACGH20} proposed
the first non-interactive zero-knowledge argument system for \QMA{} where the
communication is purely {\em classical}. Their protocol works in the random
oracle model with setup.

All of these protocols achieve only computational soundness and their security
relies on stronger cryptographic assumptions
(namely the Learning with Errors assumption, as well as post-quantum fully homomorphic encryption). Their proof structure follows by combining
powerful cryptographic constructions based on these primitives to achieve their
results.
On the
other hand, the key idea in our protocols is to use structural properties of
\QMA{} and with that, we can achieve a zero-knowledge protocol with statistical
soundness under the very
weak assumption that  post-quantum
one-way functions exist.

\subsection{Differences with previous version}

In a previous version of this work\footnote{available at \url{https://arxiv.org/abs/1911.07782v1}}, we used the results of \cite{GSY19} almost in a black-box way. In contrast, in the current version,  we provide a new proof for the technical results that we need from~\cite{GSY19}; this  not only makes this work self-contained, but it also provides a conceptually much simpler proof. More concretely, the sketch that was presented in Appendix~A of the previous version has now become a full proof in \Cref{sec:simulatable-history}.
We note that this also allowed us to find a small bug in the proof of~\cite{GSY19}, and provide a relatively easy fix (see~\Cref{R:bug}).

Fermi Ma pointed out a bug in the proof sketch of a proposal of statistical
zero-knowledge argument for \QMA{} that we have in previous versions of this
paper.
\subsection{Structure}

The remainder of this document is structured as follows:
\Cref{sec:preliminaries} presents Preliminaries and Notation. In
\Cref{sec:consistency}, we prove our results on the CLDM problem  and we present our framework of simulatable proofs, with the technical portion of this contribution appearing in \Cref{sec:simulatable-history}.
\Cref{sec:xizk-protocol} establishes the zero-knowledge~$\Xi$ protocol
for~$\QMA$, while in \Cref{sec:PoQ}, we define a proof of quantum knowledge and show that the interactive proof system
satisfies the definition. Finally, in \Cref{sec:qniszk}, we show a
non-interactive zero-knowledge proof for~\QMA in the secret parameter model.

\subsection*{Acknowledgements}
We thank Dorit Aharonov, Thomas Vidick, and the anonymous reviewers for help in improving the presentation of this work.
We thank Andrea Coladangelo, Thomas Vidick and Tina Zhang for discussions on
the definition of proofs of quantum knowledge. A.G.~thanks Christian Majenz for
discussions on a suitable title for this work. We thank Fermi Ma for pointing
out a bug in the proof sketch of the statistical zero-knowledge argument result
which was withdrawn.
A.B.~is supported by the U.S. Air Force Office of
Scientific Research under award number FA9550-17-1-0083,  Canada's  NFRF and NSERC, an Ontario ERA, and the University of Ottawa's Research Chairs program.

 \else

\clearpage
   \section{Introduction}

   \label{sec:Intro}

The complexity class $\QMA$ is the quantum analogue of
$\NP$, the class of problems whose solutions can be verified in deterministic
polynomial time. More precisely,  in $\QMA$, an all-powerful prover produces a
quantum proof that is verified by a quantum polynomially-bounded verifier. Given
the probabilistic nature of quantum computation, we require that for true
statements, there exists a quantum proof that makes the verifier accept with high probability (this
is called {\em completeness}),
whereas all ``proofs'' for  false statements are rejected with high
probability (which is called {\em soundness}).

The class $\QMA$ was first defined by Kitaev \cite{KSV02}, who also showed that deciding if a $k$-local
Hamiltonian problem has low-energy states  is $\QMA$-complete. The importance of this result is two-fold: first, from a theoretical computer science perspective, it is  the
 quantum analogue of the Cook-Levin theorem, since it establishes the first non-trivial $\QMA$-complete problem. Secondly, it shows deep links between physics and complexity theory, since the $k$-local Hamiltonian problem is an important  problem in many-body
physics. Thus, a better understanding of \QMA{} would lead
to a better understanding of the power of quantum resources in proof
verification, as we well as the role of {\em quantum entanglement}  in
low-energy states.

Follow-up work strengthened our
understanding of this important complexity class, \emph{e.g.},  by showing that $\QMA$ is contained in the complexity
class $\PP$~\cite{KW00}\footnote{\PP{} is the complexity class of
decision problems that can be solved by probabilistic polynomial-time algorithms
with error strictly smaller than $\frac{1}{2}$.};
that it is possible to reduce completeness and soundness errors without increasing the length of the
witness~\cite{MW05}; understanding the difference between quantum and classical
proofs~\cite{AK07,GKS16,FK18}; the possibility of perfect
completeness~\cite{Aar09b}; and, more recently, the relation of \QMA{} with non-local
games~\cite{NV17,NV18,CGJV19}.

Also, much follow-up work focused on understanding the complete problems for
$\QMA$, mostly by improving the parameters of the  \QMA{}-hard Local Hamiltonian
problem, or making it closer to models more physically
relevant~\cite{KR03,Liu06,KKR06,OT08,CM14,HNN13,BC18}.  In 2014, a
survey of $\QMA$-complete languages~\cite{Boo14} contained a list of 21 general
problems that are known to be $\QMA$-complete\footnote{We remark that these
problems can be clustered as variations of a handful of base problems.}, and since then,
the situation has not drastically changed. This contrasts with the development
of $\NP$, where only a few years after the developments surrounding
\textsf{3--SAT},  Karp published a theory of reducibility, including a list of
21 $\NP$-complete problems~\cite{Kar72}; while 7 years later, a celebrated book
by Garey and Johnson surveyed over 300 $\NP$-complete problems~\cite{GJ90}.\footnote{The first edition of Garey and Johnson~\cite{GJ90} was published in 1979.}

Recently, the role of \QMA{} in quantum cryptography has also been explored. For instance, several results used ideas of the
\QMA{}-completeness of the Local Hamiltonian problem in order to perform verifiable
delegation of quantum computation~\cite{FHM18,Mah18,Gri19}.
Furthermore, another line of work studies {\em zero-knowledge protocols} for
\QMA{}~\cite{BJSW16,BJSW20,VZ19arxiv}; which is extremely relevant, given the fundamental importance  of
 zero-knowledge protocols for \NP{} in cryptography .

Despite the multiple advances in our understanding of $\QMA$ and related techniques,
 a number of fundamental open questions remain. %
In this work, we solve some of these open problems by showing:
\begin{enumerate*}[label=(\roman*)]
\item \QMA{}-hardness of the
Consistency of Local Density Matrix (CLDM) problem under Karp reductions;
\item  ``commit-and-open''
  Zero-Knowledge (ZK) proof of quantum knowledge (PoQ) protocols for \QMA{}; and
\item  a non-interactive zero-knowledge (NIZK)
protocol in the secret parameter scenario.
\end{enumerate*}
Our main technical contribution consists in showing that every problem in \QMA{} admits a verification algorithm whose history state\footnote{See \Cref{eq:history-state-ex}.} is {\em  locally simulatable}, meaning that the reduced density matrices on any small set of qubits is efficiently computable (without knowledge of the quantum witness).
In order to be able to explain our results in more details and appreciate their contribution to a better
understanding of \QMA{}, we first give an overview of these areas and how they relate to these particular problems.

\subsection{Background}
\label{sec:background}
In this section, we discuss the background on the topics that are relevant to this work, summarizing their current state-of-the-art.

\myparagraph{Consistency of Local Density Matrices (CLDM)}
The Consistency of Local Density
matrices problem ($\CLDM$) is as follows: given the classical description of local density matrices
$\rho_1,\ldots,\rho_m$, each on a set of at most $k$ qubits and for a global system of $n$ qubits, is
there a state $\tau$ that is consistent with such reduced states? Liu~\cite{Liu06}
showed that this problem is in \QMA{} and that it is \QMA{}-hard under Turing
reductions, \emph{i.e.}, a deterministic polynomial time algorithm with access to an
oracle that solves $\CLDM$ in unit time can solve any problem in \QMA{}.

We remark that this type of reduction is rather troublesome for \QMA{}, since
the class is not known (nor expected) to be closed under complement, \emph{i.e.}, it is
widely believed that $\QMA{} \ne \class{co}\QMA{}$. If this is indeed the case, then Turing
reductions do not allow a black-box generalization of results regarding the $\CLDM$ problem  to all
problems in~\QMA{}.
This highlights the open problem of establishing the $\QMA{}$-hardness of the $\CLDM$ problem under Karp reductions, \emph{i.e.}, to show  an efficient mapping between yes- and no-instances of any \QMA{} problem  to yes- and
no-instances of $\CLDM$, respectively.

\myparagraph{Zero-Knowledge (ZK) Proofs for $\QMA$}
In an \emph{interactive} proof, a limited party, the \emph{verifier}, receives the
 help of some untrusted powerful party, the \emph{prover}, in order to decide if some
statement is true.
This is  a generalization of a \emph{proof}, where we allow multiple rounds of interaction.
As usual, we require that the  \emph{completeness} and \emph{soundness} properties hold.
For cryptographic applications, the \emph{zero-knowledge (ZK)} property is often
desirable: here, we require that the verifier learn nothing
from the interaction with the prover. This property is formalized by
showing the existence of an efficient \emph{simulator},
which is able to reproduce (\emph{i.e., simulate}) the output of any given
verifier on a \emph{yes} instance (without having direct access to the actual
prover or witness)\footnote{Different definitions of ``reproduce'' result in
different definitions of zero-knowledge protocols. A protocol is \emph{perfect
zero-knowledge} if the distribution of the output of the simulator is exactly the
same as the distribution of output of transcripts of the protocol.
A protocol is \emph{statistical
zero-knowledge} if such distributions are statistically close. Finally, a
protocol is \emph{computational zero-knowledge} if no efficient algorithm can
distinguish both distributions. The convention is that in the absence of such specification,
we are considering the case of computational zero-knowledge.}.

As paradoxical as it sounds, statistical zero-knowledge
interactive proofs  are known to be possible for a host of languages, including the Quadratic
Non-Residuosity, Graph Isomorphism, and Graph Non-Isomorphism
problems~\cite{GMW91,GMR89}; furthermore, all languages that can be proven by
multiple provers ($\MIP$) admits perfect zero-knowledge MIPs~\cite{BGKW88}.
What is more, by introducing computational assumptions, it was shown that all languages that admit an
 interactive proof system also admit a zero-knowledge interactive proof
 system~\cite{BOGG88}. Zero-knowledge interactive proof systems have had a
 profound impact in multiple areas, including cryptography \cite{GMW87} and complexity
 theory \cite{Vad07}.

We now briefly review the zero-knowledge interactive proof system for the
$\NP$-complete problem of Graph 3-colouring (\COL{}). This is a
3-message proof system, and has the additional property that, given a witness,
the prover is efficient. As a first message, the prover \emph{commits} to a
\emph{permutation} of the given 3-colouring (meaning that the prover randomly
permutes the colours to obtain colouring $c$, and produces a list $(v_i,
\mathsf{commit}(c(v_i)))$, using a cryptographic primitive $\mathsf{commit}$
which is a \emph{commitment scheme}). In the second message, the verifier
chooses uniformly at random an edge $\{v_i, v_j\}$ of the graph. The prover
responds with the information that allows the verifier to open the commitments to
the colouring of the vertices of this edge (and nothing more). The verifier \emph{accepts} if and
only if the revealed colours are different. It is easy to see that the protocol is complete and sound. For the zero-knowledge property, the simulator consists in a process that \emph{guesses} which edge will be requested by the verifier and commits to a colouring that satisfies the prover in case this guess is correct. If the guess is incorrect, the technique of \emph{rewinding} allows the simulator to re-initialize the interaction until it is eventually successful. Protocols that follow the \emph{commit-challenge-response} structure of this proof system are called \emph{$\Sigma$-protocols}\footnote{The Greek letter $\Sigma$  visualizes the flow of the protocol.} and, due to their simplicity, they play a very important role, for instance in the celebrated Fiat-Shamir transformation~\cite{FS87}.

The foundations of zero-knowledge in the quantum world were established by
Watrous, who showed a technique called \emph{quantum rewinding}~\cite{Wat09b}
which is used to show the security of some classical zero-knowledge proofs
(including the protocol for \COL{} described above), even against
quantum adversaries. The importance of this technique is that quantum
measurements typically \emph{disturb} the measured state. When we consider
quantum adversaries, such difficulties concern even \emph{classical} proof
systems, due to the rewinding technique that is ubiquitous (see example in the
case of \COL{} above).
Indeed, in the quantum setting, intermediate measurements (such as checking if
the guess is correct) may compromise the success of future executions, since it
is not possible {\em a priori} to ``rewind'' to a previous point in the execution in a
black-box way.

Another dimension where quantum information poses new challenges is in the study
of interactive proof systems for \emph{quantum} languages.
 We point out that Liu~\cite{Liu06}
observed very early on that the $\CLDM$ problem should admit a simple zero-knowledge proof
system following the ``commit-and-open'' approach, as in the \COL{} protocol.
Inspired by this observation,
recent progress has established the existence of zero-knowledge protocols
for all of $\QMA$~\cite{BJSW16,BJSW20}. We note that although the proof system used there is
reminiscent of a $\Sigma$-protocol, 
there are a number of reasons why it is not a ``natural'' quantum analogue of a $\Sigma$ protocol. These include:
\begin{enumerate*}[label=(\roman*)]
\item the use of a coin-flipping protocol, which makes the communication cost higher than 3 messages;
\item the fact that the verifier's message is not a random challenge; and
\item the final answer from the prover is not only the opening of some committed
values.
\end{enumerate*}

Recently, Vidick and Zhang~\cite{VZ19arxiv} showed how to make classical all of the
interaction between the verifier and the prover in~\cite{BJSW16,BJSW20}, by
considering {\em argument systems}\footnote{
Argument systems are a relaxation of proof systems where
the prover is also bounded to  polynomial-time computation, and, for positive instances, the
prover is provided a witness to the \NP{} instance. This model allows much more
efficient protocols which enables it to be used in practice
\cite{BSCG+14,PHGR16,BSCR+19}.
}
instead of
proof systems. In their protocol, they compose the result of
Mahadev~\cite{Mah18} for verifiable delegation of quantum computation by
classical clients with the zero-knowledge protocol of~\cite{BJSW16,BJSW20}.

\myparagraph{Zero-Knowledge Proofs of Knowledge (PoK)} In a zero-knowledge proof,
the verifier becomes convinced of the \emph{existence} of a witness, but this a
priori has no bearing on the prover actually having in her possession  such a
witness. In some circumstances, it is important to guarantee that the prover
actually has a witness. This is the realm of a   \emph{zero-knowledge proof of
knowledge (PoK)}~\cite{GMR89,BG93}.

We give an example to depict this subtlety. Let us consider the task of anonymous
credentials~\cite{Cha83}. In this setting, Alice wants to authenticate into
some online service using her private credentials. In order to protect her credentials, she could engage in a
zero-knowledge proof; this, however would be unsatisfactory, since the verifier in this scenario would be become convinced of the \emph{existence} of accepting credentials, which
does not necessarily translate to Alice actually being in the \emph{possession} of these credentials.
To remedy this situation,  the PoK property establishes an
``if-and-only-if'' situation: if the verifier accepts, then we can guarantee that
the  prover actually \emph{knows} a witness. This notion is formally defined by
requiring the existence of an \emph{extractor},
which is polynomial-time process $K$ that outputs a valid witness when given
oracle access to some prover $P^*$ that makes
the verifier accept with high enough probability.

In the quantum case, there has been some positive results in terms of the
security of classical proofs of knowledge for $\NP$ against quantum
adversaries~\cite{Unr12}. However, in the fully quantum case (that is,
proofs of quantum knowledge for $\QMA$), no scheme has been
proposed.  One of the possible reasons why no such proof of quantum knowledge
protocols was proposed is the lack of a {\em simple} zero-knowledge proof for~\QMA{}.

\myparagraph{Non-Interactive Zero-Knowledge Proofs (NIZK)}
The interactive nature of zero knowledge proof systems (for instance, in $\Sigma$-protocols) means that in some situations they are not applicable since they require the parties to be simultaneously online.
Therefore, another desired property  of such
proof systems is that they are \emph{non-interactive}, which means the whole
protocol consists in a single message from the prover to the
verifier. \emph{Non-interactive zero-knowledge
proofs} (NIZK) is a fundamental construction in modern cryptography and has far-reaching
applications, for instance to cryptocurrencies~\cite{BSCG+14}.

We note that NIZK is known to be impossible in the standard
model~\cite{GO94}, \emph{i.e.}, without extra assumptions, and therefore NIZK has been considered in different
models. In one of the models most relevant in cryptography,
we assume a common reference string (CRS)~\cite{BFM88}, which can be seen as a trusted
party sending a random string to both the prover and the verifier.
In another model,
the trusted party is allowed to send different (but correlated) messages to the
prover and the verifier; this is called the secret parameter
setup~\cite{PS05}. Classically, this model has been shown to be very powerful,
since even its {\em statistical} zero-knowledge version is equivalent to all of
the problems in the complexity class \class{AM} (this is the class that
contains problem that can be verified by public-coin polynomial-time
verifiers). As mentioned in~\cite{PS05}, this model
encompasses another model for NIZK where the prover and the verifier perform an
{\em offline} pre-processing phase (which is independent of the input) and then
the prover provides the ZK proof~\cite{KMO89}. This inclusion holds since the parties could
perform secure multi-party computation to compute the trusted party's operations.

In the quantum case, very little is known on non-interactive zero-knowledge.
Chailloux, Ciocan, Kerenidis and Vadhan studied this problem in a setup where the
message provided by the trusted party can depend on the instance of the
problem~\cite{CCKV08}. Recently, some results also showed that the
Fiat-Shamir transformation for classical protocols is still safe in the quantum
setting, in the quantum random oracle model~\cite{LZ19,DFMS19,Cha19iacr}.
One particular and intriguing open question is the possibility of NIZKs for \QMA{}.

\subsection{Results}
\label{sec:results}
As we have shown so far, the state-of-the-art in the study of $\QMA$ is that the
body of knowledge is still developing, and that there are some specific goals
that, if achieved, would help us better understand $\QMA$ and devise
new protocols for quantum cryptography. Given this context,  we present
now our results in more detail.

Our first result (\Cref{sec:consistency}) is to
show that the CLDM  problem is \QMA{}-hard under
Karp reductions, solving the  $14$-year-old problem proposed  by Liu~\cite{Liu06}.

\begin{mainresult} \label{result:cldm}
  The $\CLDM$ problem is $\QMA$-complete under Karp reductions.
\end{mainresult}
We capture the techniques used in establishing the above into a new
characterization of \QMA{} that provides the best-of-both worlds in terms of two
proof systems for \QMA{} in an abstract way:
we define \SimQMA{} as the complexity class with proof
systems that are \begin{enumerate*}[label=(\roman*)]
\item  locally verifiable (as in the Local Hamiltonian problem), and
\item  every reduced density matrix of the witness can be efficiently computed (as in the $\CLDM$ problem).
\end{enumerate*} This results is the basis for our applications to quantum cryptography:

\begin{application}\label{result:simqma}
  $\SimQMA = \QMA$.
\end{application}

Next, we define a quantum notion of a classical $\Sigma$-protocol, which we call
a $\Xi$-protocol\footnote{Besides being an excellent symbolic reminder of the interaction in a 3-message proof system, $\Xi$ is chosen  as a shorthand for what we might otherwise call a $q\Sigma$ protocol, due to the resemblance with the pronunciation as ``\textbf{csi}gma''.} (please note, both a $\Sigma$ and $\Xi$ protocol is also referred to
throughout as  ``commit-and-open'' protocols.)  Using our characterization
given in \Cref{result:simqma}, we show a $\QMA$-complete language that admits a  $\Xi$-protocol. Taking into account the importance of $\Sigma$ protocols for zero-knowledge proofs, we are able to show (\Cref{sec:xizk-protocol}) a quantum analogue of the celebrated $\cite{GMW91}$ paper:
\begin{application}\label{result:xi-proof}
  All problems in $\QMA$ admit a computational zero-knowledge $\Xi$-proof system.
\end{application}

Then we provide the definition of
Proof of Quantum Knowledge (PoQ).\footnote{This definition is  joint work with Coladangelo, Vidick and Zhang \cite{CVZ19}.} In short, we say that a proof system is a PoQ if there exists a quantum polynomial-time
\emph{extractor}~$K$ that has oracle access to a quantum prover  which makes the verifier accept with high enough probability, and  the extractor is able to output a sufficiently good
witness for a ``\QMA{}-relation''. We note that this definition for a PoQ is not a
straightforward adaptation of the classical definition; this is because \NP{} has many properties such as perfect completeness,
perfect soundness and even that proofs can be copied, that are not expected to
hold in the \QMA{} case. More details are given in~\Cref{sec:PoQ}.
We are then able to show that our $\Xi$
protocol for \QMA{} described in Result~\ref{result:xi-proof} is PoQ.
This is the first proof
of knowledge for~\QMA{}.\footnote{See also independent and concurrent work by Coladangelo, Vidick and Zhang~\cite{CVZ19}.}

\begin{application}\label{result:NIZK}
All problems in $\QMA$ admit a zero-knowledge proof of quantum knowledge  proof system
  and a statistical zero-knowledge proof of quantum knowledge argument system.
\end{application}
We remark that using techniques for post-hoc delegation of quantum
computation~\cite{FHM18}, our PoQ  for \QMA{} may be understood as a
\emph{proof-of-work} for quantum computations, since it could
be used to convince a verifier that the prover has indeed created the {\em history state} of some
pre-defined computation. This is very relevant in the scenario of
testing small-scale quantum computers in the most adversarial model possible: the zero-knowledge property ensures that the verifier learns nothing but the truth of the statement, while the PoQ property means that the prover has indeed prepared a ground state with the given properties.
Comparatively, all currently known protocols either make assumptions on the devices, or certify only the answer of the computation, but not the knowledge of the prover.

Finally, using the techniques of \Cref{result:xi-proof}, we show that every problem
in~\QMA{} has a non-interactive {\em statistical}
zero-knowledge proof in the secret parameter model.
We are even able to
strengthen
our result to the complexity class \class{QAM} (recall that in a~\class{QAM} proof system, the verifier first sends a random string to the
prover, who answers with a quantum proof). Note that $\class{QAM}$ trivially
contains~$\QMA$.

\begin{application}
  All problems in \class{QAM}  have a
  non-interactive statistical zero-knowledge protocol in the secret parameter
  model.
\end{application}

  Note that, as in the classical case~\cite{PS05}, our result also implies a QNIZK protocol
where the prover and the verifier run an offline (classical) pre-processing phase
(independent of the witness) and then the prover sends the quantum ZK proof to
the verifier.
We note also that even though these models are less relevant to the cryptographic applications of
NIZK, we think that our result moves us towards a QNIZK protocol for \QMA{}
in a more standard model.

\subsection{Techniques}
\label{sec:intro-simulatable}
\label{sec:techniques}

The starting point for our results are  {\em locally simulatable codes}, as
defined in \cite{GSY19}. We give now a rough intuition on the properties
of such codes and leave the details to \Cref{sec:simulatable-history}.

First, a quantum error correcting code is
\emph{$s$-simulatable} if there exists an efficient classical algorithm that
outputs the reduced density matrices of codewords on every subset of at most~$s$
qubits.
Importantly, this algorithm is oblivious of the logical state that
is encoded.
We note
that it was already known that  the reduced density matrices
of codewords hide the encoded information, since
quantum error correcting codes can be used in secret sharing
protocols~\cite{CGL99}, and in~\cite{GSY19} they show that there exist
codes such that the classical description of the reduced density matrices of the
codewords can be efficiently computed.
Next, \cite{GSY19} extends the notion of simulatability of {\em
logical operations} on encoded data as follows.
Recalling the theory of fault-tolerant quantum computation, according to which some quantum
error-correcting codes allow computations over \emph{encoded} data by using ``transversal'' gates and
encoded magic states. The definition of $s$-simulatability is extended to require that the simulator
also efficiently computes the reduced
density matrix on at most $s$ qubits of intermediate steps of the
{\em physical} operations that implement a logical gate on the encoded data
(again, by transversal gates and magic states).

\begin{example}\label{ex:simulatability}
 Let us suppose that the encoding map $\Enc$ admits transversal application of the one-qubit gate $G$,\emph{i.e.}, $G^{\otimes N}\Enc(\ket{\psi}) = \Enc(G\ket{\psi})$.
  The simulatability property requires that the density matrices on at most~$s$ qubits of $(G^{\otimes t} \otimes I^{\otimes (N -t)})\Enc(\ket{\psi})$ should be efficiently computed, for every $0 \leq t \leq N$.
\end{example}

In~\cite{GSY19}, the authors show that the concatenated Steane code is a locally simulatable code.
With this tool, in~\cite{GSY19}, it is shown that  every
$\MIP^*$ protocol\footnote{$\MIP^*$ is the set of languages that admit
a classical \emph{multi-prover} interactive proof, where, in addition, the
provers share entanglement} can be made zero-knowledge, thus quantizing the celebrated
result of~\cite{BGKW88}.
Here, we provide an alternative proof for the simulatability of concatenated Steane codes. Our new proof is much simpler than the proof provided in \cite{GSY19}, but it holds for a slightly weaker statement (but which is already sufficient to derive the results in \cite{GSY19}).
Then, for the first time,  we apply the
techniques of simulatable codes from~\cite{GSY19} to~$\QMA$, which enables us to solve many open problems as previously described.

In order to explain our approach to achieving our main result,  we first recall the  quantum
Cook-Levin theorem proved by~Kitaev \cite{KSV02}. In his proof, Kitaev uses the
circuit-to-Hamiltonian construction~\cite{Fey82}, mapping an arbitrary \QMA{}
verification circuit $V = U_T \ldots U_1$ to a local Hamiltonian $H_V$ that enforces
that low energy states are {\em history states} of the computation,
\emph{i.e.}, a \emph{superposition} of the snapshots of $V$ for every timestep $0 \leq
t\leq T$:
\begin{equation}\label{eq:history-state-ex}
    \ket{\hist} =
    \frac{1}{\sqrt{T+1}}\sum_{t=0 \ldots T+1} \ket{t} \otimes
    U_t\ldots U_1\ket{\psi_{init}}.
\end{equation}
In the above, the first register is called the {\em clock} register, and it
encodes the timestep of the computation, while the second register contains the
snapshot of the computation at time $t$, \emph{i.e.}, the quantum gates $U_1, \ldots ,U_t$
applied to the initial state $\ket{\psi_{init}} = \ket{\phi}\ket{0}^{\otimes A}$, that consists of the quantum
witness and auxiliary qubits. The Hamiltonian $H_V$ also
guarantees that $\ket{\psi_{init}}$ has the correct form at $t = 0$, and that
the final step {\em accepts}, \emph{i.e.}, the output qubit is close to~$\ket{1}$.

In \cite{GSY19}, they note that an important obstacle to making a state similar to $\ket{\hist}$\footnote{In \cite{GSY19}, they are simulating history states for $\mathsf{MIP}^*$ computation and therefore they need to deal also with arbitrary Provers' operations.}
locally simulatable is its dependence on the witness state $\ket{\phi}$.
The solution is to consider
a different verification algorithm $V'$ that implements $V$ on {\em
encoded data}, much like in the theory of fault-tolerant quantum computing. In more details, for a fixed locally simulatable code,
$V'$  expects the encoding of the original witness $\Enc(\ket{\phi})$ and then,
with her raw auxiliary states, she creates encodings of auxiliary states $\Enc(\ket{0})$ and
magic states $\Enc(\ket{\MS})$, and then performs the computation $V$ through
transversal gates and magic state
gadgets, and finally decodes the output qubit. This gives rise to a new history state:

\begin{equation}
    \ket{\hist'} =
    \frac{1}{\sqrt{T'+1}}\sum_{t=0 \ldots T'+1} \ket{t} \otimes
    U'_t\ldots U'_1\ket{\psi_{init}'},
\end{equation}
where $\ket{\psi_{init}'} = \Enc(\ket{\phi})\ket{0}^{\otimes A'}$ and
$U'_1,\ldots,U_{T'}$ are the gates of $V'$ described above.
Using the techniques from \cite{GSY19},\footnote{We remark that we also need to fix a small bug in their proof.  The bug fix deals with technicalities regarding $V'$ and~$\ket{\psi'}$ that are beyond the scope of this overview. See \Cref{S:new-verification} and \Cref{R:bug} for more details.}   we can show that from the properties of the locally
simulatable codes, the reduced density matrix on every set of $5$ qubits of $\ket{\hist'}$ can be
efficiently computed.  In this work, we prove that these reduced density
matrices are in fact \QMA{}-hard instances of \CLDM. More concretely, we show
that these reduced density matrices of a hypothetical history state of an
accepting \QMA{}-verification can  always be computed, and there exists a global
state (namely the history state) consistent with these reduced density matrices
if and only if the original \QMA{} verification accepts with overwhelming probability (and
therefore we are in the case of a yes-instance).

Our main result opens up a number of possible applications to cryptographic
settings. However, as we discussed in ~\Cref{sec:results} we face a tradeoff. In $\CLDM$,
we have
the description of the local density matrices, which yields a zero-knowledge $\Xi$ protocol.
On the other hand, the $\QMA$ verification for CLDM is non-local: we need multiple copies of the global state to
perform tomography on the reduced states,\footnote{See \Cref{lem:containment-qma}.} instead of a single
copy that is needed in the Local Hamiltonian problem.

In order to combine these two desired properties in a single object, we describe
a powerful technique that we call \emph{locally simulatable proofs}.
In a
locally simulatable proof system for some problem $A = (\ayes,\ano)$, we require that:
\begin{enumerate*}[label=(\roman*)]
\item the
verification test performed by the verifier acts on at most $k$ out of the $n$
qubits of the proof, and
\item for every $x \in \ayes$,  there exists a
locally simulatable witness $\ket{\psi}$, \emph{i.e.}, a state $\ket{\psi}$ that passes all the
local tests and such that for every $S\subseteq [n]$ with $|S| \leq k$,  it is possible to compute the reduced state of the
$\ket{\psi}$ on $S$ efficiently (without the help of the prover).
\end{enumerate*}
Notice that we have no extra restrictions on $x \in \ano$, since any quantum witness
should make this verifier  reject with high probability.

We then show that  all problems in $\QMA$ admit a locally simulatable proof
system. In order to achieve this, we use the local tests on the encoded version of the $\QMA$ verification algorithm that come from the Local Hamiltonian
problem, together with the fact that the
history state of such computation is a low-energy state and is simulatable (which is used to establish the \QMA{}-hardness of \CLDM{}).

We remark that a
direct classical version of locally simulatable proofs as we define them is impossible. This is because,
given the local values of a classical proof, it is always possible to reconstruct the full proof by gluing these pieces together.
The fact that this operation is hard to perform quantumly is intrinsically related to entanglement: given the local density
matrices, it is not a priori possible to know which parts are entangled in order to glue
them together. As discussed in the next section, this allows us to achieve a
type of simple zero-knowledge protocol that defies all classical intuition.

\subsubsection{Locally Simulatable Proofs in Action}
\label{subsection:intro-locally-simulatable-proofs}
We now sketch how each of \Cref{result:xi-proof}--\Cref{result:NIZK} is obtained via the lens of locally simulatable proofs.

\myparagraph{Zero Knowledge}
We use the characterization $\QMA=\SimQMA$ to  give a new zero-knowledge
proof system for~$\QMA$. Our protocol is much simpler than previous
results~\cite{BJSW16,BJSW20}, and it follows the ``commit-challenge-response'' structure of a
 $\Sigma$-protocol. Since our commitment is a quantum state (the challenge and
response are classical), we call this type of protocol a ``\emph{$\Xi$-protocol}'' (see \Cref{sec:results}).

The main idea is to use the quantum one-time pad to split the first message in
the protocol
into a quantum and a classical part. More concretely,
the prover sends $X^aZ^b\ket{\psi}$ and commitments to each bit of $a$ and $b$ to
the verifier, where $\ket{\psi}$ is  a locally simulatable quantum witness for some
instance~$x$ and $a$ and
$b$ are uniformly random strings.
The
verifier picks some $c \in [m]$, which corresponds to one of the tests  of the
simulatable proof system, and asks the prover to open the commitment of the
encryption keys to the corresponding qubits.  The honest prover opens the
commitment corresponding to the one-time pad keys of the qubits involved in
test $c$.   The verifier then checks if:
\begin{enumerate*}[label=(\roman*)]
\item the openings are correct and,
\item  the decrypted
reduced state passes test $c$.
\end{enumerate*}

Assuming the existence of unconditionally binding and computationally hiding
commitment schemes, we show that our protocol
is a computational zero-knowledge proof system for \QMA{}.  Completeness and
soundness follow trivially, whereas the zero-knowledge property is established
by constructing a simulator that exploits the properties of the locally simulatable proof
system and the rewinding technique of Watrous~\cite{Wat09b}.

To the best of our knowledge, this is the first time that quantum techniques are used in zero-knowledge to achieve a
commit-and-open protocol that \emph{requires no randomization of the witness}.
Indeed, for reasons already discussed, all classical zero-knowledge $\Sigma$
protocols require a \emph{mapping} or \emph{randomization} of the witness
(\emph{e.g.} in the $\COL$ protocol, this is the permutation that is applied to
the colouring before the commitment is made).
We thus conclude that quantum
information enables a new level of  encryption that is
not possible classically:
the ``juicy'' information is present in the global state, whose local
parts are {\em fully} known~\cite{GSY19}.

\myparagraph{Proof of Quantum Knowledge for~$\QMA$}

As discussed in \Cref{sec:results}, our first challenge here is to define a Proof of Quantum Knowledge (PoQ). We recall that in the
classical setting, we require an extractor that outputs some witness that passes
the $\NP$ verification with probability~$1$, whenever the verifier accepts with
probability greater than some parameter $\kappa$, known as the knowledge
error.

In the quantum case, given:
\begin{enumerate*}[label=(\roman*)]
\item that we are not able to clone quantum states and
\item \QMA{} is not known to be closed under perfect completeness, the best that
we can hope for is to extract some quantum state that would pass the $\QMA$
verification with some probability to be related to the acceptance probability
in the interactive protocol, whenever this latter value is above some
threshold~$\kappa$.
\end{enumerate*}

To define a PoQ, we first  fix the verification algorithm $V_x$ for some instance
of a problem in  \QMA{}.
We also assume   $P^*$ to be a prover that makes the verifier accept with probability at least
$\eps > \kappa$ in the $\Xi$ protocol.\footnote{Note that we reserve the word
``verifier'' here for the $\Xi$ protocol and refer to $V_x$ as the \QMA{}  verification
algorithm.} We assume that
$P^*$ only performs unitary operations on a private and message
registers.
We then define a quantum polynomial-time algorithm $K$ that has oracle access to
$P^*$, meaning that $K$ can execute the unitary operations of $P^*$, their
inverse operations and has access to the message register of
$P^*$.\footnote{This model is already considered by \cite{Unr12} in his work of
quantum proofs of knowledge for \NP{}.}
The
protocol is said to be a Proof of Quantum Knowledge if $K$ outputs, with
non-negligible probability, some quantum state $\rho$ that  would make $V_x$
accept with probability at least $q(\eps,n)$, where $q$ is known as the quality function,
or aborts otherwise.

The difficulty in showing that our $\Xi$ protocols are PoQs lies in the fact
that any measurement performed by the
extractor disturbs the state held by $P^*$, and therefore
when we rewind $P^*$ by applying the inverse of his operation, we do not come back to the original state. We
overcome this difficulty in the following way.  We set $\kappa$ to be some value
very close to $1$, namely $\kappa = 1 - \frac{1}{p(n)}$ for some  large enough
polynomial $p$. Our extractor starts by simulating $P^*$ on the first message of
the $\Xi$ protocol, and then holds the (supposed) one-time-padded state and the
commitments to the one-time-pad keys.
$K$ follows by iterating over all possible challenges of the
$\Xi$ protocol, runs $P^*$ on this challenge,  perform the verifier's check and
then rewinds~$P^*$. By the assumption that $P^*$ has a very high acceptance
probability, the measurements performed by $K$ do not disturb the state too
much, and in this case, $K$ can retrieve the correct one-time pads for every
qubit of the witness. If $K$ is successful (meaning that $k$ is able to open every
committed bit), then $K$ can decode the original one-time-padded state
and it is a good witness for $V_x$ with high probability.

We then analyse the sequential repetition of the protocol, that allows us to
have a PoQ with {\em exponentially small} knowledge error $\kappa$, and
extracts one good witness from  $P^*$ (out of the polynomially many copies that
$P^*$ should have in order to cause the verifier to accepted in the multiple runs of the protocol).

\myparagraph{Non-Interactive zero knowledge proof for QMA in the secret parameter model}

Finally, in \Cref{sec:NIZK}, we achieve our non-interactive statistical
zero-knowledge protocol for \QMA{} in the secret parameter setting using
 techniques similar to our $\Xi$ protocol: the trusted party
chooses the one-time pad key and a random (and small) subset of these values
that are reported to the verifier. Since the prover does not know which are the
values that were given to the verifier, he should act as in the
$\Xi$-protocol, but now the verifier does not actually need to ask for the
openings, since the trusted dealer has already sent them.
Although this is a less natural model, we hope that this result will shed some
light in developing $\QNIZK$ proofs for $\QMA$ in more commonly-used models.

\subsection{Open problems}

\myparagraph{Further $\QMA$-complete languages} We note that a number of
problems are currently known to be $\QMA$-complete under Turing reductions,
including the $N$-representability~\cite{LCV07}~\footnote{
  In \cite{LCV07}, the authors reduce the Local Hamiltonian problem on qubits
  into the Local Hamiltonian problem on fermions, and then they propose a Turing
  reduction from LH on fermions to the $N$-representability problem. The missing
  step is reducing CLDM directly to the $N$-representability problem, which
  might be straightforward, but needs a formal proof.  }
and bosonic
$N$-representability problems~\cite{WMN10} as well as the universal functional
of density function theory (DFT)~\cite{SV09}. It is an open question if these
problems can be shown to be $\QMA$-complete under Karp reductions using
the techniques presented in our work.

\myparagraph{Complexity of $k$ CLDM for $k < 5$} We prove in this work that
$5$-CLDM
is \QMA{}-hard under Karp reductions. We leave as an open problem proving if the
problem is still \QMA{}-complete for $k < 5$.

\myparagraph{Marginal reconstruction problem}
We remark that the classical version of CLDM is defined as follows: given the
description of $m$ marginal distributions on sets of bits $C_1,\ldots,C_m$, such
that $|C_i| \leq k$, decide if there is a probability distribution that is close
to those marginals, or such a distribution does not exist.  This problem was
proven $\NP$-complete by Pitowsky~\cite{Pit91}, and its containment
in $\NP$ is proved  by using the fact that such distribution can be seen as
a point $p$ in the {\em correlation polytope} in a polynomial-size Hilbert space. In
this case, by Caratheodory's theorem, $p$ is a convex combination of
polynomially many vertices of such polytope, and therefore these
vertices serve as the $\NP$-proof and a linear program verifies if there is a
convex combination of them that is consistent with the marginals of the
problem's instance.

The difference here is
that the proof and the marginals are different (but connected) objects.
We leave as an open
problem if we can extract a  notion of a locally simulatable classical proof from this (or any other)
problem, and its applications to cryptography and complexity theory. In
particular, we wonder if there is a natural zero-knowledge protocol for this problem.

\myparagraph{Applications of quantum ZK protocols} In classical cryptography, ZK
and PoK protocols are a fundamental primitive since they are crucial ingredients in
a plethora of applications. We discussed in  \Cref{sec:results}
that our quantum
ZK PoQ for \QMA{} could be used as a proof-of-work for quantum
computations. An interesting open problem
is finding other settings in which the benefits of our simple  ZK protocols for \QMA{}
can be applied. We list now some
possibilities that could be explored in future work:
authentication with uncloneable credentials~\cite{CDS94};
proof of quantum ownership~\cite{BJM19iacr}; or
ZK PoQ verification for  quantum money~\cite{AC12}.

\myparagraph{Practical ZK protocols for QMA} Even if we reach a conceptually much simpler ZK protocol
for QMA, the resources needed for it are still very far
from practical. We leave as an open problem if one could devise other protocols
that are more feasible from a physical implementation viewpoint, which could
include classical communication protocols based on the protocols proposed by
Vidick and Zhang~\cite{VZ19arxiv}, or
device-independent ones based on the ideas of Grilo~\cite{Gri19}.

\myparagraph{Non-interactive Zero-knowledge protocols for $\QMA$ in the CRS model} In this work, we propose a
QNIZK protocol where the information provided by the trusted dealer is
asymmetric. We leave as an open problem if one could devise a protocol where the
dealer distributes a common reference string (CRS)(or shared EPR pairs) to the prover
and the verifier.

A possible way of achieving such non-interactive protocol would be to explore
the properties of $\Xi$-protocols, as done classically with $\Sigma$-protocols.
For instance, the well-known Fiat-Shamir transformation~\cite{FS87}  allows us
to make $\Sigma$-protocols non-interactive (in the Random Oracle model). We
wonder if there is a version of this theorem when the first message can be
quantum.

\myparagraph{Witness indistinguishable/hiding protocols for $\QMA$}
Classically, there are two weaker notions that can substitute for ZK in
different applications. In Witness Indistinguishable (WI) proofs, we require that the
verifier cannot distinguish if she is interacting with a prover holding a witness
$w_1$ or~$w_2$, for any $w_1 \ne w_2$. In Witness Hiding (WH), we require that the
verifier is not able to cook-up a witness for the input herself. We note that
zero-knowledge
implies both  such definitions, and we leave as an open problem
finding WI/WH protocols for $\QMA$ with more desirable properties than the known ZK
protocols.

\myparagraph{Computational Zero-Knowledge proofs vs.~Statistical Zero-Knowledge
arguments}
Classically, it is known that the class of problems
with computational ZK proofs is closely related to the class of problems with statistical
ZK arguments~\cite{OV07}. We wonder if this relation is also true in the quantum setting.

\subsection{Concurrent and subsequent works}
Concurrently to this work, Bitansky and Shmueli~\cite{BS19}  proposed
the first quantum zero-knowledge argument system  for \QMA{} with constant
rounds and negligible soundness. Their main building block is a {\em non
black-box quantum extractor} for a post-quantum commitment scheme.

Also concurrently to this work,
Coladangelo, Vidick and Zhang~\cite{CVZ19} proposed a non-interactive argument of quantum knowledge for \QMA{}
with a quantum setup phase (that is independent of the witness) and a classical online phase.
Subsequently to our work, Alagic, Childs, Grilo and Hung~\cite{ACGH20} proposed
the first non-interactive zero-knowledge argument system for \QMA{} where the
communication is purely {\em classical}. Their protocol works in the random
oracle model with setup.

All of these protocols achieve only computational soundness and their security
relies on stronger cryptographic assumptions
(namely the Learning with Errors assumption, as well as post-quantum fully homomorphic encryption). Their proof structure follows by combining
powerful cryptographic constructions based on these primitives to achieve their
results.
On the
other hand, the key idea in our protocols is to use structural properties of
\QMA{} and with that, we can achieve a zero-knowledge protocol with statistical
soundness under the very
weak assumption that  post-quantum
one-way functions exist.

\subsection{Differences with previous version}

In a previous version of this work\footnote{available at \url{https://arxiv.org/abs/1911.07782v1}}, we used the results of \cite{GSY19} almost in a black-box way. In contrast, in the current version,  we provide a new proof for the technical results that we need from~\cite{GSY19}; this  not only makes this work self-contained, but it also provides a conceptually much simpler proof. More concretely, the sketch that was presented in Appendix~A of the previous version has now become a full proof in \Cref{sec:simulatable-history}.
We note that this also allowed us to find a small bug in the proof of~\cite{GSY19}, and provide a relatively easy fix (see~\Cref{R:bug}).

Fermi Ma pointed out a bug in the proof sketch of a proposal of statistical
zero-knowledge argument for \QMA{} that we have in previous versions of this
paper.
\subsection{Structure}

The remainder of this document is structured as follows:
\Cref{sec:preliminaries} presents Preliminaries and Notation. In
\Cref{sec:consistency}, we prove our results on the CLDM problem  and we present our framework of simulatable proofs, with the technical portion of this contribution appearing in \Cref{sec:simulatable-history}.
\Cref{sec:xizk-protocol} establishes the zero-knowledge~$\Xi$ protocol
for~$\QMA$, while in \Cref{sec:PoQ}, we define a proof of quantum knowledge and show that the interactive proof system
satisfies the definition. Finally, in \Cref{sec:qniszk}, we show a
non-interactive zero-knowledge proof for~\QMA in the secret parameter model.

\subsection*{Acknowledgements}
We thank Dorit Aharonov, Thomas Vidick, and the anonymous reviewers for help in improving the presentation of this work.
We thank Andrea Coladangelo, Thomas Vidick and Tina Zhang for discussions on
the definition of proofs of quantum knowledge. A.G.~thanks Christian Majenz for
discussions on a suitable title for this work. We thank Fermi Ma for pointing
out a bug in the proof sketch of the statistical zero-knowledge argument result
which was withdrawn.
A.B.~is supported by the U.S. Air Force Office of
Scientific Research under award number FA9550-17-1-0083,  Canada's  NFRF and NSERC, an Ontario ERA, and the University of Ottawa's Research Chairs program.

\fi

\ifQIP
\else
\section{Preliminaries}
\label{sec:preliminaries}

\label{sec:prelim}

\subsection{Notation}
  For $n \in \mathbb{N}$, we define  $[n] := \{0, \ldots , n-1\}$. For some finite
  set $S$, we denote $s \inr S$ as an element $s$ picked uniformly at random
  from $S$. We say that a
  function $f$ is negligible ($f(n) = \negl(n)$), if for every constant $c$, we
  have
  $f(n) = o\left(\frac{1}{n^c}\right)$. Given two discrete probability distributions $P$
  and $Q$ over the
  domain $\mathcal{X}$, we define its statistical distance as $d(P,Q) = \sum_{x
  \in \mathcal{X}} |P(x) - Q(x)|$.

\subsection{Quantum computation}
We assume familiarity with  quantum computation, and refer to~\cite{NC00} for
the definition of basic concepts such as qubits, quantum states (pure and mixed), unitary
operators, quantum circuits and quantum channels.

For an $n$-qubit state $\rho$ and an $m$-qubit state $\sigma$, we define
$\rho^{\reg{S}} \otimes \sigma^{\reg{\overline{S}}}$ as the $m + n$-qubit
quantum state~$\tau$ that consists of the tensor product of $\rho$ and $\sigma$
where the qubits of~$\rho$ are in the positions indicated by $S \subseteq [m+n]$,
$|S| = n$, and the qubits of $\sigma$ are in the positions indicated by
$\overline{S}$, with the ordering of the qubits consistent with the ordering in
$\rho$ and $\sigma$, as well as the integer ordering in $S$ and $\overline{S}$.
We extend this notation  and write $A^{\reg{S}} \otimes
B^{\reg{\overline{S}}}$ for operators
$A$ and $B$ acting on $|S|$ and $|\overline{S}|$ qubits, respectively.

We define quantum gates with sans-serif font ($\X, \Z$,\ldots), and we define
$\Id$, $\X$, $\Y$ and $\Z$ to be the Pauli matrices, $\mathcal{P}_k=
\{\Id,\X,\Y,\Z\}^{\otimes k}$, $\Had$ to be the Hadamard gate,
$\CNOT$ to be the controlled-Not gate, $\Pg = \left(\begin{matrix} 1 & 0 \\ 0
  &
i \end{matrix}\right)$ and $\T = \left(\begin{matrix}  1 & 0 \\ 0 &
e^{\frac{i \pi}{4}} \end{matrix}\right)$. A Clifford circuit is a quantum circuit
composed of Clifford gates: $\Id$, $\X$, $\Y$, $\Had$, $\CNOT$ and $\Pg$. It is well-known
that universal quantum computation can be achieved with Clifford and $\T$ gates.

For an operator $A$, the trace norm is $\trNorm{A} := \tr{\sqrt{A^\dag A}}$,
which is the sum of the singular values of $A$. For two quantum states $\rho$
and $\sigma$, the trace distance between them is \[D(\rho,\sigma) =
\frac{1}{2} \trNorm{\rho - \sigma} = \max_P \tr{P(\rho-\sigma)}, \] where the
maximization is taken over all possible projectors~$P$.

  If $D(\rho,
  \sigma) \leq \eps$, we say that $\rho$
  and $\sigma$ are $\eps$-close. If $\eps = \negl(n)$, then we say that
  $\rho$ and $\sigma$ are statistically indistinguishable, and we write $\rho
  \approx_s \sigma$.

  If for every polynomial time algorithm $\mathcal{A}$, we have that
  \[\left| \Pr[\mathcal{A}(\rho) = 1] - \Pr[\mathcal{A}(\sigma) = 1] \right| \leq
  \negl(n),\]
  then we say that $\rho$ and  $\sigma$ are
  computationally indistinguishable, and we write $\rho \approx_c \sigma$.

  For some $S \subseteq \01^*$,  let $\{\Psi_x\}_{x \in S}$ and $\{\Phi_x\}_{x
  \in S}$ be two families of quantum channels from
  $q(|x|)$ qubits to $r(|x|)$ qubits, for some polynomials $q$ and $r$. We say
  that these two families are {\em computationally indistinguishable}, and
  denote it by $\Psi_x \approx_c \Phi_x$, if for
  every $x \in S$ and  polynomial $s$ and $k$  and every state $\sigma$ on $q(|x|) + k(|x|)$
  and every polynomial-size circuit acting on $r(|x|) + k(|x|)$ qubits, it
  follows that
  \[\left|\Pr[Q((\Psi_x \otimes \Id)(\sigma)) = 1] -
  \Pr[Q((\Phi_x \otimes \Id)(\sigma)) = 1]
  \right| \leq \negl(n).\]

Finally, we state a result on rewinding by \cite{Wat09b}.

\begin{lemma}[Lemma $9$ of \cite{Wat09b}]\label{lem:rewinding}
  Let $Q$ be an quantum circuit that acts on an $n$-qubit state~$\ket{\psi}$ and $m$ auxiliary systems $\ket{0}$. Let
  \[ p(\psi) = \norm{(\bra{0} \otimes
  I)Q(\ket{\psi}\otimes \ket{0}^{\otimes m})}^2 \text{  and  } \ket{\phi(\psi)} =
  \frac{1}{\sqrt{p(\psi)}}(\bra{0} \otimes
  I)Q(\ket{\psi}\otimes \ket{0}^{\otimes m}). \]
  Let
  $p_0, q\in (0,1)$ and $\eps \in (0,\frac{1}{2})$ such that
  $i)$  $|p(\psi) - q| < \eps$, $p_0(1-p_0) \leq 1(1-q)$, and $p_0 \leq p(\psi)$.
  Then there is a quantum circuit $R$ of size at most
  \[O\left(\frac{log(1/\eps)size(Q)}{p_0(1-p_0)}\right),\]
  such that  on input $\ket{\psi}$, $R$ computes the quantum state
  $\rho(\psi)$ that
  satisfies
  \[\bra{\phi(\psi)}\rho(\psi)\ket{\phi(\psi)} \geq 1- 16 \eps
  \frac{\log^2\frac{1}{\eps}}{p_0^2(1-p_0)^2}.\]
\end{lemma}

\subsection{Complexity classes}
In this section, we define several complexity classes that are considered in
this work.

\begin{definition}[\QMA]
A promise problem $A=(\ayes,\ano)$ is in \class{QMA}  if there exist
polynomials $p$, $q$ and a polynomial-time uniform family of quantum circuits
$\set{Q_n}$, where $Q_n$ takes as input a string $x\in\Sigma^*$ with
$\abs{x}=n$, a $p(n)$-qubit quantum state $\ket{\psi}$, and $q(n)$ auxiliary qubits in state
$\ket{0}^{\otimes q(n)}$, such that:
  \begin{description}
    \item[] \textbf{Completeness:}  If $x\in\ayes$, there exists some $\ket{\psi}$
      such that $Q_n$ accepts $(x,\ket{\psi})$ with probability at least
      $1-\negl(n)$.
    \item[] \textbf{Soundness:} If $x\in\ano$, for any state $\ket{\psi}$, $Q_n$
      accepts $(x,\ket{\psi})$ with probability at most $\negl(n)$.
  \end{description}
\end{definition}

We say that a witness for $x$ is  {\em good} if it makes the verification algorithm accept with probability $1 - \negl(|x|)$.

We define now a quantum interactive protocol between two parties.

\begin{definition}[Quantum interactive protocol between $A$ and $B$ $(A
  \leftrightarrows B)$]
  Let $\mathtt{A}$ and $\mathtt{B}$  be the private registers of parties $A$ and
  $B$, respectively, and $\mathtt{M}$ be the message register.
A quantum interactive protocol between $A$ and $B$ is a sequence of unitaries $U_0,\ldots,U_m$
  where $U_{i}$ acts on registers $\mathtt{A}$ and $\mathtt{M}$ for even $i$,
  and on registers $\mathtt{B}$ and $\mathtt{M}$ for odd $i$. The size of the
  register $\mathtt{A}, \mathtt{B}$ and $\mathtt{M}$, the number $m$ of messages
  and the complexity of the allowed $U_i$ is defined by each
  instance of the protocol. We can also consider interactive protocols where $A$
  outputs some value after interacting with $B$, and we also denote such
  output as $(A \leftrightarrows B)$.
\end{definition}

\begin{definition}
  A promise problem $A=(\ayes,\ano)$ is in $\QZK$ if there is an interactive
  protocol $(V \leftrightarrows P)$, where $V$ is polynomial-time and is given
  some input $x \in A$ and outputs a classical bit indicating acceptance or
  rejection of $x$, $P$ is unbounded, and the following holds
  \begin{description}
    \item[] \textbf{Completeness:}
      { If $x\in\ayes$,
      $\Pr [ (V \leftrightarrows P) = 1] \geq 1 - \negl(n)$.}
    \item[] \textbf{Soundness:} {If $x\in\ano$, for all $P^*$, we have that
      $\Pr [(V \leftrightarrows P^*) = 1)] \leq \frac{1}{\poly(n)}$.}
    \item[]  \textbf{Computational zero-knowledge:} For any $x \in \ayes$ and
      any polynomial-time $V'$ that receives the inputs $x$ and some state
      $\zeta$, there
      exists a polynomial-time quantum channel $\mathcal{S}_{V'}$ that also
      receives~$x$ and $\zeta$ as input such that
      $(V' \leftrightarrows P) \approx_c \mathcal{S}_{V'}$.
  \end{description}
\end{definition}

Following the result by Pass and Shelat~\cite{PS05}, we define now the notion of
non-interactive zero-knowledge proofs in the {\em secret} parameters
model.

\begin{definition}[Non-interactive statistical zero-knowledge proofs in the secret parameter model]
  A triple of algorithms $(D,P,V)$ is a non-interactive statistical
  zero-knowledge proof in
  the secret parameter model for a promise problem $A =
  (A_{yes},A_{no})$ where $D$ is a probabilistic polynomial time algorithm, $V$
  is a probabilistic polynomial time algorithm and $P$ is an unbounded algorithm such that there exists a negligible
  function $\eps$ such that the following conditions follow:
  \begin{description}
    \item[] \textbf{Completeness:} for every $x \in A_{yes}$, there exists some $P$
      \[Pr[(r_P, r_V) \leftarrow D(1^{|x|}); \pi \leftarrow P(x,r_P); V(x,r_V,\pi) =
      1] \geq 1 - \eps(n).\]
    \item[] \textbf{Soundness:} for every $x \in A_{no}$ and every $P$
      \[Pr[(r_P, r_V) \leftarrow D(1^{|x|}); \pi \leftarrow P(x,r_P); V(x,r_V,\pi) =
      1] \leq \eps(n).\]
    \item[] \textbf{Statistical zero-knowledge:} there is a probabilistic polynomial time
      algorithm $\mathcal{S}$ such that for every $x \in A_{yes}$,
      the statistical distance of the distribution of  the output
      of $\mathcal{S}(x)$ and the distribution of $(r_V,\pi)$ for
     $(r_P, r_V) \leftarrow D(1^{|x|})$ and $ \pi \leftarrow P(x,r_P)$ is
      $\negl(n)$.
  \end{description}
\end{definition}

\subsection{Local Hamiltonian problem}
\label{sec:circuit-to-ham}

We discuss now the Local Hamiltonian problem, the quantum analog of MAX-SAT problem.

\begin{definition}
  \label{def:local Hamiltonian}
  For
  $k \in \mathbb{N}$, $a,b \in \real$ with $a < b$, the $k$-\emph{Local Hamiltonian} problem with parameters $a$ and $b$
  is the following promise problem.
  Let $n$ be the number of qubits of a quantum system.
  The input is a set of $m(n)$ Hamiltonians $H_0, \ldots, H_{m(n)-1}$
  where $m$ is a polynomial in $n$, for all $i$ we have that $\norm{H_i} \leq 1 $
  and each $H_i$ acts on $k$ qubits out of the $n$ qubit system.
  For $H = \sum_{j = 1}^{m(n)} H_j$ the following two conditions hold.
    \begin{itemize}
\item[\textbf{Yes.}] There exists a state $n$-qubit state $\ket{\varphi}$ such that
      $
        \bra{\varphi} H \ket{\varphi}
        \leq a \cdot m(n) .
      $
\item[\textbf{Yes.}] For every $n$-qubit state $\ket{\varphi}$
      it holds that
      $
        \bra{\varphi} H \ket{\varphi}
        \geq b \cdot m(n) .
      $
    \end{itemize}
\end{definition}

In the proof of containment in \QMA{} for the $k$-Local Hamiltonian problem for $b -a > \frac{1}{\poly(n)}$, Kitaev showed how to estimate the energy of a local term. The verification procedure consists of picking one term  uniformly at random, and then measuring the energy of the that term. Notice that this verification procedure can be seen as $m$ POVMs, each one acting on at most a $k$-qubit system. We record this in the following lemma.

\begin{lemma}\label{L:local-verification}
  There exists a verification algorithm for the $k$-Local Hamiltonian problem with parameters $a,b \in \mathbb{R}$, $b - a \geq \frac{1}{\poly(n)}$, consisting of picking one of $m = \poly(n)$ $k$-qubit POVMs $\{\Pi_1, I-\Pi_1\}, \ldots, \{\Pi_{m(|x|)}, I -\Pi_{m(|x|)}  \}$ and accepting if and only if the witness projects onto $\Pi_i$.
\end{lemma}

In Kitaev's proof of \QMA{}-hardness of the Local Hamiltonian problem, he uses
the circuit-to-Hamiltonian construction proposed by Feynman~\cite{Fey82} in order to reduce
arbitrary \QMA{} verification procedures to time-independent Hamiltonians in a
way that the Local Hamiltonian has low-energy states if and only if the \QMA{} verification
accepts with high probability.

More concretely, Kitaev shows a reduction from a quantum circuit $V$
consisting of $T$ gates $U_1,\ldots,U_T$ acting on
a $p$-qubit state $\ket{\psi}$ provided by the prover and an auxiliary register
$\ket{0}^{\otimes q}$ to some Hamiltonian $H_V = \sum_{i \in [m]} H_i$, where the terms $H_1,\ldots,H_m$ act on $T + p + q$ qubits,
 and they range between the following types:
  \begin{description}
    \item[clock consistency]
    $H^{clock}_t =  \kb{01}_{t,t+1}$, for $0 \leq t \leq T-1$
    \item[initialization]  For $j \in [q]$, $H^{init}_j = \kb{0}_{0} \otimes
      \kb{1}_{T+p+j}$,
    \item[propagation] Let $J_t$ be the set of qubits on which $U_t$ acts non-trivially.
       ~\\$H^{prop}_0 = \frac{1}{2} \Big(\kb{0}_{0} +
    \kb{10}_{0,1} - \ketbra{1}{0}_{0} \otimes
    (U_0)_{J_0}
    - \ketbra{0}{1})_c(0) \otimes
    (U_t^\dagger)_{J_0} \Big)$
  ~\\ $H^{prop}_T = \frac{1}{2} \Big(\kb{10}_{T-1,T} +
    \kb{1}_{T} - \ketbra{1}{0}_{T} \otimes
    (U_0)_{J_T}
    - \ketbra{0}{1})_c(T) \otimes
    (U_t^\dagger)_{J_T} \Big)$
       \begin{align*}
         \text{For } 1 \leq t \leq T-1,
         H^{prop}_t = \frac{1}{2} \Big(&\kb{100}_{t-1,t,t+1} +
    \kb{110}_{c(t-1,t,t+1)} \\
         &- \ketbra{110}{100}_{t-1,t,t+1} \otimes
    (U_t)_{J_t}
    - \ketbra{100}{110}_{t-1,t,t+1} \otimes
    (U_t^\dagger)_{J_t} \Big)
       \end{align*}
    \item[output] $H^{out} = \kb{1}_{T} \otimes \kb{0}_{T+1}$
  \end{description}

The facts that we need from this reduction are summarized in the following~lemma.
\begin{lemma}[\cite{KSV02}]\label{lem:kitaev}
  If there exists some state $\ket{\psi}$ that makes $V$ accept with probability
  $1 - \negl(n)$, then the {\em history state}
  \[
    \frac{1}{\sqrt{T+1}}\sum_{t \in [T+1]} \ket{\unary(t)} \otimes
    U_t\ldots U_1(\ket{\psi}\ket{0}^{\otimes q}),
  \]
  has energy $\negl(n)$ according to $H_V$.
  If every quantum state $\ket{\psi}$ makes $V$ reject with probability at
  least~$\eps$, then the groundenergy of $H_V$ is at least
  $\Omega\left(\frac{\eps}{T^3}\right)$.
\end{lemma}

\subsection{Quantum error correcting codes and fault-tolerant quantum computing}
\label{sec:codes}
For $n > k$, an $[[N,K]]$ quantum error correcting code (QECC) is a mapping
from a  $K$-qubit state~$\ket{\psi}$ into an $N$-qubit
state $\Enc(\ket{\psi})$.  The distance of an $[[N,K]]$ QECC is $D$
if for an arbitrary quantum operation $\mathcal{E}$ acting on
$(D-1)/2$ qubits, the original state $\ket{\psi}$ can be recovered from
$\mathcal{E}(Enc(\ket{\psi}))$, and in this case we call it an $[[N,K,D]]$~QECC.
For a $[[N,1,D]]$-QECC and its encoding $\Enc$, we overload notation and for an $k$-qubit system $\phi$ we write  $\Enc(\phi) := \Enc^{\otimes k}(\phi)$.

Given a fixed QECC, we say that some $k$-qubit gate $U$ can be applied \emph{transversally}, if
we apply $U^{\otimes n} Enc(\ket{\psi}) = Enc(U\ket{\psi})$ where $\ket{\psi}$ is a $k$-qubit state and the $i$-th tensor of $U$ acts on the $i$-th qubit of the encodings of each qubit of $\ket{\psi}$.
It is known that no code admits a universal set of
transversal gates~\cite{EK09}. In order to overcome this difficulty to
achieve fault-tolerant computation, one can use tools from computation by
teleportation. In this case,
provided a resource called {\em  magic states}, one can simulate the non-transversal
gate by applying some procedure on the target qubit and the magic state and such
procedure contains only (classically controlled) transversal gates.

In this work, we consider the $k$-fold concatenation of the Steane code. We simply state the properties we need from it in this work and we refer to \cite{GSY19} for further details.
The $k$-fold concatenation of the Steane code is $[[7^k,1,3^k]]$ QECC such that
all Clifford operations can be performed transversally and the $\mathsf{T}$
gates can be applied using the magic state $\mathsf{T}\ket{+}$.

\subsection{Commitment schemes}

A commitment scheme is a two-phase  protocol between two parties, the Sender and
Receiver. In the first phase, the Sender, who holds some message $m$, unknown by
the Receiver, sends a {\em commitment } $c$ to the Receiver. In the second
phase, the Receiver will reveal the committed value $m$. We require two
properties of the protocol: hiding, meaning that the Receiver cannot guess $m$
from $c$, and binding, which stays that the Receiver cannot decide to open
value $m' \ne m$ when $c$ was committed for $m$.

It is well-known that there are no commitment schemes with unconditionally hiding
\emph{and} unconditionally binding properties, but such schemes can be achieved
if either of the properties
holds only {\em computationally}.\footnote{Where {\em unconditionally} means
that the property is guaranteed even against unbounded adversaries, in contrast
to {\em computational} case, when the property is only guaranteed against quantum
polynomial-time adversaries.}

In this work, we consider commitment schemes which
are
unconditionally
binding but computationally hiding.
There are recent quantum-secure instantiations of these schemes, assuming
the hardness of Learning Parity
with Noise~\cite{JKPT12,XXW13,BKLP15}.

We present now the formal definition of such schemes.

\begin{definition}[Computationally hiding and unconditionally binding commitment
  schemes]\label{def:comp-hiding-commitment}
  Let~$\eta$ be some security parameter, $p$ be some polynomial, and
  $\mathcal{M}, \mathcal{C}, \mathcal{D} \subseteq \01^{p(\eta)}$ be the message space, commitment space and opening space,
  respectively.
A computationally hiding and unconditional binding commitment
  scheme consists of a pair of algorithms $(\textsf{commit}, \textsf{verify})$,
  where
  \begin{itemize}
    \item $\textsf{commit}$ takes as input a value from~$\mathcal{M}$ and some value in
      $\mathcal{C} \times \mathcal{D}$,
    \item $\textsf{verify}$ takes as input a value from
      $\mathcal{M} \times \mathcal{C} \times \mathcal{D}$ and outputs some value in
      $\01$
  \end{itemize}
  with the following properties:
  \begin{description}
    \item[] \textbf{Correctness:} If $(c,d) \leftarrow \textsf{commit}(m)$, then
      $\textsf{verify}(m,c,d) = 1$.
    \item[] \textbf{Computationally hiding:} For any polynomial-time quantum adversary
      $\mathcal{A}$ and  $(c,d) \leftarrow
      \textsf{commit}(m)$
      \[\Pr[ \mathcal{A}(c) = m]  \leq \frac{1}{|\mathcal{M}|} + \negl(\eta).\]
    \item[] \textbf{Unconditionally binding:}
      For any $(m,d)$ and $(m',d')$, it follows that
      \[\textsf{verify}(m,c,d) = \textsf{verify}(m',c,d') = 1 \implies m = m'.      \] \end{description}
\end{definition}

\section{Consistency of local density matrices and locally simulatable proofs}
\label{sec:consistency}
In this section,  prove that the CLDM problem is in \QMA{} (\Cref{sec:CLDM-in-QMA}), and we present our Theorem on  the \QMA{}-hardness of the CLDM problem (\Cref{sec:cldm})\footnote{Note that, together with the technical details in \Cref{sec:simulatable-history}, our proof is self-contained}.

One drawback of the containment of CLDM in \QMA{} is that the verification procedure must check a super-constant number of qubits of the witness. Notice that for the local-Hamiltonian problem,
checking a constant number of qubits is sufficient for an inverse polynomial completeness/soundness gap, but then we do not have the full knowledge of reduced density matrices of the witness as in CLDM. Our techniques allow us to define a new object called Locally Simulatable proofs, where we have the best of both worlds: full knowledge of the reduced density matrices of the witness and local verification.  We present our framework of Locally Simulatable proofs in \Cref{sec:simulatable-proof} (see \Cref{subsection:intro-locally-simulatable-proofs} for an overview on how these locally simulatable proofs are used and \Cref{sec:xizk-protocol,sec:NIZK} for  details).

Let us start by formally defining the CLDM problem.

\begin{definition}[Consistency of local density matrices problem (CLDM)~\cite{Liu06}]
Let $n \in \mathbb{N}$.  The input to the consistency of local density matrices
  problem consists of $((C_1,\rho_1),\ldots,(C_m,\rho_m))$  where
 $C_i \subseteq [n]$ and $|C_i| \leq k$; and
  $\rho_i$ is a density matrix on $|C_i|$ qubits and
  each
  matrix entry of $\rho_i$ has precision $poly(n)$. Given two parameters
  $\alpha$ and $\beta$, assuming that one of the following conditions is true,
  we have to decide which of them holds.
  \begin{itemize}
    \item[\textbf{Yes.}] There exists some $n$-qubit quantum state $\tau$ such
      that for every $i \in [m]$, $\trNorm{\Tr_{\overline{C_i}}(\tau) - \rho_i} \leq
      \alpha$.
    \item[\textbf{No.}]
      For every  $n$-qubit quantum state $\tau$, there exists some $i \in [m]$ such
      that $\trNorm{\Tr_{\overline{C_i}}(\tau) - \rho_i} \geq \beta$.
  \end{itemize}
\end{definition}
\begin{remark}
  Note that in \cite{Liu06}, the definition of the problem sets $\alpha =0$.
  In our case, we define the problem more generally, otherwise we would only
  achieve $\QMA_1$ hardness (the version of \QMA{} with perfect completeness)
  rather than $\QMA$-hardness.
\end{remark}

\subsection{Consistency of local density matrices is in \QMA.}
\label{sec:CLDM-in-QMA}

In the proof of containment of CLDM in $\QMA$, Liu uses a characterization of $\QMA$
called $\QMA+$~\cite{AR03}. For completeness, we start by showing
the containment of
CLDM in $\QMA$, by presenting a standard verifier for the problem, which is
a straightforward composition of the results from~\cite{AR03} and~\cite{Liu06}.

\begin{lemma}\label{lem:containment-qma}
  The consistency of local density matrices problem is in $\QMA$
  for any $k = O(\log n)$, and \mbox{$\alpha, \beta$} such that $\eps :=
  \frac{\beta}{4^k}-\alpha \geq \frac{1}{poly(n)}$.
\end{lemma}
\begin{proof}
  Let $((C_1,\rho_1),\ldots,(C_m,\rho_m))$ be an instance for CLDM.
  Let  $p$ be a polynomial such that $p(n)\eps^2 = \Omega(n)$, the verification
  system expects some state $\psi$ consisting of $p(n)$ copies of the state
  $\tau$ that is supposed to be consistent with the local density matrices. The
  verifier then picks $i \in [m]$ and $P \in \mathcal{P}_{|C_i|}$ uniformly at
  random. The verifier then measures each (supposed) one of the $p(n)$ copies
  according to the observable $P^{\reg{C_i}} \otimes I^{\reg{\overline{C_i}}}$, and let $\hat{p}$ be the
  average of its outcomes.
  The verifier accepts
  if and only if $\left|\hat{p} - \Tr(P \rho_i)\right|\leq \alpha + \frac{\eps}{2}$.

  \medskip
  In the completeness case, we have that $\psi = \tau^{\otimes \ell}$ and in this
  case,
   each of the $p(n)$ measurements is $1$ with probability
  $\Tr\left(\left(P^{\reg{C_i}}\otimes I^{\reg{\overline{C_i}}}\right) \tau\right)$ for every $i$.
  By Hoeffding's inequality, we have that with
  probability at least $1-2 \exp(p(n)\eps^2/8) = 1 - \negl(n)$,
   \begin{align*}
     \left|\Tr((P^{\reg{C_i}}\otimes I^{\reg{\overline{C_i}}})\tau) - \hat{p} \right|\leq \frac{\eps}{2}.
    \end{align*}
    Using the fact that $\tau$ is consistent with $\rho_i$, we also have that
    \begin{align*}
   \left|\Tr((P^{\reg{C_i}}\otimes I^{\reg{\overline{C_i}}})\tau) - \Tr(P\rho_i)\right|\leq \alpha,
    \end{align*}
   and therefore by the triangle inequality, the verifier accepts with probability at least
   $1 - \negl(n)$.

   \medskip
   For soundness, let $\psi_j$ be the reduced density
   state considering the register of the $j$th copy of the state. Let $\tau=
   \frac{1}{p(n)} \sum_{j} \psi_j$.
   Since we have a no-instance,  there exists some $i \in [m]$ such
   that
   $\trNorm{\Tr_{\overline{C_i}}(\tau) - \rho_i} \geq \beta$. Let us write
   \[\Tr_{\overline{C_i}}(\tau) - \rho_i = \sum_{P \in
   \mathcal{P}_{|C_i|}} \gamma_P P,\]
   for $\gamma_P = \Tr((P^{\reg{C_i}}\otimes I^{\reg{\overline{C_i}}}) \tau) - \Tr(P\rho_i)$.
   There must be some choice of $P \in \mathcal{P}_{|C_i|}$ such that $|\gamma_P| \geq
   \frac{\beta}{4^{|C_i|}}$.
   Notice that by the definition of $\tau$,
   \[\Tr((P^{\reg{C_i}} \otimes I^{\reg{\overline{C_i}}})\tau) = \frac{1}{p(n)}\sum_{j} \Tr((P^{\reg{C_i}} \otimes I^{\reg{\overline{C_i}}})\psi_j).\]
   In this case, if we measure each register corresponding to the $j$th copy of
   the state, then the expected value of its average is $\Tr((P^{\reg{C_i}} \otimes I^{\reg{\overline{C_i}}})\tau)$. Let
   $\hat{p}$ be the average of the outcomes of the performed measurements.
   Again
   using Hoeffding's inequality, we have that
  with probability at least $1-2 \exp(p(n)\eps^2/8) = 1 - \negl(n)$,
   \begin{align*}
     \left|\Tr((P^{\reg{C_i}} \otimes I^{\reg{\overline{C_i}}})\tau) - \hat{p}\right|\leq \frac{\eps}{2}.
    \end{align*}
   and therefore, we have that for this fixed $i$ and $P$, the prover accepts with
   probability $\negl(n)$. Since such $i$ is picked with probability
   $\frac{1}{m}$ and such a $P$ is picked with probability at least
   $\frac{1}{4^k}$, the overall acceptance is at most $1 -
   O\left(\frac{1}{m4^k}\right)$ (where we account for the negligible factors inside
   the $O$-notation).
\end{proof}

\subsection{Consistency of local density matrices is \QMA{}-hard}
\label{sec:cldm}

We show  now that CLDM is $\QMA$-hard under standard Karp reductions.

\begin{theorem}\label{thm:CLDM-hard}
  The consistency of local density matrices problem is $\QMA$-hard under Karp
  reductions.
\end{theorem}

At a high level, the proof consists in showing
 a verification algorithm for every problem in \QMA{} such that the reduced density matrices of the history state of the verification procedure for yes-instances is simulatable. More precisely, we show that  for any fixed $s$,  there exists a verification algorithm $\VerSimX = U_T \cdots U_1$ such that there is a classical algorithm that outputs the classical description of the reduced density matrix of
 \begin{align}
\label{eq:simulatable-history}
  \Phi = \frac{1}{T+1}  \sum_{t,t' \in [T+1]} \ketbra{\unary(t)}{\unary(t')} \otimes U_t\cdots U_1 \left(\WitSim \otimes \kb{0}^{\otimes A}\right)U_1^{\dagger}\cdots U_{t'}^\dagger,
 \end{align}
on any subset of $s$ qubits in time $\poly(|x|,2^s)$.  Here, $\WitSim$ is some witness that makes $\VerSimX$ accept with probability $1-\negl(|x|)$\footnote{In particular, notice that this can only be true if $x$ is a yes-instance.}. This is formalized in the following Lemma.

\newcommand{\bodygsy}{
  For any problem $A = (\ayes,\ano)$ in $\QMA$ and $s \in \mathbb{N}$, there is a uniform family of verification algorithms $\VerSimX = U_T \cdots U_1$ for $A$ that acts on a witness of size $p(|x|)$ and $q(|x|)$ auxiliary qubits
such that there exists a polynomial-time deterministic algorithm $\Sim$ that on input $x \in A$
and $S \subseteq [T+p+q]$  with $|S| \leq 3s+2$, $\Sim(x,S)$ outputs the
  classical description of an $|S|$-qubit  density matrix $\rho(x,S)$
  with the following properties
  \begin{enumerate}
    \item If $x$ is a yes-instance, then there exists a witness $\WitSim$ that makes $\VerSimX$ accept with probability at least $1 - \negl(n)$ such that
      $\trnorm{\rho(x,S) -
      \Tr_{\overline{S}}(\hist)} \leq \negl(n)$, where
      \[\hist =
  \frac{1}{T+1}  \sum_{t,t' \in [T+1]} \ketbra{\unary(t)}{\unary(t')} \otimes U_t \cdots U_1 \left(\WitSim \otimes \kb{0}^{\otimes q}\right)U_1^{\dagger} \cdots U_{t'}^\dagger,
      \] is the history state of $\VerSimX$ on the witness $\WitSim$.
    \item Let $H_i$ be one term from the circuit-to-local Hamiltonian construction from $\VerSimX$ and $S_i$ be the set of qubits on which $H_i$ acts non-trivially. Then for every $x \in A$,  $\Tr(H_i\rho(x,S_i)) = 0$.
  \end{enumerate}
}

\begin{lemma}[Simulation of history states]\label{lem:gsy}
  \bodygsy
\end{lemma}

\Cref{sec:simulatable-history} is devoted to a self-contained proof of \Cref{lem:gsy}. This proof is inspired by Lemmas $15$, $16$ and $17$ of~\cite{GSY19}. However, here we provide a simpler proof for a simpler statement,\footnote{In \cite{GSY19}, they are simulating history states for $\mathsf{MIP}^*$ computation and therefore they need to deal also with arbitrary
provers' operations} and we also fix a small bug in their proof.

Assuming the above Lemma, we are now ready to prove \Cref{thm:CLDM-hard}.

\begin{proof}[Proof of \Cref{thm:CLDM-hard}]
  Let $A = (\ayes, \ano)$ be a promise problem in $\QMA$ and let
  $\VerSimX$ be the \QMA{} verification circuit for $A$ stated in \Cref{lem:gsy}.

  Given input $x \in A$, our reduction consists of using $\Sim$ to compute the CLDM instance
  \begin{align} \label{eq:density-matrices}
    \{(S,\rho(x,S)) : \text{ $S$ is a subset of the qubits in the history state with } |S| \leq 5 \}.
  \end{align}

  We show now that if $x\in\ayes$, there exists a state that is consistent with all $\rho(x,S)$.
 Let $\WitSim$ and~$\VerSimX$ be given by \Cref{lem:gsy} for $s = 5$ (and thus $\VerSimX$ accepts $\WitSim$ with probability $1 - \negl(n)$).
  Then, also by
\Cref{lem:gsy}, we have that the history state
  $\hist$ of the computation of $\VerSimX$ on $\WitSim$ is consistent with
the reduced density matrices defined in
  \Cref{eq:density-matrices}.

  We show now that if $x \in \ano$, for every state $\tau$, there exists some $S \subseteq [T+p+q]$ with $|S| \leq
5$, such that
  $\trNorm{\Tr_{\overline{S}}(\tau)
  - \rho(x,S)} > \frac{1}{T^5}$.
  Let us assume, by way of contradiction, that there exists a $\tau$
  such that for every $S \subseteq [T+p+q]$ with $|S| \leq 5$, $\trNorm{\Tr_{\overline{S}}(\tau) - \rho(x,S)} \leq \frac{1}{T^5}$.
   Let~$H_{\VerSimX}$ be the $5$-local Hamiltonian resulting from the circuit-to-Hamiltonian construction on circuit $\VerSimX$.
We show that in this case,
      $\tau$ has energy $O\left(\frac{1}{T^4}\right)$ with respect to $H_{\VerSimX}$, which is a
      contradiction, since  $x \in \ano$ and therefore  $H_{\VerSimX}$ has groundenergy
      $\Omega(\frac{1}{T^3})$ by \Cref{lem:kitaev}. This finishes the proof.

Let $S_i$ be the set of at most $5$-qubits on which the $i$-th term of
  $H_{\VerSimX}$ acts and $\rho_i = \rho(x,S_i)$. We have that  the energy of such $\tau$ is at most
\[\Tr(H\tau) = \sum_{i} \Tr(H_i \Tr_{\overline{S_{i}}}(\tau)) \leq \sum_i \left(
\Tr(H_i \rho_i) + \frac{1}{T^5} \right)
\leq O\left(\frac{1}{T^4}\right),\]
where the first inequality
comes from the assumption that
$\trNorm{\Tr_{S_i}(\tau) - \rho_i} < \frac{1}{T^5}$ for all $i$, and
the second inequality follows since there are $O(T)$ terms and from
\Cref{lem:gsy} we  have that \mbox{$\Tr(H_i\rho_i) = 0$}.
\end{proof}

\subsection{Locally simulatable proofs}
\label{sec:simulatable-proof}

As previously mentioned, we present here the framework of \emph{locally simulatable proofs}, which combines in an abstract way the strong points of the local-Hamiltonian problem and CLDM. More concretely,
locally simulatable proofs allow us to perform \QMA{} verification by only checking a constant number of qubits (as in the local-Hamiltonian problem), while having full knowledge of the reduced density matrices of small subsets of the qubits of a good witness (as in CLDM).

\begin{definition}[$k$-$\SimQMA$]\label{def:simulatable-proof}
A promise problem $A=(\ayes,\ano)$ is in the complexity class $k$-$\SimQMA$ with
  soundness $\beta(n) \leq 1- \frac{1}{\poly(n)}$
if there exist
  polynomials $m$, $p$ such that
  given $x \in A$,  there is an efficient deterministic algorithm that computes
  $m(|x|)$
  $k$-qubit POVMs $\{\Pi_1, I-\Pi_1\}, \ldots , \{\Pi_{m(|x|)}, I -\Pi_{m(|x|)}  \}$, that act on some quantum state of size $p(|x|)$, such that:
  \begin{description}
    \item[] \textbf{Simulatable completeness:}  If $x\in\ayes$, there exist
      a $p(|x|)$-qubit
      state $\tau$, that we call a {\em simulatable witness},
      and a set of $k$-qubit density matrices $\{\rho(x,S)\}_{\substack{S \subseteq
  [p(n)] \\ |S| = k}}$ that can be computed in polynomial time from $x$,
      such that for all $c \in [m]$
      \[\Tr(\Pi_c \tau) \geq
      1-\negl(|x|),\]
      and for every $S \subseteq [p(n)]$ of size $k$
      \[\trNorm{\Tr_{\overline{S}}(\tau) - \rho(x,S)} \leq \negl(|x|). \]
    \item[] \textbf{Soundness:} If $x\in\ano$, for any $p(|x|)$-qubit state
      $\tau$ we have that
      \[\frac{1}{m}\sum_{c \in [m]} \Tr(\Pi_c \tau) \leq \beta(|x|).\]
  \end{description}
\end{definition}

\begin{lemma}\label{lem:simulatable-proof}
  Every problem in \QMA{} is in $5$-$\SimQMA$.
\end{lemma}
\begin{proof}[Proof (Sketch).]
  We can consider the POVMs that arise from the verification of the Local
  Hamiltonian problem and the density matrices that are the simulation of the
  history state of the computation.
  From~\Cref{L:local-verification}  and \Cref{lem:gsy},
  the result follows.
\end{proof}

\section{Simulation of history states}
\label{sec:simulatable-history}
In this section, we prove \Cref{lem:gsy}. Technically, we achieve this result using a recent notion
defined by Grilo, Slofstra and Yuen~\cite{GSY19} called
\emph{simulatable codes}.  In a simulatable code,  any $s$-qubit reduced density matrix of
a codeword can be efficiently computed independently of the encoded state. In order to stress this independence of the encoded state, we say that the reduced density matrix is \emph{simulated}.
Furthermore, the reduced density states  of the intermediate steps of the physical computation corresponding to a {\em logical} gate on encoded data (either through transversal
gates or computation by teleportation with magic states) can also be
efficiently computed. We refer to Example~\ref{ex:simulatability} and \Cref{sec:techniques} for a more detailed overview of such a simulation.

\begin{definition}\label{D:simulatable}

  Let $\calC$ be a $[[N,1,D]]$-QECC  that allows universal quantum computation on the encoded data by applying logical gates from a universal gateset $\mathcal{G}$
  with transversal gates (and possibly with the help of magic states).
  Let $G \in \mathcal{G}$ be a logical gate acting on $k_G$
  qubits, $U^{(G)}_1, \ldots ,U^{(G)}_{\ell}$ be the
  transversal circuit that is applied to the physical qubits of the encoding of a $k_G$-qubit state and  logically applies $G$ on the data through $\ell = \poly(N)$ physical gates, with the help of an $m_G$-qubit magic state~$\tau_G$.
 We say that $\calC$ is
    \emph{$s$-simulatable} if there is a deterministic algorithm $\SimC$ that
    receives as input  $G \in \mathcal{G}$,  a value $0 \leq t \leq \ell$ and a
    subset~$S \subseteq [N(m_G + k_G)]$ with $|S| \leq s$ and $\SimC(G,t,S)$
    runs in time $\poly(2^N)$ and outputs the classical
    description~\footnote{The classical description of a density matrix consists
    of each entry of the matrix with $\poly(n)$ bits of precision.} of an $|S|$-qubit density matrix $\rho(G,t,S)$ such that for every $k_G$-qubit state $\sigma$
    \begin{equation*}
      \rho(G,t,S) = \Tr_{\overline{S}}\left((U^{(G)}_t \cdots U^{(G)}_1) \Enc(\sigma\otimes \tau_G)  (U^{(G)}_{t}
      \cdots U^{(G)}_1)^\dagger\right).
    \end{equation*}
\end{definition}
\begin{remark}\label{R:decomposition-gates}
    In this work, we consider mostly QECCs that admit transversal Clifford gates and then use magic state $\ket{\T}$ to compute $\T$-gates. Thus, for concreteness,
    we take that $\mathcal{G} = \{\CNOT,\Pg,\Had,\T\}$. In this case $k_G = 2$ for  $G = \CNOT$ gate and $k_G = 1$ for the other gates, and $\tau_{\T} = \kb{\T}$  and no magic state is needed for the other gates.
  Notice that the  gates $U^{(G)}_1,\ldots ,U^{(G)}_{\ell}$ are publicly known, and therefore they do not need to be a parameter for $\SimC$.
\end{remark}

\begin{remark}\label{R:weaker-simulatable}
    \Cref{D:simulatable} is slightly weaker from the one defined in \cite{GSY19}. There, the runtime of the simulator is $\poly(2^s)$ (whereas ours is $\poly(2^N)$). That is,  the runtime of their simulator depends on the number of qubits whose density matrix is simulated whereas in our case, our simulator depends on the  size of the codeword. However, we notice that in the applications of simulatable codes  in \cite{GSY19} as well as here, QECCs of constant size are considered, and therefore our weaker definition suffices.
    \end{remark}

We call these codes   {\em locally} simulatable codes in order to emphasize
that only small parts of codewords can be simulated. In particular,
in~\cite{GSY19}, it was shown that the concatenated Steane code is locally
simulatable for some~$s$.

\begin{lemma}[\cite{GSY19}]\label{lem:simulatable}
   For every $k >
  \log(s+3)$, the $k$-fold concatenated Steane code is $s$-simulatable.
\end{lemma}

The proof of \Cref{lem:simulatable} in \cite{GSY19} is somewhat involved, going through extensive calculations using the stabilizer formalism for QECCs. Here, as a side contribution of this work (see \Cref{sec:simulatable-codes}), we present a new, simpler, proof for \Cref{lem:simulatable}.\footnote{We remark, however, that our proof does not work for the stronger notion of simulatability that would come with different parameters of \Cref{D:simulatable} as discussed in \Cref{R:weaker-simulatable}. } This not only makes our contribution self-contained, but we believe it  might facilitate the use of such notions in other contexts.
Our proof relies  on the fact that quantum error correcting codes are good for  ``hiding secrets'': if we consider a subset of qubits of a quantum error correcting code {\em smaller} than the number or errors that can be corrected, then this reduced density matrix is independent of the encoded state. We push this observation further and show that this independence of the encoded state also works for intermediate steps of the computation on encoded data.  We describe the detailed connection between QECCs and secret sharing in \Cref{sec:secret-sharing} and provide our new proof of \Cref{lem:simulatable} in \Cref{sec:simulatable-codes} (and for an intuitive description, recall the discussion in \Cref{sec:intro-simulatable}).

In \Cref{sec:proof-gsy}, we prove \Cref{lem:gsy} assuming \Cref{lem:simulatable} (for pedagogical reasons, our proof of \Cref{lem:simulatable} is deferred to a later section, as decribed above).  Very roughly, the proof proceeds by lifting the simulation of QECC to the simulation of the history states of \QMA{} computation. More concretly, we start with a verification algorithm $V_x$ for some problem $A$ in \QMA{} and we transform it into another \QMA{} verification algorithm~$\VerSimX$ for~$A$ with the same completeness and soundness parameters, but which  allows us to perform the simulation. The verification circuit $\VerSimX$ expects an encoded version of the \QMA{} witness $\ket{\psi}$ of $V_x$, where each qubit of $\ket{\psi}$ is encoded using a simulatable code, and then performs $V_x$ on the encoded data using techniques from fault-tolerant quantum computation.
With this new encoded verification circuit for \QMA{}, along with the techniques from \cite{GSY19}, we are able to show how to compute reduced density matrices of the history state of~$\VerSimX$ when a good witness is provided. %
The details of how to construct $\VerSimX$ are presented in \Cref{S:new-verification} and we show in \Cref{S:simulatable-history} that the history state of $\VerSimX$ is efficiently simulatable.

\subsection{QECCs and secret sharing
}\label{sec:secret-sharing}

Here,  we provide a proof of what can be seen as a ``composable''\footnote{By composable, we mean that it also considers the purification of the encoded state and we make no claim regarding the notion of Universal Composability in cryptography.} statement to the fact that quantum error correcting codes can be used for secret sharing~\cite{CGL99}. This is used later in \Cref{sec:proof-gsy,sec:simulatable-codes}.

\begin{lemma}
\label{L:logical-trace}
  Let $\calC$ be an $[[N,1,D]]$-QECC, $A$ be a quantum register (possibly entangled with some environment register $E$), and $A'$ be the register that is the output of encoding a quantum system in register $A$ under $\calC$. Let also $S$ be a subset of size at most $(D-1)/2$ of the qubits in $A'$ and $\overline{S} = A' \setminus S$. Then we have that there exists some state $\tau_{S}$ such that for all $\ket{\psi}_{AE}$
  \[
    \Tr_{\overline{S}}((\Enc \otimes I)(\kb{\psi}_{AE})) = \tau_S \otimes \Tr_{A}(\kb{\psi}_{AE}),
  \]
  where the encoding on the LHS only acts on register $A$.
\end{lemma}
\begin{proof}
  Let us consider the  circuit given in \Cref{Fig:secret-sharing-QECC} that receives $\ket{\psi}_{AE}$, encodes regiser $A$, swaps the qubits of~$S$ with $\ket{0}^{\otimes |S|}$ and finally applies the correction procedure of the QECC with the qubits in~$\overline{S}$ and the fresh auxiliary qubits.
 \begin{figure}
  \begin{center}
    \includegraphics{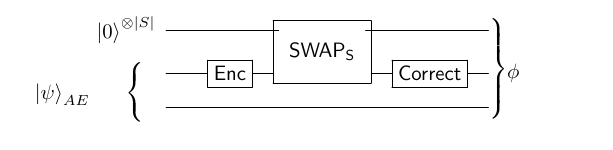}
  \end{center}
   \caption{Circuit that encodes subsystem $A$ of some input state, SWAPs a subregister $S$, $|S| \leq (D-1)/2$, of the encoding, and then applies the error correction procedure on the remaining state} \label{Fig:secret-sharing-QECC}
\end{figure}
Notice that since $|S| \leq (D-1)/2$, for any qubit $\rho$, we have that the error correction procedure maps $\kb{0}^{\otimes |S|}\otimes \Tr_{S}(\Enc(\rho))$ to $\Enc(\rho)$, and therefore
  we have that $\phi = \tau_{S,\ket{\psi}_{AE}} \otimes (\Enc \otimes I)(\kb{\psi}_{AE})$, where $\tau_{S,\ket{\psi}_{AE}} = \Tr_{\overline{S}}((\Enc \otimes I)(\kb{\psi}_{AE})$ and we have a tensor product structure, since we have that $\ket{\psi}$ is a pure state.
  We show now that  $\tau_{S,\ket{\psi}_{AE}}$ has to be independent of $\ket{\psi}_{AE}$, \emph{i.e.}, there exists some $\tau_S$ such that
  $\tau_{S,\ket{\sigma}_{AE}} = \tau_{S}$ for all $\ket{\sigma}_{AE}$.  Notice that this finishes the proof since it implies
  \[
    \Tr_{S}((\Enc \otimes I)(\kb{\psi}_{AE})) = \tau_S \otimes \Tr_{A}(\kb{\psi}_{AE}).
  \]

  Let us then prove the existence of $\tau_S$. Let us assume towards a contradiction that there exist some state $\ket{\rho}$, orthogonal to $\ket{\psi}$ such that $\tau_{S,\ket{\psi}} \ne \tau_{S,\ket{\rho}}$. By linearity, we have that $\tau_{S,\ket{\psi}} \ne \tau_{S,\frac{1}{\sqrt{2}}(\ket{\psi} + \ket{\rho})}$. Notice that we can repeat the above circuit as many times as we want and get
  $\tau_{S,\ket{\psi}}^{\otimes k}$ ($\tau_{S,\frac{1}{\sqrt{2}}(\ket{\psi} + \ket{\rho})}^{\otimes k}$) given a single copy of $\ket{\psi}$ (resp. $\frac{1}{\sqrt{2}}(\ket{\psi} + \ket{\rho})$). In particular if we pick sufficiently large $k$, we have a way to distinguish $\ket{\psi}$ from $\frac{1}{\sqrt{2}}(\ket{\psi} + \ket{\rho})$ with probability strictly larger than $\frac{1}{2}$~\cite{HW12}, but
this is a contradiction since the best success probability is $\left|\bra{\psi}\left(\frac{1}{\sqrt{2}}(\ket{\psi} + \ket{\rho})\right)\right|^2 = \frac{1}{2}$.
\end{proof}

The result in \cite{CGL99} (summarized in \Cref{C:logical-trace-one-qubit}) follows directly from \Cref{L:logical-trace} by considering a trivial system $E$.

\begin{corollary}\label{C:logical-trace-one-qubit}
  Let $\calC$ be an $[[N,1,D]]$-QECC. For any $S$ such that $|S| \leq (D-1)/2$, there exists some density matrix $\tau_S$ such that for all $\psi$
  \[
    \tau_S = \Tr_{\overline{S}}(Enc(\psi)).
  \]

\end{corollary}

\subsection{Proof of  simulation of history states}
\label{sec:proof-gsy}
Here, we present the full proof of \Cref{lem:gsy} (as discussed, this assumes \Cref{lem:simulatable}, which is proved in a self-contained way later in \Cref{sec:simulatable-codes}). We start by defining the verification algorithm~$\VerSimX$ stated in \Cref{lem:gsy} (\Cref{S:new-verification}), and then in \Cref{S:simulatable-history}, we show that the history state of the computation of $\VerSimX$ on some good witness $\WitSim$ is simulatable.

\subsubsection{The circuit $\VerSimX$}
  \label{S:new-verification}

First, we present  the verification algorithm $\VerSimX$ stated in \Cref{lem:gsy}. For that, we start with an arbitrary uniform family of circuits $\{V_x\}$ that we have from the definition of \QMA{}. By definition,~$V_x$ receives as input a $p'(|x|)$-qubit quantum state, for some polynomial $p'$, and auxiliary qubits such that if  $x \in \ayes$ there exists an input state $\ket{\psi}$ such that $V_x$ outputs $1$ with probability $1 - \negl(|x|)$, whereas if $x \in \ano$, for all inputs $\ket{\psi}$,  $V_x$ outputs $1$ with probability $\negl(|x|)$. We depict such a circuit in  \Cref{fig:v1}.

\begin{figure}[h]
  \begin{center}
  \includegraphics{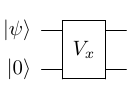}
  \end{center}
  \caption{Verification circuit for some problem $A = (\ayes,\ano)$ in \QMA{}}
  \label{fig:v1}
\end{figure}

 We have that, in order to efficiently compute reduced density matrices of the history state of~$V_x$ on a good witness~$\ket{\psi}$, it is necessary to efficiently compute reduced density matrices of $\ket{\psi}$, and in general, this is not known to be possible. Therefore, our approach here is to modify the
circuit~$V_x$ into~$\VerSimX$ such that completeness and soundness do not change, but such that we are able efficiently compute the reduced density matrices of a good witness $\WitSim$, of the snapshots of the $\VerSimX$ computation on $\WitSim$,  and of the history state of such a computation. We do this modification in two steps. In the first one, we consider~$\VerOtpX$ which is expected to receive a one-time pad of the witness, along with the one-time pad keys. $\VerOtpX$ then uncomputes the one-time pad encryption and performs the original computation. The verification circuit $\VerOtpX$ is depicted in \Cref{fig:v2}.
\begin{figure}[h]
  \begin{center}
  \includegraphics{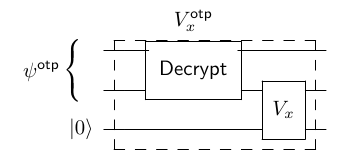}
  \end{center}
  \caption{Verification circuit $\VerOtpX$ for a \QMA{} problem. Here, the honest $\WitOtp = \frac{1}{2^{2p'}} \sum_{a,b} \kb{a,b} \otimes X^aZ^b \kb{\psi} Z^bX^a$, where $\ket{\psi}$ is a $p'$-qubit state. }
  \label{fig:v2}
\end{figure}

 It is easy to see that the completeness and soundness of $\VerOtpX$ are unchanged compared to~$V_x$.
 The reason why we perform this trivial modification is that, in the honest encrypted version,
  when we trace out all qubits of the witness but one, the remaining qubit is in the totally mixed state.
This helps us fix a small bug in \cite{GSY19}, as explained in \Cref{R:bug}. For some witness $\ket{\psi}$ of $V_x$ of size~$p'$, we define the associated witness for $\VerOtpX$ as $\WitOtp = \frac{1}{2^{2p'}}\sum_{a,b}\kb{a,b} \otimes X^aZ^b \kb{\psi} Z^bX^a$.

In the second step, we encode the witness for $\VerSimX$ with an $[[N,1,D]]$ quantum error correcting code  $\mathcal{C}$ that
is $(3s+2)$-simulatable, and
that the encoding, decoding and error detection procedures of $\calC$ have complexity $\poly(N)$.
 In this work, we set $\calC$ to be the
 $\log(3s+5)$-fold concatenated Steane code which yields a $(3s+2)$-simulatable
 code from \Cref{lem:simulatable}.

We define $\VerSimX$ as follows. It is supposed to receive the witness of $\VerOtpX$ encoded under $\calC$ and then $i)$~verifies if each qubit of the witness is correctly encoded under $\calC$, $ii)$~creates encodings of auxiliary $\ket{0}$ and $\ket{\T}$ under $\calC$, $iii)$~performs an encoded version of $\VerOtpX$, either using transversal Clifford gates or performing the $\T$-gadget described on \Cref{fig:new-t-gadget} (\Cref{sec:simulatable-codes}), and $iv)$~decodes the output of the computation. We describe $\VerSimX$ more formally in \Cref{fig:def-vstar} and depict it in \Cref{fig:v3}.

  \begin{figure}
    \hrule
  \begin{enumerate}
    \item Receive as witness the state $\WitSim = \Enc(\WitOtp) = \frac{1}{2^{2p'}} \sum_{a,b} \Enc(\kb{a,b} \otimes X^aZ^b\kb{\psi}Z^bX^a)$
    \item Run \textsf{ChkEnc}:
    \begin{enumerate}
      \item Check if each logical qubit of the witness is encoded under $\calC$, and reject if this is
      not the case \label{step:end-encoding}
    \end{enumerate}
    \item Run \textsf{ResGen}:
    \begin{enumerate}
    \item For each auxiliary qubit, encode $\ket{0}$ under $\calC$
    \item For every $\T$-gate of $V_x$, create $\ket{\T}$ and encode it under
      $\calC$
    \end{enumerate}
    \item Run $\Enc(\VerOtpX)$:
      \begin{enumerate}
    \item Undo the (encoded) one-time pad by transversally applying $\CNOT$ and $\CZ$ gates.
    \item Simulate each gate of $V_x$, either transversally, or
      using unitary $\T$-gadgets.
      \end{enumerate}
    \item Run \textsf{Dec}:
    \begin{enumerate}
      \item Decode the output bit, and accept or reject depending on its value.
      \label{step:decoding} \label{step:begin-decoding}
    \end{enumerate}
  \end{enumerate}
  \hrule
    \caption{Detailed description of $\VerSimX$}
    \label{fig:def-vstar}
  \end{figure}

\begin{figure}
  \begin{center}
  \includegraphics{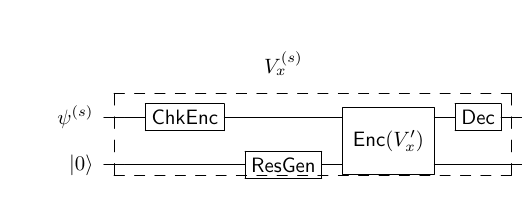}
  \end{center}
  \caption{Verification circuit $\VerSimX$ for the \QMA{} problem, where  $\WitSim = \Enc(\WitOtp) = \frac{1}{2^{2p'}} \sum_{a,b} \Enc(\kb{a,b} \otimes X^aZ^b\kb{\psi}Z^bX^a)$ and $\ket{\psi}$ is  a $p'$-qubit state.}
  \label{fig:v3}
\end{figure}

The completeness and soundness of $\VerSimX$ are straightforward: the acceptance probability of $\VerSimX$ on
\begin{align}\label{eq:structure-witness}
  \Enc(\WitOtp) = \frac{1}{2^{2p'}} \sum_{a,b}\kb{a,b} \otimes \Enc(X^aZ^b\kb{\psi}Z^bX^a)
\end{align}
is exactly the same as $\VerOtpX$ on $\WitOtp$; and  witnesses that are orthogonal to such states are rejected with probability $1$, since $\VerSimX$ first checks if the witness is correctly encoded.

Notice that the size of the circuit $\VerSimX$ is at most $\poly(|x|,N)$-times  bigger than $V_x$. Notice also that if we assume that $V_x$ contains only $\{\Had, \CNOT, \Pg, \T\}$ gates,  then all gates in $\VerSimX$ belong to the gateset $\{\X,\Z, \Pg, \Had, \CNOT, c(\Pg), \T\}$.\footnote{Here, we denote $c(\Pg)$ as the controlled-$\Pg$ gate.} This implies that each gate of $\VerSimX$ acts on at most $2$ qubits.

\begin{remark}\label{R:bug}
We note that in the proof of $\class{MIP}^* = \class{ZK}$-$\class{MIP}^*$ in \cite{GSY19}, when the encoded version of the protocol is defined (similar to \Cref{fig:v3}), the verifier does not check if the provers' answers lie in the codespace, which could potentially hurt soundness. We can easily address this issue by adding the procedure $\mathsf{ChkEnc}$ to check if the witness lies in the codespace. However, such a modification could harm the simulatability of the history state. In order to allow such a simulation (which will be proven in the next subsection), we added the intermediate circuit $\VerOtpX$ where we consider the quantum one-time padded version of the witness (along with the one-time pad keys). Such modifications can be easily incorporated in the context of \cite{GSY19}.
\end{remark}

\subsubsection{Simulation of $\VerSimX$}
\label{S:simulatable-history}

The goal of this section is prove \Cref{lem:gsy}, where the circuit $\VerSimX$ was presented in \Cref{S:new-verification}. Our final goal is to show that the reduced density matrices of the history state
\[\frac{1}{T+1} \sum_{t,t' \in [T+1]} \ketbra{\unary(t)}{\unary(t')} \otimes U_t \cdots U_1(\WitSim \otimes \kb{0}^{\otimes q}) U_1^\dagger \cdots U^\dagger_{t'}\]
of the computation of $\VerSimX$ on a  good witness $\WitSim$ can be simulated. We also show that the simulations have {\em low-energy} according to the local terms of the circuit-to-Hamiltonian construction.

In order to prove it, we first show similar properties (\emph{i.e.}~simulatability and low-energy) for every snapshot
\[U_t \cdots U_1 (\WitSim \otimes \kb{0}^{\otimes q}) U_1^\dagger \cdots U^\dagger_t\]
of the computation of $\VerSimX$ on a  good witness $\WitSim$ (\Cref{L:simulatable-snapshot}), and also for small intervals of the history state (\Cref{L:simulatable-interval}),\emph{i.e.},
\[\frac{1}{|I|} \sum_{t,t' \in I} \ketbra{\unary(t)}{\unary(t')} \otimes U_t \cdots U_1 (\WitSim \otimes \kb{0}^{\otimes q}) U_1^\dagger \cdots U^\dagger_{t'},\]
for $I = \{t_1, t_1+1,\ldots ,t_2\}$, $|I| \leq s+1$.

\begin{lemma}\label{L:simulatable-snapshot}
  Let   $A = (\ayes,\ano)$ be a problem in $\QMA$,  and $\VerSimX = U_T\cdots U_1$ be the verification algorithm described in \Cref{S:new-verification} for some input $x \in A$, where $\VerSimX$ acts on a $p(|x|)$-qubit witness and $q(|x|)$ auxiliary qubits.
  Then there exists a polynomial-time deterministic algorithm $\SimSnap$ that on input $x \in A$, $t \in [T+1]$
  and $Y \subseteq [T+p+q]$ with $|Y| \leq 3s+2$, $\SimSnap(x,t,Y)$ outputs the
  classical description of an $|Y|$-qubit  density matrix $\rho(x,t,Y)$
  with the following properties
  \begin{enumerate}
    \item If $x$ is a yes-instance, for any witness $\WitSim$ in the form of \Cref{eq:structure-witness} that makes $\VerSimX$ accept with probability $1 - \negl(n)$, we have that
      \[\trNorm{\rho(x,t,Y) - \Tr_{\overline{Y}}\left(U_t \cdots U_1(\WitSim \otimes \kb{0}^{\otimes q}) U_1^\dagger \cdots U^\dagger_t\right)} \leq \negl(n).\]
    \item For any auxiliary qubit $j \in Y$, we have that $\Tr_{\overline{\{j\}}}(\rho(x,0,Y)) = \kb{0}$.
    \item Let $t_d$ be the step just before decoding, $t \geq t_d$, and $E \subseteq Y$ be the set of qubits of the encoding of the output qubit in $Y$. We have that
      $\Tr_{\overline{E}}(\rho(x,t,Y)) = \Tr_{\overline{E}}(U_t \cdots U_{t_d+1}\Enc(\kb{1})U_{t_d+1}^\dagger \cdots U_t^\dagger)$.
  \end{enumerate}
\end{lemma}
\begin{proof}
  We prove the result by showing how to compute the reduced state of
  \[\Delta_t = U_t \cdots U_1\left(\WitSim \otimes \kb{0}^{\otimes q}\right)U_1^\dagger \cdots U^\dagger_t,\]  where $\WitSim$ is a hypothetical good witness for $\VerSimX$.

\begin{figure}
  \begin{center}
  \includegraphics{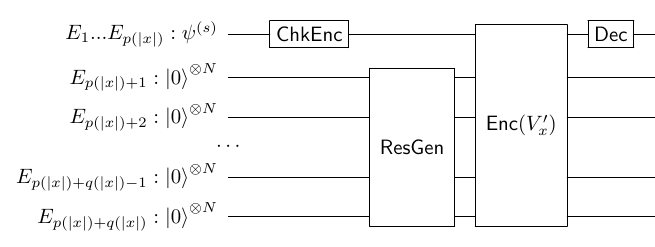}
  \end{center}
  \caption{Separation of registers in $\VerSimX$.}
  \label{fig:registers}
\end{figure}

  Let us denote by $E_i$ the set of qubits of the encoding of the $i$-th qubit (even if at step $t$, $i$-th logical qubit is  unencoded but later on the computation it will be so). We depict such regiset of qubits in \Cref{fig:registers}.
  We will split our analysis in three phases: before all the qubits are encoded,
  the logical computation and the decoding. For each of the phases, we have the
  following definitions:
  \begin{itemize}
    \item $t_0$ and $t_1$: For a fixed $t$, we define $t_0 \leq t \leq t_1$ such that at $t_0$
      and $t_1$, all qubits are all fully encoded or fully unencoded (i.e., at
      step $t_0$ and $t_1$ there are no operations such as performing a logical
      gate, encoding, decoding, etc.);
    \item $Q$:  for fixed $t_0,t_1$, we let $Q$ be the set of qubits on which
      operations $U_{t_0+1},\ldots,U_{t_1}$ act; and
    \item $\mathcal{U}$: we define $\mathcal{U}$ as the set logical qubits that
      are still unencoded by step $t_1$.
  \end{itemize}

  Let us consider the case before all qubits are encoded (i.e. until the last
  step of \textsf{ResGen}).  Let $t_0 \leq t \leq t_1$, where $t_0$ is the timestamp just before starting the operation of timestamp $t$ (which consists on either checking the encoding of a qubit of the witness or creating a resource state) and  $t_1$ be the timestamp after this operation is peformed.
     Let us assume, for simplicity of notation, that the qubits in $Q$ are the
     first $|Q|$ qubits of the state (and for the other cases follow
     analogously, but the states are permuted). Using \Cref{L:logical-trace}, we
     have that
  \begin{align} \label{eq:snapshot-first-phase}
    \xi_{Y,t_0} := \Tr_{\overline{Y \cup Q}}(\Delta_{t_0}) =
    \Tr_{\overline{Q \setminus Y}}(\Delta_{t_0}) \otimes \left(\bigotimes_{i \not\in
    \mathcal{U}} \tau_{(Y \setminus Q) \cap E_i}\right) \otimes \left(\bigotimes_{i \in
    \mathcal{U}} \kb{0}^{\otimes |Y \cap E_i|}\right),
\end{align}
  where we use the facts that
  every qubit in $Q$ is encoded at step $t_0$ (so by \Cref{L:logical-trace} we
  have the tensor product structure)
  and that for every $i \in \mathcal{U}$ we have that $E_i
  \cap Q = \emptyset$, since we know that the $i$-th qubit is still unencoded by
  step $t_1$.

  By \Cref{lem:simulatable}, $\SimC$ can compute $\tau_{Y \cap E_i}$ in time
  $\poly(2^N)$ without knowing $\Delta_{t_0}$.\footnote{In our self-contained
  proof, this appears in \Cref{L:part1}.} Therefore, if we can compute the
  $N$-qubit state $\Tr_{\overline{Q}}(\Delta_{t_0})$  in time $\poly(2^N)$,
  $\SimSnap$ can compute $\rho(x,t,Y)$ by computing
  $\Tr_{\overline{Q}}(\Delta_{t_0})$, classically simulating the unitaries
  $U_{t_0+1},\ldots,U_t$, tracing out the qubits in $Q \setminus Y$ and finally
  appending $\tau_{(Y \setminus Q) \cap E_i}$ and the auxiliary qubits. The runtime is $\poly(2^N,|x|)$, and
  \Cref{eq:snapshot-first-phase} implies that the first property of $\SimSnap$ holds for $t$ on the first phase.  Notice also that since the state~$\xi_{Y,0}$ has $\kb{0}$ on all auxiliary qubits, the second property of $\SimSnap$ holds.

  We describe now how to compute $\Tr_{\overline{Q}}(\Delta_{t_0})$ in time $\poly(2^N)$.
When $t$ lies in some intermediate step of \textsf{ChkEnc}, in other words, when $\VerSimX$ checks the encoding of qubit $q$,
  $\Tr_{\overline{Q}}(\Delta_{t_0}) = \Enc(\Tr_{\overline{q}}(\WitOtp)) = \Enc(\Tr_{\overline{q}}(I/2))$, since
  we defined the honest witness for $\WitOtp$, the $q$-th qubit is either a one-time padded state (without its one-time pad key) or it is the key for the one-time pad (and therefore it is also totally mixed). The result follows for the encode checking phase since the encoding of the totally mixed state can be trivially computed in time $\poly(2^N)$. We remark that this is the exact (and only) part of the proof where we need the properties of the intermediate verifier~$\VerOtpX$.
  If $t$ lies in some intermediate step of \textsf{ResGen},\emph{i.e.}, when $\VerSimX$ creates the encoding of the auxiliary $\ket{0}$ or $\ket{\T}$ qubits, we have that $\Tr_{\overline{Q}}(\Delta_{t_0}) = \kb{0}^{\otimes |Q|}$, whose description can be trivially computed.

\medskip

  We now consider the second phase, where the logical computation is  performed. Here, let $t_0$ be the timestamp at the beginning of the logical computation  performed at time $t$.
  In this case, we have that
\[
  \Delta_t = U_{t} \cdots U_{t_0+1} \Delta_{t_0}U_{t_0+1}^\dagger \cdots U_{t}^\dagger,  \]
and all qubits of $\Delta_{t_0}$ are  fully encoded. This
corresponds to the simulation in the middle of the application of a logical
gate, and by  \Cref{lem:simulatable}, $\Tr_{\overline{Y}}(\Delta_{t})$ can
  be efficiently computed.\footnote{In our self-contained proof, this appears in \Cref{L:cliff} for Clifford computations and in \Cref{L:T-gadget} for $\T$ computation with magic states.
  }

Finally, we reach the third phase and let $t_d$ be the timestep where the decoding
starts. Since every qubits at timestep $t_d$ is fully encoded, we can define a quantum state $\gamma$ such that  $\Delta_{t_d} = \Enc(\gamma)$ and then we define $\tilde{\Delta}_{t_d} = \frac{\Enc\left((\kb{1}\otimes I) \gamma \right)}{\Tr\left((\kb{1}\otimes I)\gamma\right)}$.
  As in \Cref{eq:snapshot-first-phase}, $\SimSnap$ can compute the reduced state on the encoding of qubits of $\tilde{\Delta_{t_d}}$ that are not the output by \Cref{lem:simulatable}, and for the output qubit, $\SimSnap$ can classically simulate $U_{t_d+1},\ldots,U_t$ on $\Enc(\kb{1})$ and then trace out the qubits not in $Y$. These operations can be performed in time $\poly(2^N,|x|)$. By construction we have the third property of $\SimSnap$.

Notice that if $\Delta_{t_d}$ is indeed the $t_d$-th step of the computation $\VerSimX$ on a good witness $\WitSim $, we have that $\Tr(\Delta_{t_0} - \tilde{\Delta}_{t_0}) \leq \negl(|x|)$, and the first property of $\SimSnap$ follows for $t \geq t_d$.
\end{proof}

We now show simulatability of intervals of the history state.

\begin{lemma}\label{L:simulatable-interval}
  Let   $A = (\ayes,\ano)$ be a problem in $\QMA$,  and $\VerSimX = U_T\cdots U_1$ be the verification algorithm described in \Cref{S:new-verification} for some input $x \in A$ and $s \geq 5$, where $\VerSimX$ acts on a $p(|x|)$ qubit witness and $q(|x|)$ auxiliary qubits.
  Then there exists a polynomial-time deterministic algorithm $\SimInt$ that on input $x \in A$,
  $I = \{t_1,t_1+1,\ldots,t_2\} \subseteq [T+1]$, $t_2 - t_1 \leq s+1$ and $S \subseteq [T + p + q]$, $|S| \leq s$, and  $\SimInt(x,I,S)$ runs in time $\poly(|x|,N)$ and outputs the
  classical description of an $|S|$-qubit  density matrix $\rho(x,I,S)$
  with the following properties
  \begin{enumerate}
    \item If $x$ is a yes-instance, then there exists a witness $\WitSim$ that makes $\VerSimX$ accept with probability at least $1 - \negl(n)$ such that
      $\trNorm{\rho(x,I,S) -
      \Tr_{\overline{S}}(\hist_I)} \leq \negl(n)$, where
      \[
        \hist_I = \frac{1}{|I|}  \sum_{t,t' \in [I]} \ketbra{\unary(t)}{\unary(t')} \otimes U_t\cdots U_1 \left(\WitSim \otimes \kb{0}^{\otimes q}\right)U_1^{\dagger}\ldots U_{t'}^\dagger,
      \] is an {\em interval} of the history state of $\VerSimX$ on the witness $\WitSim$.
      \item Let $H_i$ be a clock term ($H^{clock}_t$), initialization term ($H^{init}_j$) or the output term ($H^{out}$) from the circuit-to-Hamiltonian construction of $\VerSimX$\footnote{
        See \Cref{sec:circuit-to-ham}.}
        and $S_i$ be the set of qubits on which $H_i$ acts non-trivially. Then for every $x \in A$, we have $\Tr(H_i\rho(x,I,S_i)) = 0$.
    \item For any propagation term $H^{prop}_t$ from the circuit-to-Hamiltonian construction of $\VerSimX$ and the corresponding set of qubits $S_t$, and for every $x \in A$ and $I$ such that $t,t+1 \in I$ or $t,t+1 \not\in I$, we have $\Tr(H^{prop}_t\rho(x,I,S_t)) = 0$.
  \end{enumerate}
\end{lemma}
\begin{proof}
  For simplicity, let $\Delta_{t,t'} = U_t\ldots U_1(\WitSim \otimes \kb{0}^{\otimes q})U_1^\dagger \ldots U_{t'}^\dagger$.

  Let $C$ be the set of clock qubits and~$W$ be the set of working qubits.  We have that
  \begin{align}
    &\Tr_{\overline{S}}(\Phi_I) \nonumber \\
    &= \frac{1}{|I|} \sum_{t,t' \in I}
  \Tr_{\overline{S \cap C}}(\ketbra{\unary(t)}{\unary(t')})
  \otimes \Tr_{\comp{S \cap W}} \Paren{ \Delta_{t,t'}} \nonumber \\
    &=
 \frac{1}{|I|} \sum_{t,t' \in I}
  \Tr_{\overline{S \cap C}}(\ketbra{\unary(t)}{\unary(t')})
  \otimes \Tr_{\comp{S \cap W}} \Paren{ U_{t} \ldots U_{t_1+1}\Delta_{t_1,t_1}U_{t_1+1}^\dagger \ldots U^\dagger_{t'}} \nonumber \\
    &=
 \frac{1}{|I|} \sum_{t,t' \in I}
  \Tr_{\overline{S \cap C}}(\ketbra{\unary(t)}{\unary(t')})
    \otimes \Tr_{G \setminus S} \Paren{ U_{t} \ldots U_{t_1+1}\Tr_{\comp{(S \cap W) \cup G}}\left(\Delta_{t_1,t_1}\right)U_{t_1+1}^\dagger \ldots U^\dagger_{t'}} \nonumber
  ,
  \end{align}
  where $G$ is the set of qubits on which $U_{t_1}, \ldots ,U_{t_2}$ act.
Notice that since $t_2 - t_1 \leq |S|+1$, and each~$U_j$ acts on at most $2$ qubits, so we have $|G| \leq 2(|S|+1)$ and  $|S \cup G| \leq 3s+2$.

  Let $Y = (S \cap W) \cup G$. $\SimInt(x,I,S)$ starts by running $\SimSnap(x,t_1,Y)$ from \Cref{L:simulatable-snapshot} to compute the state $\tilde{\rho}(x,t_1,Y)$ such that $\norm{\tilde{\rho}(x,t_1,Y) - \Tr_{\comp{Y}}(\Delta_{t_1,t_1})} \leq \negl(|x|)$ for some yes-instance~$x$. Then $\SimInt(x,t_1,Y)$  computes
\begin{align}\label{eq:state-interval}
  \rho(x,I,S) = \Tr_{\overline{S \cap C}}(\ketbra{\unary(t)}{\unary(t')}) \otimes \Tr_{G \setminus S} \Paren{ U_{t} \ldots U_{t_1+1}\Tr_{\comp{(S \cap W) \cup G}}\tilde{\rho}(x,t_1,Y)U_{t_1+1}^\dagger \ldots U^\dagger_{t'}}
\end{align}
  from  $\tilde{\rho}(x,t_1,Y)$  in time $\poly(2^N,|x|)$ and it follows that for a yes-instance $x$ and a good witness $\WitSim$,
\[\trNorm{\Tr_{\overline{S}}(\Phi_I) -  \rho(x,I,S)} \leq \negl(n),\]
which proves the runtime and first property of $\SimInt$.

We show now that the output of $\SimInt$ has energy $0$ with respect to the clock, initalization and output terms of $H_{\VerSimX}$, proving the second property of $\SimInt$. For that, we consider a $5$-local term~$H_i$ and the set of qubits $S_i$ on which $H_i$ acts non-trivially. We prove the property for each type of local terms.

\begin{itemize}
  \item If $H_i$ is a clock constraint, since $\rho(x,I,S_i)$ always output reduced density matrices on clock registers that are consistent with valid unary encodings, $\Tr(H_i\rho(x,I,S_i)) = 0$ for every $I$.
  \item If $H_i$ is a initialization constraint $H^{init}_j = \kb{01}$ and $S_i$ consists of the first clock qubit and the $j$-th auxiliary qubit. We consider two subcases: if $0 \not\in I$, $\rho(x,I,S_i)$ has energy $0$ because the content in the clock qubit is $\kb{1}$ and $\Tr(\kb{01} (\kb{1} \otimes \gamma)) =0$; if $0 \in I$, then $t_1 = 0$ and by the second property of $\SimSnap$, it follows that $\tilde{\rho}(x,0,\{j\}) = \kb{0}$, and therefore $\Tr(\kb{01}\rho(x,I,S_i)) = \Tr(\kb{1}\tilde{\rho}(x,0,\{j\})) = 0 $.
  \item If $H_i$ is the constraint $H^{out} = \kb{10}$ and $S_i$ consists of the last clock qubit and the output qubit. We again have two cases: if $T \not\in I$, $\rho(x,I,S_i)$  has energy~$0$ because of the clock qubit (as in the previous cases); if $T \in I$, then since $|S_i| \leq s$, by the third property of $\SimSnap$, it follows that $\rho(x,t_0,S_i) =  \Tr_{\overline{E}}(U_t \cdots U_{t_d+1}\Enc(\kb{1})U_{t_d+1}^\dagger \dots U_t^\dagger)$, where $t_d$ is the step just before the decoding. This implies that
    \[\Tr(\kb{10} \rho(x,I,S_i)) = \Tr(\kb{10} (\kb{1} \otimes U_T\tilde{\rho}(x,T-1,S_i)U_T^\dagger)) =  \Tr(\kb{10} \kb{11}) = 0. \]
\end{itemize}

For the propagation term $H^{prop}_t$, we have again two subcases: if $t,t+1 \not\in I$, $\rho(x,I,S_i)$  has energy~$0$ because of the clock qubits (as in the previous cases); if $t,t+1 \in I$, notice that any state
    \[\sum_{t_1+1 \leq t,t' \leq t_2} \ketbra{\unary(t-t_1)}{\unary(t'-t_1)} \otimes U_{t_2} \cdots U_{t_1+1}\sigma U_{t_1+1}^\dagger \cdots  U_{t_2}^\dagger\]
    has energy~$0$,\footnote{Here we use again the notation $t_1 = \min(I)$ and $t_2 = \max(I)$.} since it has the correct propagation of the unitary $U_t$ at step $t$. This proves the third property of $\SimInt$.
\end{proof}

We finally prove \Cref{lem:gsy}.

\newtheorem*{thm:repeat-gsy}{\Cref{lem:gsy} (restated)}
\begin{thm:repeat-gsy}
  \bodygsy
\end{thm:repeat-gsy}
\begin{proof}
  We let $\VerSimX$ be the circuit for $A$ as defined in \Cref{S:new-verification} and
  we follow the notation of \Cref{L:simulatable-interval} with $\Delta_{t,t'} = U_t\ldots U_1(\WitSim \otimes \kb{0}^{\otimes q})U_1^\dagger \cdots U_{t'}^\dagger$.
\begin{figure}
  \begin{center}
  \begin{tikzpicture}[box/.style={rectangle,draw=black,thick, minimum size=1cm}]
\node[rectangle,draw=black,dashed,minimum size=1cm] at (0,0){$0$};
\node[box,fill=blue!40] at (1,0){$1$};
\node[box,fill=blue!40] at (2,0){$2$};
\node[box,fill=red!40] at (3,0){$3$};
\node[box,fill=red!40] at (4,0){$4$};
\node[box,fill=red!40] at (5,0){$5$};
\node[box,fill=blue!40] at (6,0){$6$};
\node[box,fill=blue!40] at (7,0){$7$};
\node[box,fill=red!40] at (8,0){$8$};
\node[box,fill=blue!40] at (9,0){$9$};
\node[box,fill=blue!40] at (10,0){$10$};
\draw [decorate,decoration={brace,amplitude=4pt},xshift=-3pt,yshift=0pt] (0.5,-0.8) -- (-0.30,-0.8) node [black,midway,yshift=-1em] {\footnotesize $I_1$} ;
\draw [decorate,decoration={brace,amplitude=4pt},xshift=-3pt,yshift=0pt] (1.5,-0.8) -- (0.70,-0.8) node [black,midway,yshift=-1em] {\footnotesize $I_2$};
\draw [decorate,decoration={brace,amplitude=7pt},xshift=-3pt,yshift=0pt] (5.5,-0.8) -- (1.70,-0.8) node [black,midway,yshift=-1em] {\footnotesize $I_3$};
\draw [decorate,decoration={brace,amplitude=4pt},xshift=-3pt,yshift=0pt] (6.5,-0.8) -- (5.70,-0.8) node [black,midway,yshift=-1em] {\footnotesize $I_4$};
\draw [decorate,decoration={brace,amplitude=4pt},xshift=-3pt,yshift=0pt] (8.5,-0.8) -- (6.70,-0.8) node [black,midway,yshift=-1em] {\footnotesize $I_5$};
\draw [decorate,decoration={brace,amplitude=4pt},xshift=-3pt,yshift=0pt] (9.5,-0.8) -- (8.70,-0.8) node [black,midway,yshift=-1em] {\footnotesize $I_6$};
\draw [decorate,decoration={brace,amplitude=4pt},xshift=-3pt,yshift=0pt] (10.5,-0.8) -- (9.70,-0.8) node [black,midway,yshift=-1em] {\footnotesize $I_7$};
\end{tikzpicture}
  \end{center}
  \caption{An example of how to create the intervals based on the traced out clock qubits. The filled boxes represent both the clock qubits and the corresponding timestamp in the interval. The empty box represent the timestamp $0$ (which does not have an associated clock qubit). The filled boxes in blue represent the traced out qubits and the ones in red represent the qubits that are not traced out. Finally, the interval $I_j$ contains the timestamps which are strictly smaller than the $i$-th traced out qubit and at least the $(i-1)$-st traced out qubit (or $0$ when $i = 1$).}
  \label{fig:intervals}
\end{figure}

  Let
  $C_{tr} = \{i_1,\ldots ,i_{|C_{tr}|+1}\} \subseteq \{1,\ldots ,T\}$ denote the set of clock qubits that are not in $S$, and let $i_1 < i_2 < \ldots  < i_{|C_{tr}|}$. Let us partition $[T+1]$ into the
  intervals $I_1, \ldots ,I_{|C_{tr}|+1}$ such that $I_j$ contains the timestamps which are strictly smaller than the $i_j$ and at least the $i_{j-1}$ (or $0$ when $i = 1$). More formally, we have
  \begin{itemize}
    \item $I_1 = \{ t : t \in [T+1] \text{ and } t < i_1\}$;
    \item For $2 \leq j < |C_{tr}|$, $I_j = \{ t : t \in [T+1] \text{ and } i_{j-1} \leq t < i_j\}$;
    \item $I_{|C_{tr}|} =  \{ t : t \in [T+1] \text{ and } t \geq i_{|C_{tr}|}\}$.
  \end{itemize}
  We depict such intervals in \Cref{fig:intervals}.
  Notice that since $|C_{tr}| \geq |T| - |S|$, it follows that $|I_j| \leq |S|+1$.

  We have that
  \begin{align*}
    \Tr_{i}(\ketbra{\unary(t)}{\unary(t')}) &=
    \Tr_{i}\left(
     \kb{1}^{\otimes t} \otimes \ketbra{1}{0}^{\otimes (t'-t)} \otimes \kb{0}^{T-t'}
    \right) \\
    &= \Tr(\ketbra{\unary(t)_i}{\unary(t')_i}) \bigotimes_{\substack{j \in \{1,\ldots ,T\} \\ j \ne i}} \ketbra{t_j}{t'_j}
  \end{align*}
  vanishes if and only if $\Tr(\ketbra{\unary(t)_i}{\unary(t')_i}) = 0$, \emph{i.e.}, when $t+1 \leq i \leq t'$.
  Using this argument for each $i \in C_{tr}$, it follows that  $t,t' \in I_j$ if and only if
  $\Tr_{C_{tr}}(\ketbra{\unary(t)}{\unary(t')} \otimes \Delta_{t,t'}) \ne 0$. We have then that
  \begin{align}
  \Tr_{C_{tr}}(\Phi) &=
    \frac{1}{T+1}\sum_{t,t' \in [T+1]} \Tr_{C_{tr}}(\ketbra{\unary(t)}{\unary(t')} \otimes \Delta_{t,t'}) \nonumber \\
  &=
    \frac{1}{T+1}\sum_{t,t' \in I_j} \Tr_{C_{tr}}(\ketbra{\unary(t)}{\unary(t')} \otimes \Delta_{t,t'})  \nonumber\\
  &=
\sum_{j=1}^\ell
    \frac{|I_j|}{T+1} \Tr_{C_{tr}} (\Phi_{I_j}), \label{eq:intervals-history}
  \end{align}
where in the second equality we remove the crossterms that vanish and in the third equality we regroup the terms in states with the form
        $\hist_I = \frac{1}{|I|}  \sum_{t,t' \in [I]} \ketbra{\unary(t)}{\unary(t')} \otimes \Delta_{t,t'}$.

 We can then  define $\Sim(x,S)$ as follows:
  \begin{enumerate}
    \item Compute the set $C_{tr}$;
    \item Compute the intervals $I_1,\ldots,I_{|C_{tr}|+1}$;
    \item Run $\SimInt$ from \Cref{L:simulatable-interval} to compute $\rho(x,I_j,S) = \SimInt(x,I_j,S)$
    \item Output $ \rho(x,S) = \sum_{j=1}^{|C_{tr}|}
      \frac{|I_j|}{T+1} \rho(x,I_j,S)$.
  \end{enumerate}

  Since $\SimInt$ runs in time $\poly(2^N,|x|)$, so does $\Sim$.
  From \Cref{eq:intervals-history} and the fact that
  $\trNorm{\rho(x,I_j,S) - \Tr_{\overline{S}}(\hist_{I_j})} \leq \negl(n)$, by \Cref{L:simulatable-interval},  it follows that
  $\trNorm{\rho(x,S) - \Tr_{\overline{S}}(\hist)} \leq \negl(n)$.

 Notice that a propagation term $H^{prop}_t$ acts on the clock qubits $t$ and $t+1$ and therefore, when we simulate its corresponding set of qubits, $t,t+1 \not\in C_{tr}$. It follows  by the definition of the intervals $I_j$ that in this case either $t,t+1 \in I_j$ or $t,t+1 \not\in I_j$. Thus, we have by the second and third properties of $\SimInt$ of \Cref{L:simulatable-interval} that for all intervals $I_j$ and local terms $H_i$ of the circuit-to-Hamiltonian construction of $\VerSimX$ with corresponding set of qubits $S_i$, $\Tr(H_i\rho(x,I_j,S_i)) = 0$ , and the same holds for $\rho(x,S_i)$ by convexity.
\end{proof}

\subsection{New proof for locally simulatable codes}
\label{sec:simulatable-codes}

Finally, we now provide our new (simpler) proof for \Cref{lem:simulatable}. We split our proof in three parts: in \Cref{L:part1}, we show that codewords are simulatable; then in \Cref{L:cliff} we show that intermediate steps of transversal Clifford gates are simulatable; finally in \Cref{L:T-gadget} we prove that intermediate steps of $\T$-gadgets are simulatable. These three parts together  prove \Cref{lem:simulatable} in a straightforward way. Note that our proofs are applicable for any  QECC that admits transversal Clifford gates and Clifford gadgets for non-Clifford gates (such as $\T$) with magic states.

Let us start by showing how to compute  reduced density matrices on a small set of qubits of a codeword  of a QECC.
\begin{lemma}
\label{L:part1}
  Let $\calC$ be an $[[N,1,D]]$-QECC whose encoding procedure $\Enc$ has complexity $\poly(N)$ and let $\rho$ be a qubit.
  Then there exists a classical algorithm  $\SimC^{\mathsf{CW}}$ that on input $S \subseteq [N]$,  $|S|\leq(D-1)/2$,  runs in time $poly(2^N)$ and outputs the classical description of $\rho(S)$ such that for every qubit $\sigma$, we have
    \begin{align}\label{eq:trace-out-encoding}
     \rho(S) =  \Tr_{\overline{S}}\left(Enc(\sigma)\right).
    \end{align}
\end{lemma}
\begin{proof}
  From \Cref{C:logical-trace-one-qubit}, we have that $|S| \leq (D-1)/2$ implies that there exists some state $\tau_S$ such that for all $\psi$
  \[
    \tau_S = \Tr_{\overline{S}}(\Enc(\psi)).
  \]

   In this case, the algorithm $\SimC^{\mathsf{CW}}$ can then output \Cref{eq:trace-out-encoding} by classically  computing the classical description of  $\Enc(\kb{0})$ and then computing
    $\rho(S) = \Tr_{\overline{S}}(\Enc(\kb{0})) = \Tr_{\overline{S}}(\Enc(\sigma))$ in time $\poly(2^N)$.
\end{proof}

We now apply the \Cref{L:part1} to show how to compute the reduced density matrix on a state in the intermediate steps of transversal Clifford computation on encoded data.
\begin{lemma}\label{L:cliff}
  Let $\calC$ be  the  $[[N,1,D]]$-QECC obtained by the  $k$-fold concatenation of the Steane code.   Let $G \in \{\Had,\CNOT,\Pg\}$, $m_G \in \{1,2\}$ be the number of qubits on which $G$ act and $U_1,\ldots, U_{N}$  be the physical gates that implement $G$ transversally on the encoding  of $m_G$ qubits under $\calC$, {\em i.e.}, $U_i$ applies~$G$ on the $i$-th qubit of the encoding of each qubit of an $m_G$-qubit system.
  Then there exists a classical algorithm  $\SimC^{\mathsf{Cliff}}$ that on  input $G$, $0 \leq t \leq N$ and subset $S$, $|S|\leq(D-1)/4$, we have that $\SimC^{\mathsf{Cliff}}(G,t,S)$ runs in time $poly(2^N)$ and outputs the classical description of a  state $\rho(G,t,S)$ such that for every $m_G$-qubit state $\sigma$
    \begin{align}%
    \rho(G,t,S) = \Tr_{\overline{S}}\left((U_t \cdots U_1) \Enc(\sigma) (U_{t} \cdots U_1)^\dagger\right).
    \end{align}
       Moreover, $\SimC^{\mathsf{Cliff}}$ can be modified to receive a classical bit $b$ and $\SimC^{\mathsf{Cliff}}(G,t,S,b)$ outputs the classical description of a state $\rho(G,t,S,b)$ such that
     \[\rho(G,t,S,b) = \Tr_{\overline{S}}\left((U_t^b \cdots U_1^b) ) \Enc(\sigma)(U_{t}^b \cdots U_1^b)^\dagger\right).\]
\end{lemma}
\begin{proof}
  Since tensor products of one-qubit Pauli matrices form a basis for the space of all matrices,\footnote{In other words, for every $2^t \times 2^t$ complex matrix $M$, we have
    $M = \frac{1}{2^t} \sum_{P \in \mathcal{P}_{t}} \Tr(MP) \cdot P$.} we have that
    \begin{equation*}
     \Tr_{\overline{S}}\left((U_t \cdots U_1) \Enc(\sigma) (U_{t}
      \cdots U_1)^\dagger\right) =
      \frac{1}{2^{|S|}} \sum_{P \in \mathcal{P}_{|S|}} \Tr\left((U_t \cdots U_1) \Enc(\sigma) (U_{t}
      \cdots U_1)^\dagger P\right) \cdot P,
    \end{equation*}
    where we abuse  notation  and inside the trace, we extend $P$ to act on $mN$ qubits, acting non-trivially on the qubits in $S$ and trivially ({\em i.e.}, identity) on $\overline{S}$.  Considering a single term in the above sum, we have that
    \begin{align*}
      &\Tr\left((U_t \cdots U_1) \Enc(\sigma)  (U_{t}
      \cdots U_1)^\dagger P\right)
      \\
      &= \Tr\left((U_t \cdots U_1) \Enc(\sigma)  P' (U_{t}
      \cdots U_1)^\dagger \right) \\
      &= \Tr\left(\Enc(\sigma)  P'\right) \\
      &=
      \Tr_{R}\left(\Tr_{\overline{R}}\left(\Enc(\sigma)  \right)P'\right),
    \end{align*}
    where $P'$ is the unique tensor product of Pauli matrices defined by $PU_t \cdots U_1 = U_t \cdots U_1 P'$ (where we consider again $P$ acting on $mN$ qubits as discussed above),  the second equality holds from the cyclic property of trace, and  $R$ is the set of qubits on which $P'$ acts non-trivially ({\em i.e.}, not identity). Note that since each $U_i$ acts on a distinct set of $2$ physical qubits and~$P$ acts non-trivially on at most~$|S|$ qubits, it follows that $P'$ acts non-trivially on at most $2|S|$ qubits and therefore $|R| \leq 2|S|$.
    In this case, $\SimC^{\mathsf{Cliff}}$ can compute the classical description of $\rho(R) = \Tr_{\overline{R}}\left(Enc(\sigma)  \right)$ using $\SimC^{\mathsf{CW}}$, since $|R| \leq 2|S| \leq (D-1)/2$.  With the classical description of $\rho(R)$ computed by $\Sim^{\mathsf{CW}}$ from \Cref{L:part1}, $\SimC^{\mathsf{Cliff}}$  can easily compute
    \[
      \rho(G,t,S) =
  \frac{1}{2^{|S|}}     \sum_{P \in \mathcal{P}_{|S|}} \Tr(\rho(R)P') \cdot P
    =
     \Tr_{\overline{S}}\left((U_t \cdots U_1) \Enc(\sigma) (U_{t}
      \cdots U_1)^\dagger\right)
      ,
      \]
    by iterating over all $P \in \mathcal{P}_{|S|}$, computing $P'$ and $\Tr(\sigma P')$ and then summing up $\Tr(\sigma P') \cdot P$.

      Moreover, we can easily adapt the above method to the situation where $G$ is controlled by a bit~$b$: if
 $b = 0$, run $\SimC^{\mathsf{CW}}(S)$ on the qubits, since no computation is performed on the qubits; whereas if
 $b = 1$, run $\SimC^{\mathsf{Cliff}}(G,t,S)$, since $G$ is being performed on the qubits.
\end{proof}

\begin{figure}[h]
  \begin{subfigure}[b]{0.4\linewidth}
\begin{tikzpicture}
\node at (-.3,1) {$\ket{\psi}$};
\node at (-.5,0) {$\mathsf{T}\ket{+}$};

\draw(0,1)--(1,1); 	\draw (1.5,1.02)--(3.5,1.02);
				\draw (1.5,.98)--(3.5,.98);
\draw(0,0)--(3.5,0);

\draw (.5,0)--(.5,1.1);
\draw (.5,1) circle (.1);
\filldraw (.5,0) circle (.05);

\filldraw[fill=white] (1,.75) rectangle (1.5,1.25);
\draw (1.05,1) arc (120:60:.4);
\draw (1.25,.9) -- (1.4,1.1);

\filldraw (2,1) circle (.05);
\draw (1.98,.25)--(1.98,1);
\draw (2.02,.25)--(2.02,1);
\filldraw[fill=white] (1.75,-.25) rectangle (2.25,.25);
\node at (2,0) {$\mathsf{X}^c$};

\filldraw (2.75,1) circle (.05);
\draw (2.73,.25)--(2.73,1);
\draw (2.77,.25)--(2.77,1);
\filldraw[fill=white] (2.5,-.25) rectangle (3,.25);
\node at (2.75,0) {$\mathsf{P}^c$};

\node at (3.75,1) {$c$};
\node at (4,0) {$\mathsf{T}\ket{\psi}$};
\end{tikzpicture}
  \caption{}
  \label{fig:t-gadget}
  \end{subfigure}~
  \begin{subfigure}[b]{0.4\linewidth}
\begin{tikzpicture}
  \node at (-.8,1) {$Enc(\ket{\psi})$};
  \node at (-1,0) {$Enc(\mathsf{T}\ket{+})$};

\draw(0,1)--(1,1); 	\draw (1.5,1.02)--(2,1.02);
				\draw (1.5,.98)--(2,.98);
\draw(0,0)--(4.5,0);

  \draw[red] (2.5,0.75) -- (4.5,0.75);
  \draw[red] (2.5,0.85) -- (4.5,0.85);
  \node at (3.4,1) {$\cdots$};
  \draw[red] (2.5,1.15) -- (4.5,1.15);
  \draw[red] (2.5,1.25) -- (4.5,1.25);

\filldraw[fill=white] (1.75,0.75) rectangle (2.65,1.25);
\node at (2.2,1) {$Dec$};

\draw (.5,0)--(.5,1.1);
\draw (.5,1) circle (.1);
\filldraw (.5,0) circle (.05);

\filldraw[fill=white] (1,.75) rectangle (1.5,1.25);
\draw (1.05,1) arc (120:60:.4);
\draw (1.25,.9) -- (1.4,1.1);

  \filldraw[red] (3,1.25) circle (.05);
  \draw[red] (2.98,.25)--(2.98,1.25);
  \draw[red] (3.02,.25)--(3.02,1.25);
\filldraw[fill=white] (2.75,-.25) rectangle (3.25,.25);
\node at (3,0) {$\mathsf{X}^c$};

  \filldraw[red] (3.75,1.25) circle (.05);
  \draw[red] (3.73,.25)--(3.73,1.25);
  \draw[red] (3.77,.25)--(3.77,1.25);
\filldraw[fill=white] (3.5,-.25) rectangle (4,.25);
\node at (3.75,0) {$\mathsf{P}^c$};

\node at (4.75,1.3) {$c$};
\node at (5,0) {$\mathsf{T}\ket{\psi}$};
\end{tikzpicture}
    \caption{}
\label{fig:encoded-t-gadget}
  \end{subfigure}
\caption{Gadget for performing $\T$-gate. In \Cref{fig:t-gadget}, we consider the gadget on unencoded qubits, whereas in \Cref{fig:encoded-t-gadget}, it is performed on encoded data. Note that in \Cref{fig:encoded-t-gadget}, the control qubit needs to be decoded in order to perform transversal $X^c$ and $P^c$ operations and we highlight the operations done at the decoded level in red.
}
\end{figure}

\begin{figure}[h]
\centering
\begin{tikzpicture}
  \node at (-1,2) {$\Enc(\ket{0})$};
  \node at (-1,1) {$\Enc(\ket{\psi})$};
  \node at (-1,0) {$\Enc(\mathsf{T}\ket{+})$};

  \draw[dashed] (1.25, 2.5) -- (1.25,-0.5);
  \draw(0,2) -- (6,2);
\draw(0,1)--(2.5,1);
\draw(4.5,1)--(6,1);
\draw(0,0)--(6,0);

\draw (.5,0)--(.5,1.1);
\draw (.5,1) circle (.1);
\filldraw (.5,0) circle (.05);

\draw (1,1)--(1,2.1);
\draw (1,2) circle (.1);
\filldraw (1,1) circle (.05);

\filldraw[fill=white] (1.5,0.75) rectangle (2.5,1.25);
\node at (2,1) {$\Dec$};

  \draw[red] (2.5,0.75) -- (4.5,0.75);
  \draw[red] (2.5,0.85) -- (4.5,0.85);
  \node at (3.4,1) {$\cdots$};
  \draw[red] (2.5,1.15) -- (4.5,1.15);
  \draw[red] (2.5,1.25) -- (4.5,1.25);

  \filldraw[red] (3,1.25) circle (.05);
  \draw[red] (3,1.25)--(3,0);
\filldraw[fill=white] (2.75,-.25) rectangle (3.25,.25);
\node at (3,0) {$\mathsf{X}$};

  \filldraw[red] (3.75,1.25) circle (.05);
  \draw[red] (3.75,1.25)--(3.75,0);
\filldraw[fill=white] (3.5,-.25) rectangle (4,.25);
\node at (3.75,0) {$\mathsf{P}$};

\filldraw[fill=white] (4.5,0.75) rectangle (5.5,1.25);
\node at (5,1) {$Enc$};

  \node at (7,0) {$\Enc(\mathsf{T}\ket{\psi})$};
\end{tikzpicture}
  \caption{Unitary version of the $T$-gadget described in \Cref{fig:encoded-t-gadget}. We split the gadget into two phases (which is denoted by a dashed line). In the first phase, all computation happens at the logical level. In the second phase, one qubit is decoded and some operations now happen at the {\em physical level}, which is highlighted in red.}\label{fig:new-t-gadget}
\end{figure}

     We now address the remaining part: computing using magic state gadgets. Unfortunately, we cannot  consider the encoded version of the well-known gadget to compute the $\sf T$-gate using $\ket{\T}$ magic states (see \Cref{fig:t-gadget,fig:encoded-t-gadget}), since for our applications, we require that our circuit be {\em unitary}. For that, we use the unitary version of \Cref{fig:encoded-t-gadget}, described in \Cref{fig:new-t-gadget}.

The following lemma shows how to simulate the reduced density matrices on the computation of $\T$-gadgets.

\begin{lemma}\label{L:T-gadget}
  Let $\calC$ be a $[[N,1,D]]$-QECC be the $k$-fold concatenated Steane code.
  Let $U_1,\ldots,U_{\ell}$ be the unitary circuit for the $T$-gadget described in \Cref{fig:new-t-gadget}.  Then there exists a classical algorithm  $\SimC^T$ that on input $0 \leq t \leq \ell$ and subset $S$,  $|S|\leq(D-1)/4$, we have that $\SimC^T(t,S)$ runs in time $poly(2^N)$ and outputs the classical description of an $|S|$-qubit state $\rho(t,S)$ such that for every qubit $\sigma$
    \begin{align}\label{eq:simulation-t}
      \rho(t,S) =    \Tr_{\overline{S}}\left((U_t \cdots U_1) \Enc( \kb{0} \otimes \sigma \otimes \kb{\T}) (U_{t} \cdots U_1)^\dagger\right).
    \end{align}
\end{lemma}
\begin{proof}

  Notice that if $t$ lies in the first phase of \Cref{fig:new-t-gadget}, \emph{i.e.}, the transversal application of any of the first two $\CNOT$ gadgets, then the simulation  is already covered by $\SimC^{\mathsf{Cliff}}$ of \Cref{L:cliff}.
    The challenging part here is when $t$ lies in the second phase of the gadget since  $i)$ one qubit is completely decoded in this computation and $ii)$ we are now applying controlled Cliffords whose control-qubit is not classical.
  For simplicity, we assume $\sigma = \kb{\psi}$ and the extension to mixed states follows by convexity.

    For $i \in \{1,2,3\}$, let $E_i$ be the set of qubits of the encoding of the $i$-th logical qubit, ordered from top to bottom. Let $\ket{\phi} = \CNOT_{2,1}\CNOT_{3,2}\ket{0}\ket{\psi}\ket{\T}$.
  Notice  that
  \[\Tr_{1}(\kb{\phi}) = \frac{1}{2}(\kb{0} \otimes \kb{\psi} + \kb{1} \otimes XP^\dagger\kb{\psi}PX),\]
  and by \Cref{L:logical-trace}, we have that
  \begin{align*}
    \Tr_{\overline{S \cup E_2 \cup E_3}}(\Enc(\kb{\phi})) = \frac{1}{2} \sum_{b \in \{0,1\}} \tau_{S \cap E_1} \otimes \Enc(\kb{b} \otimes (XP^\dagger)^b\kb{\psi}(PX)^b)\,,
    \end{align*}
    for some $\tau_{S \cap E_1}$ independent of the encoded state.

    In order to prove the simulatability of the second phase of \Cref{fig:new-t-gadget}, we consider two subcases: when $t$ lies in the encoding/decoding of the second qubit; and during the transversal application of the controlled $X$ and $P$ gates.

   When $t$ lies within the decoding/encoding of the second qubit, notice that the third qubit is not touched by any operation and therefore the third qubit is a codeword (for some unknown logical qubit). Using \Cref{L:logical-trace} again, we have that
    \begin{align*}
      \Tr_{E_2 \setminus S}(\sigma) =  \frac{1}{2}\sum_{b \in \{0,1\}} \tau_{S \cap E_1} \otimes \Enc(\kb{b}) \otimes \tau_{S \cap E_3},
    \end{align*}
    Therefore, in order to compute \Cref{eq:simulation-t}, $\SimC^T$ can compute the classical description of the state $\frac{1}{2} \sum_{b \in \{0,1\}} Enc(\kb{b})$ at the decoding stage  corresponding to $t$ and then compute its reduced density matrix on the qubits $S \cap E_2$ and output it (along with
  $\tau_{S \cap E_1}$ and $\tau_{S \cap E_3}$).
   An analogous argument holds if $t$ lies within the re-encoding of the second qubit.

   Finally, when $t$ lies in the transversal application of $X^b$ or $P^b$, notice that right after the decoding in the circuit described in \Cref{fig:new-t-gadget}, we have the state
    \begin{align*}
      \Tr_{\overline{S \cup E_2 \cup E_3}}(\Enc(\kb{\phi})) = \frac{1}{2} \sum_{b \in \{0,1\}} \tau_{S \cap E_1} \otimes \kb{b} \otimes \kb{0}^{N-1} \otimes \Enc(X^b(P^\dagger)^b\kb{\psi}P^bX^b)\,,
    \end{align*}
    and we want to apply a classically-controlled transversal Clifford gate with a {\em known} control qubit which is chosen uniformly at random. This case is covered by the second part of \Cref{L:cliff}, finishing the proof.
\end{proof}

\section{Zero-knowledge $\Xi$-protocol for $\QMA$}
\label{sec:xizk-protocol}
In this section, we show that simulatable proof systems lead to a zero-knowledge
protocols with a very simple proof structure, which can be classified in the
``commit-challenge-response'' framework.
As mentioned in \Cref{sec:background}, when all the messages are classical, such type of protocols are called
$\Sigma$-protocols. We extend this definition to the quantum setting
 by allowing the
first message to be a quantum state and in this case we call it a
\emph{$\Xi$-protocol}.

\begin{definition}[$\Xi$-protocol]
An $\Xi$-protocol consists of a three-round protocol between a
  prover and a verifier and it takes the following form:
\begin{itemize}
  \item[] \textbf{Commitment:} In the first round, the prover sends some initial quantum
    state.
  \item[] \textbf{Challenge:} In the second round, the verifier sends a uniformly random
    challenge $c \in [m]$.
  \item[] \textbf{Open:} The prover answers the challenge $c$ with some classical value.
\end{itemize}
\end{definition}

For simplicity we denote $\XiQZK$ as the class of problems that have a $\Xi$
computational quantum zero-knowledge proof system.

\subsection{Protocol}

\begin{figure}[H]
\begin{center}
\begin{tabular}{l | p{12cm}}
 \textbf{Notation} & \textbf{Meaning} \\
\hline
  $n$   & Number of the qubits in the \SimQMA{} proof  \\
  $k$   & Locality parameter \\
  $\Pi_c$ & POVM corresponding to a check of $\SimQMA$ proof system \\
  $S_c$ & Set of qubits on which $\Pi_c$ acts non-trivially \\
 $m$       &  Number of different $\SimQMA$ checks \\
  $\rho_S$   & Reduced density matrix of the proof on set $S$ of qubits for $|S|
  = k$ \\
  $\tau$ & Quantum state that is supposed to pass the checks and be consistent with all local
  density matrices up to negligible error \\
  $\sigma(c)$ & $\rho_c^{\reg{S_c}} \otimes \kb{0}^{\reg{\overline{S_c}}}$ \\
  $\zeta$    & Side-information of a malicious verifier \\
  $\tilde{\phi}_{a,b}$  &  $\X^a\Z^b\phi \X^a\Z^b$, for a $q$-qubit quantum state $\phi$ and $a,b \in
  \01^q$
\end{tabular}
\end{center}
\caption{Notation reference}
\label{fig:notation}
\end{figure}

We describe in \Cref{fig:sigma-ZK} the zero knowledge $\Xi$ protocol for $\QMA$,
whose informal description was given in \Cref{sec:techniques}.

\newcommand{\definitionsProt}{
    Let $A = (\ayes,\ano)$ be  a problem in $k$-$\SimQMA$ with soundness $\delta$, $x \in A$,
    $\{\Pi_c\}$ be the set of POVMs for $x$, and $\tau$ be a (supposed)
    simulatable witness
    for $x$.
}
\newcommand{\definitionsPoQ}{
    Let $A = (\ayes,\ano)$ to be  a problem in $k$-$\SimQMA$, $x \in A$,
    $\{\Pi_c\}$ be the set of POVMs for $x$
    and $P^*$
    be a prover that makes the verifier accept with probability at least
    $\kappa(n) \geq 1 - \frac{1}{2m^2}$ in the $\Xi$-protocol of
    \Cref{fig:sigma-ZK}.
}
\newcommand{\definitionsPoQSeq}{
    Let $A = (\ayes,\ano)$ to be  a problem in $k$-$\SimQMA$, $x \in A$,
    $\{\Pi_c\}$ be the set of POVMs for $x$
    and $P^*$
    be a prover that makes the verifier accept with probability at least
    $\kappa(n) \geq \left(1 - \frac{1}{2m^2}\right)^{\ell}$ in the $\ell$-fold
    sequential repetition of $\Xi$-protocol of
    \Cref{fig:sigma-ZK}.
}

\newcommand{\definitionsSim}{
    Let $A = (\ayes,\ano)$ to be a problem in $k$-$\SimQMA$, $x \in \ayes$, and
    $\rho_c = \rho(x,S_c)$ be the local density matrix
    of a simulatable witness $\tau$ for $x$ on the qubits corresponding to $\Pi_c$.
}

\begin{figure}[H]
\rule[1ex]{\textwidth}{0.5pt}
\definitionsProt
\begin{enumerate}
  \item Prover picks $a,b \inr \01^{n}$ and $r \inr \cR$,
    $\cR$ is all of the possible randomness needed to commit to $2n$ bits.
  \item Prover sends $ \widetilde{\tau}_{a,b}\otimes\kb{\commit{a,b}{r}}$, where
    $\commit{a,b}{r}$ is the commitment to each bit of
    $a$ and~$b$.
  \item The verifier sends $c\inr[m]$.
  \item The prover opens the commitment for
    $a|_{S_c}$ and $b|_{S_c}$, where $S_c$ is the set of qubits on which $\Pi_c$
    acts non-trivially.
  \item If the commitments do not open, the verifier rejects.
  \item The verifier measures
    $\X^{a|_{S_c}}\Z^{a|_{S_c}}\tilde{\tau}_{a,b}\Z^{a|_{S_c}}\X^{a|_{S_c}}$ with
    POVMs $\{\Pi_c, \Id - \Pi_c\}$, and accepts if and only if the outcome is~$\Pi_c$.
\end{enumerate}
\rule[2ex]{\textwidth}{0.5pt}\vspace{-.5cm}
    \caption{Zero-knowledge $\Xi$-protocol for \SimQMA{}.}\label{fig:sigma-ZK}
\end{figure}

\subsection{Computational zero-knowledge proof for $\QMA$}

The goal of the section is to prove that every language in QMA has a
$\Xi$-protocol that is a quantum computational zero-knowledge proof system if we
assume that the commitment used in~\Cref{fig:sigma-ZK} is
computationally hiding and unconditionally binding.

We first state two lemmas that will be proved in
\Cref{sec:xi-compl-sound,sec:xi-zk}, respectively.

\newcommand{\bodyxicompl}{
  The protocol in \Cref{fig:sigma-ZK} has completeness $1 - \negl(n)$ and
  soundness $\delta$.}
\begin{lemma}
  \label{lem:completeness-soundness}
  \bodyxicompl
\end{lemma}

\newcommand{\bodyxizk}{
  The protocol in \Cref{fig:sigma-ZK} is computational zero-knowledge.}
\begin{lemma}
\label{lem:comput-ZK}
\bodyxizk
\end{lemma}

We can then state the main theorem of this section.

\begin{theorem}
 $\QMA \subseteq \XiQZK$.
\end{theorem}
\begin{proof}
  Direct from \Cref{lem:simulatable-proof,lem:completeness-soundness,lem:comput-ZK}.
\end{proof}

\subsubsection{Proof of \Cref{lem:completeness-soundness}}
\label{sec:xi-compl-sound}
\newtheorem*{thm:repeat-xi-compl}{\Cref{lem:completeness-soundness} (restated)}
\begin{thm:repeat-xi-compl}
\bodyxicompl
\end{thm:repeat-xi-compl}
\begin{proof}
  By \Cref{def:simulatable-proof}, if $x \in \ayes$, the prover can follow the protocol honestly with some $\tau$ that is
  consistent with all the POVMs and the local density matrices. In this case, the acceptance probability is exponentially close to $1$.

  \medskip
  Let us now analyze the case for $x \in \ano$.
  Let $\psi \otimes \kb{z}$ be the state sent by the prover in the first
  message, where $\psi$ is supposed to be the copies of the
  one-time padded state that is consistent with the POVMs and the reduced
  density matrices, and
  $z$ is the commitment to the one-time pad keys.
  We assume, without loss of generality, that  $\ket{z}$ is a classical
  value, since the verifier can measure it as soon as she receives it, and the
  prover can send the $z$ that maximizes the acceptance probability.

  For challenge $c$, the prover answers with $\ket{w_c}$, where again we assume
  to be a classical value for the same reasons as above.
  Since the commitment scheme is unconditionally binding, we can define the
  strings
  $a, b \in \01^{n}$ to be the string
  containing the {\em unique} bits that could be open for the corrected committed
  bits, or $0$ if the commitment is defective.

  Let $S_c \subseteq [n]$ be defined as in \Cref{fig:sigma-ZK}.
  Notice that $\ket{w_c}$ is supposed to be the opening of bits
  of $a$ and~$b$ in the subset $S_c$.
  Let $D_c$ be the event that $w_c$ is the correct opening
  for {\em all} of such bits, and~$\textbf{1}_{D_c}$  be the
  indicator variable for such event.

  We have then that the acceptance probability is
  \begin{align}
    &\frac{1}{m}\sum_{c \in [m]} \textbf{1}_{D_c} \tr{\Pi_c
    \X^{a|_{S_c}}\Z^{b|_{S_c}}\psi
    \Z^{b|_{S_c}}\X^{a|_{S_c}}} \label{eq:soundness-first}
    \\
    & \leq
    \frac{1}{m}\sum_{c \in [m]} \tr{\Pi_c
    \X^{a|_{S_c}}\Z^{b|_{S_c}}\psi
    \Z^{b|_{S_c}}\X^{a|_{S_c}}} \nonumber \\
    &  =
    \frac{1}{m}\sum_{c \in [m]} \tr{\Pi_c
    \X^{a}\Z^{b}\psi \Z^{b}
    \X^{a}}\label{eq:soundness-last}
    \\
    & \leq
    \max_{\phi} \frac{1}{m} \sum_c \tr{\Pi_c \phi} \nonumber \\
    &\leq \delta. \nonumber
  \end{align}
  where in the equality  we use the fact that $\Pi_c$ only acts on the qubits in
  $S_c$, and the last inequality follows since $x \in \ano$
  and from \Cref{def:simulatable-proof} the $\SimQMA$ protocol has soundness $\delta$.
\end{proof}

\subsubsection{Proof of \Cref{lem:comput-ZK}}
\label{sec:xi-zk}
We prove now the zero knowledge property of the protocol.

Before presenting the simulator, let us analyze how the verification algorithm
behaves. We can assume, without loss of generality that the verifier is composed
of two verification algorithms $\hat{V}_1$ and $\hat{V}_2$.

For $\hat{V_1}$, since the classical part of the message can be
copied and the challenge sent by the verifier is measured by the prover, we can assume
$\hat{V}_1$ acts like the following
\begin{align}
  & \sum_{a,b,r} \hat{V}_1\left(\tilde{\rho}_{a,b} \otimes \kb{\commit{a,b}{r}}
  \otimes \zeta\right)\hat{V}_1^\dagger \\
  & = \sum_{a,b,r,c}p_{\rho,a,b,c,r} \phi_{\rho,a,b,c,r} \otimes \kb{\commit{a,b}{r}}\otimes \kb{c}
  ,\label{eq:structure-malicious}
\end{align}
   where $\sum_{a, b, c,r} p_{\rho, a,b,c,r} = 1$ and we have traced-out the
   copy of $\kb{c}$ that was sent to the prover (and measured).

   The message $\ket{c}$ is sent to the prover, who answers then with some value
   $\ket{o_c}$, \emph{i.e.}, the opening of the commitments corresponding to the
   challenge $c$.

   The verifier then outputs
\begin{align}
  \sum_{a,b,r,c}p_{\rho,a,b,c,r}
  \hat{V}_2\left(\phi_{\rho,a,b,c,r} \otimes
  \kb{\commit{a,b}{r}}\otimes \kb{c} \otimes \kb{o_c} \right)\hat{V}_2^\dagger.
  \label{eq:structure-malicious-2}
\end{align}

\begin{figure}[H]
\rule[1ex]{\textwidth}{0.5pt}
\definitionsSim
  \begin{enumerate}
    \item Pick $c\inr [m]$, $a,b \inr \01^{n}$, $r \inr \cR$
    \item Create the state
$\tilde{\sigma(c)}_{a,b} \otimes
      \kb{\commit{a,b}{r}} \otimes \zeta$,
      where   $\sigma(c) = \rho_c^{\reg{S_c}}\otimes \kb{0}^{\reg{\overline{S_c}}}$
    \item Run $\hat{V}_1$ on  $\tilde{\sigma(c)}_{a,b} \otimes
      \kb{\commit{a,b}{r}} \otimes \zeta$
    \item  Measure the last register in the computational basis and abort
      if it is not $\ket{c}$
    \item Otherwise, append the register $\ket{o_c}$, apply $\hat{V}_2$ and
      output the result.
  \end{enumerate}
\rule[2ex]{\textwidth}{0.5pt}\vspace{-.5cm}
  \caption{Simulator $\SimZK$ for ZK $\Xi$-protocol for QMA}
\label{fig:simulator-qzk}
\end{figure}

\medskip
In order to show zero-knowledge, we start by showing that
for a fixed $c$,  $\sum_{a, b, r} p_{\sigma,
a,b,c,r}$ is independent of~$\sigma$, up to negligible  factors, if the commitment
scheme is hiding. We denote $R = |\mathcal{R}|$.

\begin{lemma}\label{lem:independence-probability}
  Let $p_c = \frac{1}{R}\sum_{r} p_{I, 0,0,c,r}$. Then for any $\sigma$, we have
  that
  \[\left|\frac{1}{2^{2n} R} \sum_{a, b, r} p_{\sigma, a,b,c,r} -  p_c\right|
  \leq \negl(n),\]
  where the probabilities are defined as \Cref{eq:structure-malicious} for the
  polynomial-time adversary $\hat{V}_1$.
\end{lemma}
\begin{proof}
  Let us suppose that there exist some state $\rho$, a challenge $c$ and a
  polynomial $q$ such that
  \begin{align}
  \label{eq:difference-probability}
    \left|p_c - \left(\frac{1}{2^{2n}
    R} \sum_{a, b, r} p_{\sigma, a,b,c,r}\right)\right| \geq q(n).
  \end{align}
  Then it is possible to distinguish the states
  \[\frac{1}{2^{2n}
    R}
    \sum_{a,b,r} \tilde{\sigma}_{a,b}  \otimes \kb{\commit{a,b}{r}} \text{\quad and
  \quad }
  \frac{1}{R} \sum_{r} \Id \otimes \kb{\commit{0,0}{r}}\]
  by appending
  $\zeta$, applying $\hat{V}_1$ and measuring the challenge
  register in the computational basis.

  However, since the commitment scheme is computationally hiding, we have that
\begin{align}
  & \frac{1}{2^{2n} R} \sum_{a,b,r} \tilde{\sigma}_{a,b} \otimes
  \kb{\commit{a,b}{r}} \nonumber \\
  & \approx_c \frac{1}{2^{2n} R} \sum_{a,b,r} \Id \otimes \kb{\commit{a,b}{r}}
  \label{eq:approx-comp-1} \\
  & \approx_c \frac{1}{R} \sum_{r} \Id \otimes \kb{\commit{0,0}{r}}
  \label{eq:approx-comp-2}.
\end{align}
and therefore these states are indistinguishable. We conclude that the assumption in
  \Cref{eq:difference-probability} is false.
  \end{proof}

\begin{lemma}\label{lem:simulator}
  The Simulator  described in \Cref{fig:simulator-qzk} does not abort
  with
  probability at least $\frac{1}{m} - \negl(n)$. In this case, its output is
  $\negl(n)$-close to
  \[\frac{1}{2^{2n}R}\sum_{a,b,c,r}\hat{V}_2\left(\left(I\otimes
  \kb{c}\right)\hat{V}_1\left(\tilde{\sigma(c)}_{a,b}
 \otimes \kb{\commit{a,b}{r}} \otimes
  \zeta\right)\hat{V}_1^\dagger(I\otimes \kb{c})  \otimes \kb{o_{c}}\right)\hat{V}_2^\dagger.\]
\end{lemma}
\begin{proof}
  The state of $\SimZK$ after step $2$ is
  \begin{align}\label{eq:sim-first-step}
    \frac{1}{2^{2n} R m}\sum_{a,b,c,r}
\tilde{\sigma(c)}_{a,b}
\otimes
  \kb{\commit{a,b}{r}}
\otimes \zeta  \otimes \kb{c}
  \end{align}
   $\SimZK$ runs $\hat{V}_1$ on the first three register of the state in  \Cref{eq:sim-first-step}, resulting in
  \begin{align}
    &\frac{1}{2^{2n} R m}\sum_{a,b,c,r} \hat{V}_1\left(
\tilde{\sigma(c)}_{a,b}
    \otimes \kb{\commit{a,b}{r}}  \otimes \zeta\right)
    \hat{V}_1^\dagger \otimes \kb{c} \nonumber
\\
    &=\frac{1}{2^{2n} R m}
    \sum_{a,b,c,c',r} p_{\sigma(c),
    a,b,c',r} \phi_{\sigma(c),a,b,c',r} \otimes
  \kb{\commit{a,b}{r}} \otimes
  \kb{c'} \otimes
  \kb{c}\label{eq:middle-simulator}
  \end{align}

  Notice that $\SimZK$ does not abort when $c =
  c'$, and this event happens with probability
  \begin{align}\label{eq:approximation}
    \frac{1}{m} \sum_{c}
  \frac{1}{2^{2n} R}\sum_{a,b,r}p_{\sigma(c), a,b,c,r}
    \geq \frac{1}{m} \sum_{c} (p_{c} - \negl(n)) =\frac{1}{m} - \negl(n),
  \end{align}
where the inequality follows from  \Cref{lem:independence-probability} and the
  equality from the fact that $\sum_{c}
  p_{c} = 1$.

In order to provide the output of the simulator, conditioned that he did not
  abort, we
  post-select in \Cref{eq:middle-simulator} the event that $c = c'$, which gives
  us
  \begin{align}
    &\frac{1}{2^{2n} R m p_{succ}}\sum_{a,b,c,r}
  (\Id \otimes \kb{c}) \hat{V}_1\left(
\tilde{\sigma(c)}_{a,b}
    \otimes \kb{\commit{a,b}{r}}  \otimes \zeta\right)
    \hat{V}_1^\dagger (\Id \otimes \kb{c}) \otimes \kb{c} \nonumber \\
    & \approx_c
    \frac{1}{2^{2n} R}\sum_{a,b,c,r}
  (\Id \otimes \kb{c}) \hat{V}_1\left(
\tilde{\sigma(c)}_{a,b}
    \otimes \kb{\commit{a,b}{r}}  \otimes \zeta\right)
    \hat{V}_1^\dagger (\Id \otimes \kb{c}) \otimes \kb{c},
    \label{eq:approx-comp-3}
  \end{align}
  where
  we set $p_{succ} =
  \frac{1}{2^{2n}Rm}\sum_{a,b,c,r}p_{\sigma(c), a,b,c,r}$ to be
  the probability that $\SimZK$ does not abort and the
  approximation holds by \Cref{eq:approximation}.

  Finally, $\SimZK$ only needs to append the last register with $\ket{o_c}$,
  which can be performed efficiently given $a$, $b$ and $c$, and apply
  $\hat{V}_2$.
\end{proof}

\begin{lemma}\label{lem:honest-simulation}
  The output of a  simulator that does not abort is computationally indistinguishable from
  the output of the malicious verifier in the protocol in \Cref{fig:sigma-ZK}.
\end{lemma}
\begin{proof}
    In the real protocol, the first message sent by the prover is
   \[\frac{1}{2^{2n} R} \sum_{a,b,r} \tilde{\tau}_{a,b} \otimes
   \kb{\commit{a,b}{r}},\]
   and the verifier applies $\hat{V}_1$ on the message sent by the prover and
   the side-information $\zeta$, resulting in the state
   \[\frac{1}{2^{2n} R} \sum_{a,b,c,r}
  (\Id \otimes \kb{c})\hat{V}_1 \left(
\tilde{\tau}_{a,b} \otimes
  \kb{\commit{a,b}{r}} \otimes \zeta\right) \hat{V}_1^\dagger
(\Id \otimes \kb{c}).\]

On challenge $\ket{c}$, the prover then answers with $\ket{o_c}$,
   the openings of the corresponding commitments.
   The  verifier then applies $\hat{V}_2$, and outputs
   \[\frac{1}{2^{2n} R} \sum_{a,b,c,r}
  \hat{V}_2\left((\Id \otimes \kb{c})\hat{V}_1 (
\tilde{\tau}_{a,b} \otimes \kb{\commit{a,b}{r}} \otimes \zeta) \hat{V}_1^\dagger
  (\Id \otimes \kb{c}) \otimes \kb{o_c}\right)\hat{V}_2.\]

  We show now that this state is indistinguishable from the state that is output
 by the simulator, proved in \Cref{lem:simulator}. For simplicity, let
  $\xi_{a,b,r}= \kb{\commit{a,b}{r}} \otimes \zeta$. Up
  to normalization factors, we have that

  \begin{align}
    & \sum_{a,b,c,r} \hat{V}_2\left((\Id \otimes \kb{c})
    \hat{V}_1 \left(
\tilde{\tau}_{a,b} \otimes \xi_{a,b,r}
    \right) \hat{V}_1^\dagger (\Id \otimes \kb{c})
 \kb{o_c}\right)\hat{V}_2^\dagger \nonumber\\
    & \approx_c
\sum_{a,b,c,r} \hat{V}_2\left((\Id \otimes \kb{c}) \hat{V}_1 \left(
    \widetilde{\left(\Tr_{\overline{S_c}}(\tau)^{\reg{S_c}}\otimes
    I^{\reg{\overline{S_c}}}\right)}_{a,b} \otimes
 \xi_{a,b,r}
    \right) \hat{V}_1^\dagger(\Id \otimes \kb{c}) \otimes
    \kb{o_c}\right)\hat{V}_2^\dagger \label{eq:approx-comp-4}\\
    &\approx_{s} \sum_{a,b,c,r}\hat{V}_2\left((\Id \otimes \kb{c})\hat{V}_1
    \left(\tilde{\sigma(c)}_{a,b} \otimes
 \xi_{a,b,r}\right) \hat{V}_1^\dagger (\Id \otimes \kb{c})\otimes
    \kb{o_c}\right)\hat{V}_2^\dagger
  \label{eq:final-hybrid}
  \end{align}
  where in the first approximation we use the fact that the commitments of
  $a|_{\overline{S_c}}$
  and
  $b|_{\overline{S_c}}$ are never revealed and that the commitment is
  computationally hiding, and  the second approximation holds since
  we assume that $\trNorm{\Tr_{\overline{S_c}}(\sigma) - \rho_c} \leq \negl(n)$ by \Cref{def:simulatable-proof}.

  By the definition of
  $\xi_{a,b,r}$ and \Cref{lem:simulator}, the state of \Cref{eq:final-hybrid} is
  $\negl(n)$-close to the output of the simulator.
\end{proof}

We are finally ready to prove \Cref{lem:comput-ZK}.
\newtheorem*{thm:repeat-xi-zk}{\Cref{lem:comput-ZK} (restated)}
\begin{thm:repeat-xi-zk}
\bodyxizk
\end{thm:repeat-xi-zk}
\begin{proof}
  Notice that from \Cref{lem:independence-probability,lem:rewinding}, there
  exists a quantum algorithm $\SimZK'$ that runs in time
  \[O\left(m \poly(n) \left(time\left(\hat{V}_1\right)+time\left(\hat{V}_2\right)\right))\right)\]
  whose output is $\negl(n)$-close to the output of $\SimZK$, conditioned on
  not aborting.
  From \Cref{lem:honest-simulation}, the output of $\SimZK'$ is computationally
  indistinguishable from the run of the real protocol, and therefore it can be
  used as the simulator, finishing our proof.
\end{proof}

\subsection{Decreasing the soundness error}
\label{sec:parallel-repetition}

We remark that unlike the protocol in ~\cite{GSY19}, we do not know how to
show parallel repetition for our protocol.
The problem here is that for an $\ell$-fold version of our protocol, the simulator,
as in~\Cref{fig:simulator-qzk}, would need to correctly answer the question for each
one of these $\ell$ copies of the game, what would happen with probability
$\frac{1}{m^\ell}$, and therefore the rewinding technique would have an exponential
cost in~$\ell$.

However, as in the classical case, we can show that our protocol accepts {\em
sequential repetition}, since the guess for each of the iterations is performed
independently.

\begin{lemma}
  Consider  the $\ell$-fold sequential repetition of the QZK $\Xi$-protocol, where
  $1 \leq \ell = poly(n)$ and
  the verifier accepts if and only if each sequential run  accepts.
  Then this is a quantum zero-knowledge protocol for \SimQMA with completeness $1-\negl(|x|)$ and soundness
  $O\left(\delta(|x|)^\ell\right)$.
\end{lemma}
\begin{proof}
  The completeness and soundness properties hold trivially.

  Let us now argue about the zero-knowledge property. Notice that an honest
  prover $P$ has an $\ell$-fold tensor product of the honest witness, and uses a single copy per iteration.
  In this case, we can run the simulator $\ell$ times {\em sequentially}, using the
  output of the $i^\text{th}$ run as the side-information of the $i+1^\text{st}$ run.
\end{proof}

\section{Proofs of quantum knowledge}
\label{sec:PoQ}
In this section, we define a \emph{Proof of Quantum Knowledge} (\Cref{sec:def-poq}) and
and then prove that the
Zero-knowledge protocols presented in the previous sections
(\Cref{sec:proof-poq}) satisfy this new definition.

\subsection{Definition}
\label{sec:def-poq}
The content of this subsection was written in collaboration with Andrea
Coladangelo, Thomas Vidick and Tina Zhang and a similar version of it also appears in their
concurrent and independent work~\cite{CVZ19}.

\medskip

A \emph{Proof of Knowledge (PoK)} is an interactive proof system
for some relation $R$ such that if
the verifier accepts some input $x$ with high enough
probability, then she is convinced that the prover knows some witness $w$ such
that $(x,w) \in R$.
This notion is formalized by requiring the existence of an efficient \emph{extractor}~$K$
that is able to output a witness for $x$ when $K$ is given oracle access to
the prover (and is able to rewind his actions).

\begin{definition}[Classical Proof of Knowledge~\cite{BG93}]
  \label{def:pok}
  Let $R \subseteq \mathcal{X} \times \mathcal{Y}$ be a relation.  A proof
  system $(P,V)$ for $R$ is a Proof of Knowledge for $R$ with knowledge error $\kappa$ if there
  exists a polynomial $p > 0$ and a polynomial-time machine $K$
  such that for any classical interactive machine $P^*$
  that makes $V$ accept some instance $x$ of size $n$ with probability at least
  $\eps > \kappa(n)$, we have
  \[Pr\left[\left(x, K^{P^*(x,y)}(x)\right) \in R \right] \geq p\left((\eps - \kappa(n)),
  \frac{1}{n}\right). \]
\end{definition}
In the definition, $y$ corresponds to the side-information that $P^*$ has,
possibly including some $w$ such that $(x,w) \in R$.

\medskip

PoKs were originally defined only considering classical adversaries, and
this notion was first studied in the quantum setting by Unruh~\cite{Unr12}.
The first issue that arises in the quantum setting is which type of query $K$ could be able
to perform. To solve this, we
assume that $P^*$ always performs some unitary operation $U$. Notice that this can
be done without loss of generality since
\begin{enumerate*}[label=(\roman*)]
\item  we can consider the
purification of the prover,
\item  all the measurements can be
performed coherently, and
\item  $P^*$ can keep track of the round of communication
in some internal register and $U$ implicitly controls on this value.
\end{enumerate*}
Then,
the quantum extractor $K$ has oracle access to $P^*$ by
performing $U$ and $U^\dagger$ on the message register and private register
of $P^*$, but $K$ has no direct access to the latter.
We denote the extractor
$K$ with such an oracle access to $P^*$ as $K^{\ket{P^*(x,\rho)}}$, where here
$\rho$ is the (quantum) side-information held by $P^*$.

\begin{definition}[Quantum Proof of Knowledge~\cite{Unr12}]
  \label{def:qpok}
  Let $R \subseteq \mathcal{X} \times \mathcal{Y}$ be a relation.  A proof
  system $(P,V)$ for $R$ is a Quantum Proof of Knowledge for $R$ with knowledge error $\kappa$
  if there
  exists a polynomial $p > 0$ and a quantum polynomial-time machine $K$
  such that for any quantum interactive machine  $P^*$
  that makes $V$ accept some instance $x$ of size $n$ with probability at least
  $\eps > \kappa(n)$, then
  \[Pr\left[\left(x, K^{\ket{P^*(x, \rho)}}(x)\right) \in R \right] \geq p\left((\eps - \kappa(n)),
  \frac{1}{n}\right). \]
\end{definition}

\begin{remark}\label{rem:repetition-extractor}
  In the fully classical case of \Cref{def:pok}, the extractor could repeat the procedure
  $\poly((\eps - \kappa(n))$ times in order to increase the success probability.
  We notice that this is not known to be possible for a general quantum $P^*$, since the final measurement to extract the witness would possibly
  disturb the internal state of $P^*$, making it impossible to simulate the
  side-information that $P^*$ had originally in the subsequent simulations.
\end{remark}

\medskip
We finally move on to the full quantum setting, where we want a {\em Proof of
Quantum Knowledge} (PoQ). Here, at the end of the protocol, we want the verifier
to be convinced that the prover has a {\em quantum witness} for the input $x$.

The first challenge is defining the notion a ``relation'' between the input $x$
and some quantum state $\ket{\psi}$. Classically, we implicitly consider relations for \NP
languages by fixing some verification algorithm $V$ for it, and
defining $(x,w) \in R$ if and only if $V$ accepts the input $x$ with the witness $w$.

Quantumly, the situation is a bit more delicate, since a witness $\ket{\psi}$
leads to acceptance probability $Pr[Q(x, \ket{\psi})=1]$. This issue also
appears with probabilistic complexity classes such as MA, and
we can solve it
by fixing some parameter $\gamma$ and defining the relation to contain $(x,
\ket{\psi})$ for all quantum states $\ket{\psi}$
that lead to acceptance probability at least $\gamma$. The difference here is
that we need also to consider the  {\em mixture} of such quantum states, since they are
also valid witnesses in a QMA protocol.
Therefore, fixing some quantum verifier $Q$ and
$\alpha$,  we define a {\em quantum relation} as  follows
\[R_{Q,\gamma} = \{(x,\sigma) :  Q\text{ accepts } (x, \sigma) \text{ with probability at least } \gamma\}.
\]
Notice that with $R_{Q,\gamma}$, we implicitly define subspaces
$\{\mathcal{S}_x\}_x$ such that
$(x,\sigma) \in R_{Q,\gamma}$ if and only if $\sigma \in \mathcal{S}_x$.

With this in hand, we can define a $\QMA$-relation.

\begin{definition}[$\QMA$-relation]
  \label{def:qma-relation}
  Let $A = (\ayes,\ano)$ be a problem in \QMA{}, and let $Q$ be an associated
  quantum polynomial-time verification algorithm (which takes as input an
  instance and a witness), with completeness $\alpha$ and soundness $\beta$.
  Then, we say that $R_{Q, \gamma}$ is a \QMA{}-relation with completeness $\alpha$ and soundness $\beta$ for the problem $A$.
  In particular, we have that for $x \in \ayes$, there exists some $\ket{\psi}$ such that $(x,\ket{\psi}) \in R_{Q,\alpha}$ and for $x \in \ano$, for every $\rho$ and  $\eps > 0$ it holds that $(x,\rho) \not\in R_{Q,\beta+\eps}$.
  \end{definition}

We can finally define a Proof of Quantum Knowledge.

\begin{definition}[Proof of Quantum Knowledge]
  \label{def:poq-single}
  Let $R_{Q,\gamma}$ be a \QMA{} relation.
  A proof
  system $(P,V)$ is a Proof of Quantum Knowledge for $R_{Q,\gamma}$ with knowledge
  error
  $\kappa(n) > 0$ and quality $q$, if there exists a polynomial $p > 0$ and a  polynomial-time machine $K$ such that for any quantum interactive machine $P^*$ that
  makes $V$ accept some instance $x$ of size $n$ with probability at least
  $\eps > \kappa(n)$, we have
  \begin{enumerate}
    \item $K^{\ket{P^*(x, \rho)}}(x)$ aborts with probability at most $p\left((\eps - \kappa(n)),
  \frac{1}{n}\right)$, and
    \item if $K^{\ket{P^*(x, \rho)}}(x)$ does not abort, it outputs the quantum
      state $\phi$ such  that $\left(x, \phi \right) \in R_{Q,q(\eps,\frac{1}{n})}$.
  \end{enumerate}
\end{definition}
 
\subsection{Proof of quantum-knowledge for our $\Xi$-protocol}
\label{sec:proof-poq}

We show  now that the $\Xi$-protocol of \Cref{fig:sigma-ZK} is a Proof of Quantum Knowledge with
knowledge error inverse polynomially close to~$1$.

\begin{figure}[H]
\rule[1ex]{\textwidth}{0.5pt}
  \definitionsPoQ
  \begin{enumerate}
    \item Run $P^*$ and store the first message $\psi \otimes \kb{z}$
    \item For every challenge $c$
      \begin{enumerate}
    \item Simulate $P^*$ on challenge $c$
    \item Check (coherently) if the answer correctly opens the committed value, if not abort
    \item Copy the opening of the committed values
    \item Run $P^*$ backwards on challenge $c$
    \end{enumerate}
  \item Let $a, b \in \01^{n}$ be the opened strings
    \item Output $\X^a\Z^b \psi \Z^b\X^a$
\end{enumerate}
\rule[2ex]{\textwidth}{0.5pt}\vspace{-.5cm}
\caption{Single-shot Knowledge extractor $K$}
  \label{fig:extractor}
\end{figure}

\begin{lemma}\label{lem:single-shot}
  Let $A$ and $\{\Pi_c\}$ be as defined in \Cref{fig:extractor} and
  $Q$ be the verification algorithm for $A$ that consists of
  picking $c \in [m]$ uniformly at random and measuring the
  provided witness with $\Pi_c$.
  Let $\kappa(n) = 1 - \frac{1}{2m^2}$ and $K$ be the $\poly(n)$-time extractor defined in
  \Cref{fig:extractor}. If
  a quantum interactive machine  $P^*$
  makes $V$ accept the instance $x$ of size $n$ with probability at least
  $\eps := 1 - \delta> \kappa(n)$, then we have that
  \begin{enumerate}
    \item $K^{\ket{P^*(x, \rho)}}(x)$ aborts with probability at most $1 - m^2\delta$, and
    \item if $K^{\ket{P^*(x, \rho)}}(x)$ does not abort, it outputs the quantum
      state $\phi$ such  that $\left(x, \phi \right) \in R_{Q,1-\delta -
  m^2\delta}$.
  \end{enumerate}
\end{lemma}
\begin{proof}
  Let $a, b \in \01^{2 n}$ be the unique values that can be
  opened by the classical value $\ket{z}$ sent by the prover (or zeroes if the
  commitments are mal-formed).
  Doing the same calculations of
  Equations~\eqref{eq:soundness-first} to~\eqref{eq:soundness-last}, and
  assuming that the acceptance probability of the original protocol is at least
  $1 - \delta$, it follows that
  \begin{align}
    \label{eq:original-soundness}
    \left(x, \X^{a}\Z^{b}\psi \X^{a}\Z^{b}\right) \in
  \mathcal{A}_{Q,1-\delta}.
  \end{align}

  Our goal now is show how to retrieve the values $a$ and $b$,
  without damaging the quantum state $\psi$ too much.

  Let $\mu_{VMP}$ be the state shared by the verifier and prover after
  the commitment phase. Since the message (for honest verifiers) is always a classical value, we can
  model $P^*$'s behaviour with the unitary $U_{c}$ performed by him on
  challenge $c$. Let also $\Pi$ be the projection of $V$ onto the acceptance subspace,
  $\hat{\Pi}_c = U_c^\dagger \Pi U_c$ be the operation that performs the
  prover's unitary for challenge $c$, performs the measurement of the verifier,
  and then undoes the prover's unitary.

  Given that $P^*$ makes $V$ accept with probability at least $1-\delta$ and each
  challenge is picked with probability $\frac{1}{m}$, it follows
  that for any challenge $c$, $V$ accepts with probability at least
  $1 - m\delta$, otherwise the acceptance probability would be strictly less than
  $1- \delta$. In other words, it follows that for every $c$, we have that
  \[\Tr\left(\hat{\Pi}_c \mu \hat{\Pi}_c\right) \geq 1 - m\delta,\]
   which implies that
  \begin{align}\label{eq:poq-post-measured}
    \Tr\left(\hat{\mu}\right) \geq
    1 - m^2\delta, \quad \text{ for } \hat{\mu} = \hat{\Pi}_m \ldots \hat{\Pi}_1\mu\hat{\Pi}_1 \ldots \hat{\Pi}_m
  \end{align}
  Let $O$ be the register where the verifier
  holds the original state $\psi$ and $\phi =\Tr_{\overline{O}}\left(\hat{\mu}\right)$. We have that
  \begin{align}\label{eq:distance-middle-step}
    D\left(\psi, \phi \right)
    = D\left(\Tr_{\overline{O}}(\mu), \Tr_{\overline{O}}\left(\hat{\mu}\right) \right)
    \leq D\left(\mu, \hat{\mu}\right) \leq m^2\delta.
  \end{align}
  where in the equality we use the definition of $\psi$, $\phi$ and the register
  $O$, the first inequality follows since trace distance is contractive under
  CPTP maps and the last inequality holds by~\Cref{eq:poq-post-measured}.

  Notice that the decision of an abort by $K$ is strictly less restrictive than the rejections of an honest
  verifier, since the verifier also tests that the commitment correctly opens.
  In this case, we can conclude that $K$ does not abort with probability at
  least $1 - m^2\delta$, proving item $1$ of the statement.

  Then, if we condition on the event that $K$ does not abort, all the
  committed information is opened and since the commitment is
  biding, $K$ holds the unique values
  $a,b \in \01^{n}$ that can be opened from $z$. $K$ finishes by
  outputting
  $\X^{a}\Z^{b}\phi \Z^{b}\X^{a}$.
  It follows from \Cref{eq:distance-middle-step} and the fact that trace
  distance is preserved under unitary operations, we have
  \[
    D\left(\X^{a}\Z^{b}\psi \X^{a}\Z^{b},
    \X^{a}\Z^{b}\phi \X^{a}\Z^{b}\right)
    \leq  m^2\delta,
  \]
  and therefore
  \[
    \left(x, \X^{a}\Z^{b}\phi \X^{a}\Z^{b}\right) \in R_{Q,
    1-\delta-m^2\delta},
  \]
  which finishes the proof of item $2$ of the statement.
\end{proof}

\subsubsection{Sequential repetition}
In the previous section, we have the quantum extractor that works if the
knowledge error is very high, namely inverse polynomially close to $1$. We
show here how to decrease the knowledge leakage, by considering the sequential
repetition of the $\Xi$-protocol.

\begin{figure}[H]
\rule[1ex]{\textwidth}{0.5pt}
  \definitionsPoQSeq{}
  \begin{enumerate}
    \item $K$ chooses $i \in [\ell]$
    \item For $0 < j < i$
    \begin{enumerate}
      \item Pick $c_j$ uniformly at random and put it in the message register
      \item Run $P^*$ with $c_j$
      \item If $V$ would reject, abort
    \end{enumerate}
    \item Run Single-shot extractor for the $i$-th game
\end{enumerate}
\rule[2ex]{\textwidth}{0.5pt}\vspace{-.5cm}
\caption{Knowledge extractor $K$}
  \label{fig:extractor-sequential}
\end{figure}

\begin{lemma}\label{lem:good-iterations}
  Fix $\ell \geq 1$. If some quantum interactive machine  $P^*$
  that makes $V$ accept some instance $x$ of size $n$ in $\ell$ sequential
  repetitions of the $\Xi$-protocol with probability at least
  $\eps$, then
  there exists $i \in [\ell]$, such that
the probability that $P^*$ passes game $i$,
  conditioned on the event that $P^*$ passed the games $1,\ldots,
  i-1$, is at least
  $\eps^{-\ell}$.
\end{lemma}
\begin{proof}
  Let us prove this by contradiction. Let $E_j$ be the event that $P^*$ passes
  the game $j$. So let us assume that for all $j$,
  $Pr[E_j| E_1 \wedge \ldots \wedge E_{j-1}] < \eps^{-\ell}$.

  Notice that we can bound the overall acceptance probability as
  \begin{align*}
   \eps &= Pr[E_1 \wedge \ldots \wedge E_\ell] \\
         &= Pr[E_1] Pr[E_2| E_1] \ldots Pr[E_\ell| E_1 \wedge \ldots \wedge E_{\ell-1}] \\
         &< \left( \eps^{-{\ell}} \right)^{\ell} \\
         &= \eps,
  \end{align*}
  which is a contradiction.
\end{proof}

\begin{lemma}\label{lem:sequential}
  Let $A$ and $\{\Pi_c\}$ be defined as in \Cref{fig:extractor-sequential} and
  $Q$ be the verification algorithm for $A$ that consists of
  picking $c \in [m]$ uniformly at random and measuring the
  provided witness with $\Pi_c$.
  Let $\kappa(n) = \left(1 - \frac{1}{2m^2}\right)^{\ell}$ and $K$ be the $\poly(n)$-time extractor defined in
  \Cref{fig:extractor-sequential}. If
  a quantum interactive machine  $P^*$
  makes $V$ accept the instance $x$ of size $n$ with probability at least
  $\eps := (1 - \delta)^\ell > \kappa(n)$, then we have that
  \[Pr\left[ \left(x, K^{\ket{P^*(x,\rho)}}(x)\right) \in R_{Q,1-\delta -
  m^2\delta} \right] \geq \frac{\eps}{\ell}(1 - m^2\delta). \]
  and $K$ runs in time $\poly(n)$.
\end{lemma}
\begin{proof}
  From \Cref{lem:good-iterations}, we know that there exists at least one value
  of $i^* \in [k]$ such that the success of probability in the $i^*$-th round,
  conditioned on the event of success on rounds $1,\ldots,i^*$ is
  $1-\delta$, and we recall that by definition $\delta < \frac{1}{2m^2}$. We have then that the value $i$ guessed by $K$ is equal to
  $i^*$ with probability at least $\frac{1}{\ell}$.

  \medskip

  Let us assume now that $i = i^*$. Using a slightly different notation from
  ~\Cref{lem:single-shot}, let $\mu$ be the initial state of $P^*$ and
  $\hat{\Pi}_j = \Pi_{acc}^{(j)}U_jV_j$ be the verifier operation in round $j$
  $V_j$, followed by the provers' operation operation on round $j$, and finally
  the projection onto the acceptance subspace on the $j$-th round.

  The probability
  that $K$ does not abort in the first $i-1$ steps is
    \begin{align*}
    \Tr\left(
    \hat{\Pi}_{i-1}\ldots \hat{\Pi}_{1}\mu\hat{\Pi}_{1} \ldots
    \hat{\Pi}_{i-1}
    \right) \geq
    \Tr\left(
    \hat{\Pi}_{\ell}\ldots \hat{\Pi}_{1}\mu\hat{\Pi}_{1} \ldots
    \hat{\Pi}_{\ell}
    \right) = \eps,
  \end{align*}
  where the equality holds since we assume that $P^*$ makes the verifier accept
  with probability $\eps$.

  Therefore, we have that with probability $\frac{\eps}{\ell}$, $K$ made the
  correct guess $i = i^*$ and the simulation did not abort during
  the first $i-1$ steps. In this case,  we have that the probability that $P^*$ passes the $i$-th game
  is at least $1-\delta > 1-\frac{1}{2m^2}$ and the result then follows by \Cref{lem:single-shot}.
\end{proof}

\begin{remark}
  We notice that that $P^*$ would need multiple copies of the witness
  in order to pass the sequential repetitions of the $\Xi$-protocol. However,
  the extractor of~\Cref{fig:extractor-sequential} can only extract {\em one} such copy.
  It could be easily extended to output a constant number of
  copies and we leave as an open problem achieving better PoQ extractors for our
  protocol.
\end{remark}

\section{Non-interactive zero-knowledge protocol for $\QMA$ in the secret parameter model}
\label{sec:NIZK}
\label{sec:qniszk}
In this section, using similar techniques of \Cref{sec:xizk-protocol}, we show
that all problems in \SimQMA have a QNISZK protocol in the secret parameter model,
quantizing the result by Pass and Shelat~\cite{PS05}.

We start by defining the model.

\begin{definition}[Quantum non-interactive proofs in the classical secret parameter model]
  A triple of algorithms $(D,P,V)$ is called  a \emph{quantum non-interactive proof in
  the secret parameter model} for a promise problem $A =
  (A_{yes},A_{no})$ where $D$ is a probabilistic polynomial time algorithm, $V$
  is a quantum polynomial time algorithm and $P$ is an unbounded quantum
  algorithm such that there exists a negligible
  function $\eps$ such that the following conditions follow:
  \begin{description}
    \item[] \textbf{Completeness:} for every $x \in A_{yes}$, there exists some $P$
      \[Pr[(r_P, r_V) \leftarrow D(1^{|x|}); \pi \leftarrow P(x,r_P); V(x,r_V,\pi) =
      1] \geq 1 - \eps(n).\]
    \item[] \textbf{Soundness:} for every $x \in A_{no}$ and every $P$
      \[Pr[(r_P, r_V) \leftarrow D(1^{|x|}); \pi \leftarrow P(x,r_P); V(x,r_V,\pi) =
      1] \leq \eps(n).\]
    \item[] \textbf{Statistical zero-knowledge:} for every $x \in A_{yes}$, there is a
      polynomial time algorithm $\mathcal{S}$ such that for the state
      $\sigma = \mathcal{S}(x)$ and $\rho = \sum_{(r_V,s_P) \leftarrow
      D(1^\ell)} p_{r_V,s_P} \kb{r_V} \otimes P(x,r_P)$
      we have that $\sigma \approx_s \rho$.
  \end{description}
\end{definition}

\medskip

The non-interactive protocol is very similar to the $\Xi$ protocol, with a small
(but crucial) change:
instead of using commitments for the one-time pad key, the trusted party picks
these values and reveals just a constant number of these values to the verifier
(and all of them to the prover).
Let us be a bit more precise.
The trusted party picks uniformly at random the one-time pad keys $a$ and $b$
and sends them to the prover. For the verifier, the trusted party sends $S$,
$a|_S$ and $b|_S$, where $S$ is a random subset of $k$ indices of $a$ and $b$.

The prover uses $a$ and $b$ to one-time pad the simulatable proof
and sends this one-time padded state to the verifier.

The verifier picks one of the checking terms uniformly at random.
If the qubits corresponding to the chosen term are not in $S$, the verifier
accepts.
Otherwise, the verifier uses $a|_S$ and $b|_S$ to decrypt the qubits
corresponding to such term, and finally performs the check.
The completeness of the protocol is straightforward. For
soundness, we have that with inverse polynomial probability the revealed bits
will allow the verifier to check the desired term of the encoded history state.
Finally, the zero-knowledge property holds since the quantum proof is simulatable.

We describe now the protocol more formally.

\begin{figure}[H]
\rule[1ex]{\textwidth}{0.5pt}
\definitionsProt
\begin{enumerate}
  \item $D$ picks $a, b \in \{0,1\}^{n}$ and
    $S \subseteq  [n]$, with $|S| = k$, uniformly at random.
  \item $D$ sends $(a,b)$ to the prover and $(S, a|_S, b|_S)$ and to the verifier.

  \item The prover sends $\widetilde{\tau}_{a,b}$ to the verifier
  \item The verifier picks $c \inr [m]$
  \item If the qubits corresponding to $\Pi_c$ are not in $S$, the verifier
    accepts
  \item Otherwise, the verifier applies $\X^{a|_S}\Z^{b|_S}$ to the qubits in
    $\Pi_c$ and accepts according to its output.
\end{enumerate}
\rule[2ex]{\textwidth}{0.5pt}\vspace{-.5cm}
\caption{QNIZK protocol in the secret parameter model  for $\SimQMA$}\label{fig:QNIZK}
\end{figure}

\begin{remark}
  We make the verifier pick one of the terms and then check with the set
  $S$ in order to simplify the proof of soundness. The verifier could instead
  check some $\Pi_c$ that matches with $S$ or accept when this is not possible and
  this protocol would still be sound.
\end{remark}

\begin{lemma}\label{lem:completeness-soundness-qnizk}
  The protocol in \Cref{fig:QNIZK} has completeness $1 - \negl(|x|)$
and soundness $1-\frac{1 - \delta}{n^k}$.
\end{lemma}
\begin{proof}
  If $x \in A_{yes}$, then  the prover sends the honest one-time pad of the
  witness $\tau$ that makes the verifier of the \QMA{} protocol
  accept with  probability exponentially close to $1$.

  In this case, if the qubits of the randomly chosen $\Pi_c$ are not included in
  $S$, the verifier always accepts, otherwise $V$ accepts with probability
  exponentially close to $1$, by the properties of $\tau$.

  Let $x \in \ano$, $(a,b)$ and $(S, a|_S, b|_S)$ be the values sent by the trusted
  party to the prover and verifier, respectively. Let also  $\rho$ be the quantum state sent by the prover and  $\sigma = \X^a\Z^b
  \rho \X^a\Z^b$.

  By definition, the acceptance probability of the protocol is then
  \begin{align}
    &\frac{1}{m} \sum_{S,c} Pr[S_c \not\subseteq S] + Pr[S_c \subseteq S]\tr{\Pi_c
    \X^{a|_{S_c}}\Z^{b|_{S_c}}\rho \Z^{b|_{S_c}}\X^{a|_{S_c}}}
    \\
    &= \left(1-\frac{1}{n^k}\right) +
    \frac{1}{n^k}
    \tr{\frac{1}{m}\sum_{c} \Pi_c \X^{a|_{S_c}}\Z^{b|_{S_c}}\rho \Z^{b|_{S_c}}\X^{a|_{S_c}}}.
  \end{align}

  Notice that
  \[
    \tr{\frac{1}{m}\sum_{c} \Pi_c \X^{a|_{S_c}}\Z^{b|_{S_c}}\rho \Z^{b|_{S_c}}\X^{a|_{S_c}}}   =
  \tr{\frac{1}{m}\sum_{c} \Pi_c \sigma }  \leq
  \max_\tau \tr{\frac{1}{m}\sum_{c} \Pi_c \tau }  \leq  \delta
  \]
  where the first equality holds since $\Pi_c$ is acting only on the decoded
  values and the last inequality holds since $x \in \ano$.

  Therefore, the overall acceptance probability is at most $1 -
  \frac{1-\delta}{n^k}$.
\end{proof}

\begin{figure}[H]
\rule[1ex]{\textwidth}{0.5pt}
\definitionsSim
  \begin{enumerate}
    \item The simulator picks random values $a,b \in \{0,1\}^\ell$ and $S \subseteq [\ell]$
  \item Simulator computes the
    (constant-size) reduced density matrix $\sigma$ of the qubits in positions
      $S$ and let $\sigma(S) = \rho^{\reg{S}} \otimes I^{\reg{\overline{S}}}$
  \item Output $\kb{S} \otimes \kb{a|_S} \otimes \kb{b|_S} \otimes \sigma$.
\end{enumerate}
\rule[2ex]{\textwidth}{0.5pt}\vspace{-.5cm}
\caption{Simulator for the QNIZK protocol}
\end{figure}

\begin{lemma}\label{lem:zk-qnizk}
  The protocol is statistical zero-knowledge.
\end{lemma}
\begin{proof}
  We show that the protocol is statistical zero-knowledge by showing the density
  matrices of the output of the simulator and the real protocol are close.

  In the real protocol, let $\tau$ be the simulatable proof
  in the \QMA{} protocol for a yes-instance.
  In the honest run of the protocol, we have that the view of the verifier after the
  prover sends the message is
  \begin{align}
    \label{eq:view-qnizk}
    \frac{1}{2^{2n}\binom{n}{k}}\sum_{a,b,S} \kb{S}\otimes\kb{a|_S}\otimes\kb{b|_{S}}
  \otimes \tilde{\tau}_{a,b}.
  \end{align}

  Notice that since we are averaging over all possible values of
  $a_{\overline{S}}$ and
  $b_{\overline{S}}$, \Cref{eq:view-qnizk} is equal to
  \[\rho_p = \frac{1}{2^{2|S|}\binom{n}{k}} \sum_{a|_S,b|_S,S}
  \kb{S}\otimes\kb{a|_S}\otimes\kb{b|_{S}}\otimes
  \left(
    \widetilde{\left( \tau_{S} \right)_{a|_S,b|_S}}^{\reg{S}}
  \otimes   I^{\reg{\overline{S}}}
    \right)
  ,\]
  where $\tau_{S} = Tr_{\overline{S}}(\tau)$.

  By definition, the output of the simulator is
  \[\rho_s = \frac{1}{2^{2|S|}\binom{n}{k}}  \sum_{S,a|_S,b|_S}
  \kb{S}\otimes\kb{a|_S}\otimes\kb{b|_{S}} \otimes
  \left(
    \widetilde{\left( \rho_S \right)_{a|_S,b|_S}}^{\reg{S}}
  \otimes   I^{\reg{\overline{S}}}
  \right)
  . \]

  To conclude the proof, we have that
  \begin{align*}
    D(\rho_p, \rho_s) \leq D(\tau_S, \sigma_S) \leq
  \negl(n),
  \end{align*}
  where the first inequality holds since the trace distance
  is subadditive under tensor product and
  preserved under unitary operations.  The second inequality follows from
  \Cref{def:simulatable-proof}.
\end{proof}

\begin{theorem}
  Every problem in \class{QMA} has a QNISZK in the secret parameter model.
\end{theorem}
\begin{proof}
  Direct from \Cref{lem:simulatable-proof,lem:completeness-soundness-qnizk,lem:zk-qnizk}.
\end{proof}

\begin{remark}
  In the cryptography literature, there is a notion called {\em adaptive} soundness and
  zero-knowledge
  where the witness is chosen {\em after} the trusted party provides the secret parameter (or CRS). We notice that our protocols can also handle these stronger notions.
\end{remark}

\subsection{Extension to $\QAM$}
\label{sec:NIZK-QAM}
In \cite{KLGN19}, the authors generalize both the complexity classes
\QMA to allow interaction between the verifier and prover to allow public
randomness, both classical (\emph{i.e.}, classical coins) and  quantum (\emph{i.e.}, sharing
EPR pairs). In this framework, we consider the class $\QAM$, where the
verifier sends random coins to the prover, who then answers  with a quantum
state. We notice that in \cite{KLGN19}, this complexity class is called~$cq\QAM$.

\begin{definition}[\QAM]
A promise problem $A=(\ayes,\ano)$ is in \class{QAM} if and only if there exist
polynomials $r$, $p$, $q$ and a polynomial-time uniform family of quantum circuits
  $\set{Q_{r,n}}$, where $Q_{r,n}$ takes as input a string $x\in\Sigma^*$ with
$\abs{x}=n$, a $p(n)$-qubit quantum state $\ket{\psi}$, and $q(n)$ auxiliary qubits in state~$\ket{0}^{\otimes q(n)}$, such that:
  \begin{description}
    \item[] \textbf{Completeness:}
      { If $x\in\ayes$,\;
      $\Pr_r [ \exists \ket{\psi_r} \text{ s.t. $Q_n$ accepts $(x,\ket{\psi})$}]
      \geq 2/3$.}
    \item[] \textbf{Soundness:} {If $x\in\ano$, \;
      $\Pr_r [ \forall \ket{\psi_r} \text{ $Q_n$ accepts $(x,\ket{\psi})$}]
      \leq 1/3$.}
  \end{description}
\end{definition}

It is straightforward to generalize \Cref{def:simulatable-proof} and define
$\class{SimQAM}$ where  the POVMs and the reduced density matrices depend also
in the public random string $r$. It is not hard to see that we can also
generalize \Cref{lem:simulatable-proof} and show that $\QAM = \class{SimQAM}$.

In this case, our QNIZK protocol can be adapted to this complexity
class, by just making the trusted party pick also the $r$ uniformly at random and
sending it to both the prover and the verifier. Given a fixed $r$, the same arguments
as shown for $\SimQMA$ hold.
\begin{theorem}
  Every problem in \class{QAM} has a QNISZK in the secret parameter model.
\end{theorem}

\bibliographystyle{bst/alphaarxiv}
\bibliography{bib/full,bib/quantum,bib/more}

\end{document}